\documentclass[11pt]{article}

\usepackage[margin=1in]{geometry}
\usepackage{iftex}
\usepackage[T1]{fontenc}
\usepackage{lmodern}
\usepackage{microtype}
\usepackage{amsmath,amssymb,mathtools}
\usepackage{amsthm,thmtools}
\usepackage{etoolbox}

\makeatletter
\@ifundefined{newcounteralias}{}{%
  \renewcommand{\thmt@autorefsetup}{
    \@xa\def\csname\thmt@envname autorefname\@xa\endcsname\@xa{\thmt@thmname}%
  }
}
\makeatother
\usepackage{mathrsfs}
\usepackage{xifthen}
\usepackage{enumitem}
\usepackage{booktabs}
\usepackage[dvipsnames]{xcolor}
\usepackage{tikz}
\usetikzlibrary{arrows.meta}
\usepackage{tcolorbox}
\tcbuselibrary{breakable}
\usepackage[numbers]{natbib}
\usepackage{todonotes}

\usepackage{hyperref}
\hypersetup{
  colorlinks=true,
  pdfpagemode=UseNone,
  citecolor=OliveGreen,
  linkcolor=NavyBlue,
  urlcolor=Magenta,
  pdfstartview=FitW
}
\usepackage[capitalise,nameinlink]{cleveref}

\newcommand{\addsharedenv}[1]{\declaretheorem[sibling=theorem]{#1}}

\theoremstyle{plain}
\declaretheorem{theorem}
\forcsvlist{\addsharedenv}{observation,claim,fact,lemma,proposition,corollary,property}

\theoremstyle{definition}
\forcsvlist{\addsharedenv}{definition,remark,condition,assumption,example}

\crefname{claim}{Claim}{Claims}
\Crefname{claim}{Claim}{Claims}
\crefname{property}{Property}{Property}
\Crefname{property}{Property}{Property}

\crefformat{propositionitem}{#2(#1)#3}
\Crefformat{propositionitem}{#2(#1)#3}
\crefrangeformat{propositionitem}{#3(#1)#4--#5(#2)#6}
\Crefrangeformat{propositionitem}{#3(#1)#4--#5(#2)#6}
\crefmultiformat{propositionitem}{#2(#1)#3}{ and~#2(#1)#3}{, #2(#1)#3}{, and~#2(#1)#3}
\Crefmultiformat{propositionitem}{#2(#1)#3}{ and~#2(#1)#3}{, #2(#1)#3}{, and~#2(#1)#3}

\def\*#1{\mathbf{#1}}
\def\+#1{\mathcal{#1}}
\def\-#1{\mathrm{#1}}
\def\^#1{\mathbb{#1}}
\def\!#1{\mathfrak{#1}}
\def\$#1{\mathscr{#1}}

\newcommand{\tuple}[2][]{\ifthenelse{\isempty{#1}}
  {\left(#2\right)}{\left(#2\right)_{#1}}}
\newcommand{\set}[2][]{\ifthenelse{\isempty{#1}}
  {\left\{#2\right\}}{\left\{#2\right\}_{#1}}}

\DeclareMathOperator{\oPr}{\mathbf{Pr}}
\renewcommand{\Pr}[2][]{\ifthenelse{\isempty{#1}}
  {\oPr\left[#2\right]}{\oPr_{#1}\left[#2\right]}}
\DeclareMathOperator{\oE}{\mathbf{E}}
\newcommand{\E}[2][]{\ifthenelse{\isempty{#1}}
  {\oE\left[#2\right]}{\oE_{#1}\left[#2\right]}}
\DeclareMathOperator{\oVar}{\mathbf{Var}}
\newcommand{\Var}[2][]{\ifthenelse{\isempty{#1}}
  {\oVar\left[#2\right]}{\oVar_{#1}\left[#2\right]}}

\newcommand{\per}[1]{\operatorname{per}\left(#1\right)}
\newcommand{\Cov}{\operatorname{Cov}}
\newcommand{\one}{\mathbf{1}}

\newcommand{\Dir}{\mathcal{E}}
\newcommand{\trel}{\tau_{\mathrm{rel}}}

\newcommand{\tmix}{\tau_{\mathrm{mix}}}
\newcommand{\Tmix}{\tau_{\mathrm{mix}}}
\newcommand{\fpras}{\ensuremath{\mathsf{FPRAS}}}

\newcommand{\e}{\mathrm{e}}
\renewcommand{\epsilon}{\varepsilon}
\renewcommand{\emptyset}{\varnothing}

\newcommand{\supp}{\mathsf{supp}}
\newcommand{\Paths}{\operatorname{Paths}}
\newcommand{\En}{\mathsf{En}}

\title{An $\widetilde{O}(n^3)$-time FPRAS for the Permanent}
\author{
  Xiaoyu Chen\thanks{
    Email: \texttt{xiaoyu@mit.edu},
    Massachusetts Institute of Technology.
    Supported by the NSF CAREER grant CCF-2443045, and the Reed Fund at MIT.
  }
  \and
  Heng Guo\thanks{
    Email: \texttt{hguo@inf.ed.ac.uk},
    School of Informatics, University of Edinburgh,
    Edinburgh, United Kingdom.
  }
  \and
  Eric Vigoda\thanks{
    Emails: \texttt{\{vigoda, xiongxinyang\}@ucsb.edu},
    University of California, Santa Barbara.
    Supported in part by NSF grant CCF-2147094.
  }
  \and
  Xiongxin Yang\footnotemark[3]
}
\date{}

\begin{document}

\maketitle

\begin{abstract}
  We give an FPRAS for the permanent of an $n\times n$ $0/1$ matrix with running time $\widetilde O(n^3)$.
  Our algorithm extends to a strongly polynomial FPRAS for arbitrary nonnegative matrices with the same running-time bound (but in the unit-cost arithmetic model).
  Jerrum, Sinclair, and Vigoda (2004) gave the first FPRAS for the permanent of a nonnegative matrix.
  The running time was subsequently improved to $\widetilde{O}(n^7)$ by Bez\'akov\'a, \v{S}tefankovi\v{c},
  Vazirani, and Vigoda (2008), and recently to $\widetilde{O}(n^6)$ by Chen, Vigoda, and Yang (2026).

  We introduce a multicommodity-flow bound inspired by electrical flows,
  replacing the usual path-length factor by routing energy.
  For the classical JSV chain with perfect-matching weights multiplied by $n^2$,
  we prove a relaxation-time bound of $O(n^3\log n)$ and obtain comparable bounds for estimating the stationary probabilities of having a specified pair of unmatched vertices.
  We use a separate chain on assignments to generate a warm start at any given temperature,
  avoiding the long initialization of the JSV chain.
  Finally, we replace the JSV chain with a new chain that uses only $n$ hole weights and speeds up hole-weight learning by a factor of $n$.
\end{abstract}

\thispagestyle{empty}

\newpage
\pagenumbering{roman}
\setcounter{tocdepth}{2}
\tableofcontents
\newpage

\pagenumbering{arabic}

\section{Introduction}

Evaluation of the permanent of a matrix is a fundamental problem in theoretical computer science.
For an $n\times n$ matrix $A=(a_{ij})$, its permanent is defined as
\begin{align*}
  \per{A}:=\sum_{\sigma\in S_n}\prod_{i=1}^n a_{i,\sigma(i)},
\end{align*}
where $S_n$ is the set of permutations of $\{1,\dots,n\}$.
For the special case of a $0/1$ matrix $A$, the permanent $\per{A}$ is the number of perfect matchings in the bipartite graph $G=(V_1\cup V_2,E)$ where $|V_1|=|V_2|=n$ and $A$ is the bipartite adjacency matrix of $G$.

Valiant~\cite{Valiant79} proved that computing the permanent exactly is \#$\mathsf{P}$-complete,
even for $0/1$ matrices.
This hardness result motivated the search for efficient approximation algorithms for the permanent.
In a seminal work, Jerrum, Sinclair, and Vigoda~\cite{JSV04} developed the first \emph{fully polynomial randomized approximation scheme} ($\fpras$) for the permanent of $0/1$ matrices,
which also extends to arbitrary nonnegative matrices.
Building on their approach, Bez\'akov\'a, \v{S}tefankovi\v{c}, Vazirani,
and Vigoda~\cite{BSVV08} reduced the running time for $0/1$ matrices to $O(n^7\log^4 n)$.
More recently, Chen, Vigoda, and Yang~\cite{CVY26} improved the running time to $\widetilde O(n^6)$.

In this paper, we significantly improve the running time to $\widetilde{O}(n^3)$.

\begin{theorem}
  \label{thm:01-main}
  There is a randomized algorithm that, given an $n\times n$ $0/1$ matrix $A$ and $\varepsilon\in(0,1)$,
  outputs an estimate $\widehat p\geq0$ such that
  \begin{align*}
    \Pr{(1-\varepsilon)\per{A}\leq\widehat p\leq(1+\varepsilon)\per{A}}
    \geq\frac34,
  \end{align*}
  and has running time
  \begin{align*}
    O\left(n^3\varepsilon^{-2}\log^6(n/\varepsilon)\right).
  \end{align*}
\end{theorem}

As in previous works \cite{JSV04,BSVV08,CVY26}, our algorithm for $0/1$ matrices extends to nonnegative matrices.

\begin{corollary}
  \label{cor:weighted-main}
  There is a strongly polynomial randomized algorithm that, given an $n\times n$ nonnegative matrix $A$ and $\varepsilon\in(0,1)$,
  outputs an estimate $\widehat p\geq0$ such that
  \begin{align*}
    \Pr{(1-\varepsilon)\per{A}\leq\widehat p\leq(1+\varepsilon)\per{A}}
    \geq\frac34,
  \end{align*}
  and has running time
  \begin{align*}
    O\left(n^3\varepsilon^{-2}\log^6(n/\varepsilon)\right).
  \end{align*}
\end{corollary}

The running-time bounds in \cref{thm:01-main,cor:weighted-main} count arithmetic operations and comparisons at unit cost,
as in \cite{LSW98,JSV04}.
For nonnegative matrices, the arithmetic-operation bound is independent of the magnitudes and encoding lengths of the entries.
For $0/1$ inputs, the finite-precision implementation in \cref{sub:bit-complexity} uses $O(n^3\varepsilon^{-2}\log^8(n/\varepsilon))$ bit operations,
and hence has bit complexity $\widetilde O(n^3\varepsilon^{-2})$.

\subsection*{High-level algorithmic contributions}
\label{sub:high-level}

Our algorithm builds on the simulated-annealing framework of~\cite{JSV04,BSVV08,CVY26}.
We first recall this framework and its cost per temperature,
and then describe our three contributions: (i) a multicommodity-flow technique that replaces the path-length factor by the routing energy,
(ii) a fast perfect-matching sampler that provides warm starts without waiting for a Markov chain to mix,
and (iii) a new chain, the FAR chain, that learns only $n$ hole weights per temperature.

\paragraph{The annealing framework.}
Let $G=(V_1\cup V_2,E)$ be the bipartite graph of $A$, where $|V_1|=|V_2|=n$,
and let $\mathcal P^*$ be its set of perfect matchings, so that $\per{A}=|\mathcal P^*|$.
We work in the complete bipartite graph $K_{n,n}$ on $(V_1,V_2)$.
Let $\mathcal P$ be its perfect matchings, let $\mathcal N(u,v)$ be its near-perfect matchings with holes (i.e., unmatched vertices) at $u\in V_1$ and $v\in V_2$,
and let $\Omega:=\mathcal P\cup\bigcup_{(u,v)\in V_1\times V_2}\mathcal N(u,v)$.
The central task in estimating $\per{A}$ is to quickly generate approximately uniform samples from $\mathcal P^*$.
As in previous work, we use Markov chains on the enlarged space $\Omega$;
the near-perfect matchings are needed to connect the state space by local moves.

One obstacle is that $G$ may have exponentially more near-perfect than perfect matchings;
hole weights $w(u,v)$ compensate for this by rebalancing the $n^2+1$ blocks $\mathcal P$ and $\mathcal N(u,v)$.
We use simulated annealing on $\Omega$.
Edges of $G$ have activity one, nonedges have activity $\lambda$,
and $\lambda$ decreases from $1$ to at most $1/n!$ along a cooling schedule $\lambda_0>\cdots>\lambda_L$;
let $\mu_i$ be the distribution on $\mathcal P$ proportional to the activities at $\lambda_i$.
At each temperature, the ideal weights $w^*(u,v)=\lambda(\mathcal P)/\lambda(\mathcal N(u,v))$ give all blocks the same total weight.
In our boosted version of the JSV chain (\cref{sec:JSV-chain}),
perfect matchings receive an extra factor $n^2$, so that $\mathcal P$ has stationary probability $1/2$ under the ideal weights and constant probability under accurate estimates;
conditioned on $\mathcal P$, the stationary distribution is exactly $\mu_i$.
The ideal weights depend on unknown partition functions.
Phase~1 learns accurate estimates of them along the schedule, starting from $\lambda_0=1$,
where they are known, and Phase~2 estimates the permanent by a telescoping product of ratio estimators.

\paragraph{Cost per temperature.}
Our schedule has $L=O(n\log^2 n)$ temperatures, and weights learned at one temperature must remain rough at the next,
so the running time is essentially $L$ times the work per temperature.
Inspired by~\cite{GLYYZ26}, we use a non-adaptive schedule and analyze the whole telescoping product at once,
so Phase~2 needs only $O(\varepsilon^{-2})$ independent samples per temperature (\cref{lem:cooling-schedule}).
At each temperature, we must reach approximate stationarity (burn-in),
estimate the block probabilities needed to refine the hole weights (learning),
and produce the Phase-2 samples.
For the JSV chain started from a fixed perfect matching $M^*$ of $G$,
our relaxation-time and occupation bounds (\cref{lem:jsv-relaxation,lem:occupation-estimation}) give the following running time for the warm-up algorithm in \cref{sec:overview},
which already improves~\cite{CVY26} by a factor of $\widetilde O(n)$.
\begin{equation}
  \label{eq:time-JSV}
  {\underbrace{L}_{\mathclap{\widetilde O(n)}}}
  \cdot\Bigl(\,
  {\underbrace{{\overbrace{\trel}^{\mathclap{\widetilde O(n^3)}}}
      \cdot{\overbrace{\log\bigl(1/\pi(M^*)\bigr)}^{O(n\log n)}}}_{\text{burn-in}}}
  +{\underbrace{\widetilde O(n^3)}_{\text{learning}}}
  +{\underbrace{\varepsilon^{-2}\cdot\trel}_{\text{Phase 2}}}
  \,\Bigr)
  =\widetilde O\bigl(n^5+n^4\varepsilon^{-2}\bigr).
\end{equation}
In our final algorithm, the sampler and the FAR chain reduce the burn-in, the learning,
and each Phase-2 sample to $\widetilde O(n^2)$, giving
\begin{equation}
  \label{eq:time-final}
  {\underbrace{L}_{\mathclap{\widetilde O(n)}}}
  \cdot\Bigl(\,
  {\underbrace{\widetilde O(n^2)}_{\substack{\text{burn-in}}}}
  +{\underbrace{\widetilde O(n^2)}_{\substack{\text{learning}}}}
  +{\underbrace{\varepsilon^{-2}\cdot\widetilde O(n^2)}_{\substack{\text{Phase 2}}}}
  \,\Bigr)
  =\widetilde O\bigl(n^3\varepsilon^{-2}\bigr).
\end{equation}
In both bounds, $\widetilde O$ also hides the $O(\log n)$ independent learning batches per temperature.

\paragraph{Routing energy instead of path length.}
Classical multicommodity-flow bounds~\cite{DS91,Sinclair92} control a Poincar\'e inequality by the product of the congestion and the path length,
because Cauchy--Schwarz is applied along each path.
For matchings, switching an augmenting path can take $\Theta(n)$ moves, which costs a factor of $n$.
We instead average the telescoping identities over a randomized routing \emph{before} applying Cauchy--Schwarz (see \cref{eq:routing-energy-bound}).
This replaces the path length by the \emph{routing energy}, a quantity inspired by electrical flows and defined as the sum over transitions of the squared probability that a random path uses the transition (\cref{lem:energy-flow,lem:energy-MCF}).
The energy is at most the expected path length but can be much smaller.
The resulting bound generalizes the classical multicommodity-flow bound of~\cite{Sinclair92} and the recent coupled-flow bound of~\cite{CVY26};
the general technique is presented in \cref{sec:energy-flows-overview}.

To exploit this, consider $I\in\mathcal N(u,v)$ and $F\in\mathcal P$ such that $I\oplus F$ is a single augmenting path between $u$ and $v$.
The standard canonical path switches this path from one end.
Starting at $u$, it moves the hole by successive slides, each replacing an edge of $I$ with an adjacent edge of $F$,
and finally adds the remaining edge of $F$.
We instead switch the path from \emph{both} ends, choosing the next end to advance by a P\'olya-urn rule,
so that after $t$ slides all $t+1$ splits between the two ends are equally likely.
Each matching reachable after $t$ slides then has probability $1/(t+1)$,
so each of the at most $2(t+1)$ transitions leaving these matchings is used with probability at most $1/(t+1)$;
since every path uses exactly one of them, these transitions contribute at most $1/(t+1)$ to the energy.
Summing over $t$, including the final edge additions, gives energy $O(\log n)$,
although the paths can have length $\Theta(n)$ (\cref{lem:two-ended-energy-overview}; see \cref{sub:routing-energy-idea} for details).
Alternating cycles are switched with constant energy by randomizing the starting edge (\cref{lem:cycle-energy-overview}).

In general, $I\oplus F$ also contains alternating cycles.
Following~\cite{CVY26}, we route from $I$ only to the perfect matching obtained by switching the augmenting path in $I\oplus F$.
If $I$ and $F$ are drawn independently from the stationary distribution conditioned on $\mathcal N(u,v)$ and on $\mathcal P$,
respectively, this endpoint is still distributed as $F$ (\cref{lem:augmenting-path-coupling}),
so the alternating cycles drop out of the comparison between a near-perfect block and $\mathcal P$.
Together with the cycle routing, which controls the variation within $\mathcal P$ (\cref{lem:perfect-matching-comparison}),
this gives relaxation time $O(n^3\log n)$ for the JSV chain (\cref{lem:jsv-relaxation}),
improving~\cite{BSVV08} by a factor of $\widetilde O(n)$.
The same flows show that, started at stationarity, $\widetilde O(n^3)$ steps of the JSV chain estimate each of the $n^2+1$ block probabilities to constant relative error (\cref{lem:occupation-estimation});
these estimates refine rough hole weights to the accuracy needed at the next temperature.

\paragraph{Warm starts from a separate perfect-matching sampler.}
With the bounds above, the burn-in term is the bottleneck in \cref{eq:time-JSV},
because $\pi(M^*)$ can be as small as $\Theta(1/n!)$.
We bypass mixing altogether.
At a fixed temperature $\lambda_i$, given rough estimates of the ideal hole weights,
we sample a perfect matching from approximately $\mu_i$ in time $\widetilde O(n^2)$ (\cref{lem:fast-burnin}).
Because $\mathcal P$ has constant stationary probability, such a sample is a warm start for the learning chain,
and independent samples also serve directly as the Phase-2 samples,
at cost $\widetilde O(n^2)$ instead of $\trel$ each.
The sampler needs only the $n$ \emph{anchor} weights $w(r,\cdot)$ for one arbitrarily chosen $r\in V_1$,
and it has four ingredients.
\begin{itemize}
  \item \emph{One-sided rescaling.} Dividing column $v$ by $w(r,v)$ and normalizing the rows gives a row-stochastic matrix $B$ with the same distribution on perfect matchings.
    Under $B$, every near-perfect block has weight $O(B(\mathcal P))$,
    and every anchor block $\mathcal N(r,v)$ has weight $\Omega(B(\mathcal P))$ (\cref{lem:burnin-balance}).
  \item \emph{Relaxing the matching constraint.} We sample assignments $\sigma:V_1\to V_2\cup\{\ast\}$ in which several left vertices may share a right vertex,
    at the price of a factor $(1+J)^{-1}$, $J=\Theta(n^2)$, per extra assignment.
    Keeping $\lambda_i$ fixed, we raise an assignment activity $a$ from $1/n$ to $J$ in $O(n)$ stages,
    each providing a warm start for the next;
    at $a=J$, perfect matchings have probability at least $1/2$ and conditional distribution exactly $\mu_i$ (\cref{lem:burnin-schedule}).
  \item \emph{Spectral gap from a Bochner identity.} At each stage we run a continuous-time dynamics that resamples one coordinate at a time,
    at rate $1+J$ for vertices sharing their right vertex and rate $1$ otherwise.
    We prove a Bochner-type identity for multi-particle resampling chains,
    which gives spectral gap at least $1$ whenever the rates are nondecreasing in the occupancy (\cref{lem:particle-gap,cor:occupancy-particle-gap}).
  \item \emph{Cost from real-rootedness.} The number of resampling events in a stage is controlled by $(1+a)$ times the number of unassigned vertices,
    which stays $O(1+\min\{n,n/\sqrt a\})$ up to logarithmic terms.
    We prove this bound using the Heilmann--Lieb theorem~\cite{HL72} and a mixed-derivative inequality for hyperbolic polynomials (\cref{prop:burnin-factorization,lem:burnin-mixed-derivative,lem:burnin-matching-estimates}).
    Summed over the stages, this gives $O(n^2)$ resampling events, up to logarithmic factors.
\end{itemize}

\paragraph{Learning only $n$ hole weights with the FAR chain.}
Since the sampler uses only the anchor weights, it suffices to learn $n$ instead of $n^2$ hole weights at each temperature.
For this purpose, and inspired by the sampler, we design the \emph{fixed anchor reweighting} (FAR) chain on $\Omega$ (\cref{sec:algorithm-FAR}).
It is obtained from a continuous-time chain defined from $B$,
in which left slides are weighted by the column sums of $B$,
so that the transition capacities match the loads of the two-ended routing.
We discretize this chain by normalizing the rates locally at each state rather than by a global bound,
which avoids $\Theta(n)$ idle steps at perfect matchings (\cref{rem:FAR-continuous});
the resulting change in the stationary weights is the explicit factor $W(u,v)$,
which is divided out when the weights are updated.
For the anchor blocks $\mathcal N(r,v)$, the energy-weighted flows give occupation estimates after $O(n^2\log n)$ steps (\cref{lem:FAR-occupation-estimation}),
a factor of $n$ fewer than for the JSV chain.
These estimates require only a warm start, which the sampler provides,
so we never need to bound the mixing time of the FAR chain.
The anchor weights are then normalized using the expansion of the permanent along row $r$ (\cref{eq:far-weight-update}).
Thus learning costs $\widetilde O(n^2)$ per temperature, as in \cref{eq:time-final}.

\paragraph{Organization.}
In \cref{sec:overview} we define the boosted JSV chain and present the two-phase framework,
obtaining the bound $\widetilde O(n^5+n^4\varepsilon^{-2})$.
In \cref{sec:burn-in} we develop the fast perfect-matching sampler and use it to initialize FAR weight learning in \cref{sec:algorithm-FAR}.
Once the weights are learned, independent calls to the sampler supply products of estimates across temperatures,
giving total time $\widetilde O(n^3\varepsilon^{-2})$.

In \cref{sec:energy-flows-overview} we develop energy-weighted multicommodity flows,
and \cref{sec:JSV-flow} proves the associated flow bounds for the JSV and FAR chains.
\Cref{sec:fast-sampling-proofs} gives the detailed analysis of sampling at a fixed temperature.
\Cref{sec:missing} treats the shared temperature-cooling estimates, the cubic-time approximation guarantee,
the extension to nonnegative matrices, and the bit complexity.

\section{Markov chain preliminaries}
\label{sec:preliminaries}

Let $P$ be the transition matrix of a finite irreducible Markov chain on $\Omega$,
and let $\pi$ denote its unique stationary distribution.
Throughout this section, we assume that $P$ is reversible and lazy, meaning that, for every $x,y\in\Omega$,
\begin{align*}
  \pi(x)P(x,y)=\pi(y)P(y,x),
  \qquad
  P(x,x)\geq\frac12.
\end{align*}

For a function $f:\Omega\to\mathbb R$, write $\E[\pi]{f}:=\sum_{x\in\Omega}\pi(x)f(x)$.
Its Dirichlet form and variance are
\begin{align*}
  \Dir_P(f,f)
  &:=
  \frac12\sum_{x,y\in\Omega}
  \pi(x)P(x,y)(f(x)-f(y))^2,
  \\
  \Var[\pi]{f}
  &:=
  \E[\pi]{f^2}-\E[\pi]{f}^2
  =
  \frac12\sum_{x,y\in\Omega}
  \pi(x)\pi(y)(f(x)-f(y))^2.
\end{align*}
The \emph{relaxation time} is the inverse of the spectral gap.
Equivalently, $\trel(P)$ is the smallest constant $C>0$ such that every $f:\Omega\to\mathbb R$ satisfies the Poincar\'e inequality
\begin{align}
  \Var[\pi]{f}\leq C\Dir_P(f,f).
  \label{eq:poincare-characterization}
\end{align}

For a distribution $\mu$ on $\Omega$, its \emph{total variation distance} to $\pi$ and \emph{$\chi^2$-divergence} relative to $\pi$ are
\begin{align*}
  \|\mu-\pi\|_{\mathrm{TV}}
  :=\frac12\sum_{x\in\Omega}|\mu(x)-\pi(x)|,\qquad
  \chi^2(\mu\Vert\pi)
  :=\sum_{x\in\Omega}\frac{(\mu(x)-\pi(x))^2}{\pi(x)}.
\end{align*}
For $0<\eta<1$, the \emph{mixing time from an initial distribution $\mu$} is
\begin{align*}
  \tmix(P,\eta;\mu)
  :=\min\left\{
  t\in\mathbb Z_{\geq0}:
  \|\mu P^t-\pi\|_{\mathrm{TV}}\leq\eta
  \right\}.
\end{align*}
When $\mu=\delta_x$ is the point mass at $x$, we write $\tmix(P,\eta;x)$.
The worst-case mixing time is $\tmix(P,\eta):=\max_{x\in\Omega}\tmix(P,\eta;x)$.
When the error parameter is omitted, we take $\eta=1/4$.
The standard bound in~\cite[Chapter~12]{LPW17} gives
\begin{align}
  \tmix(P,\eta;x)
  \leq
  1+\trel(P)\log\frac1{\eta\pi(x)}.
  \label{eq:mixing-from-relaxation}
\end{align}

For a general initial distribution, the variance-contraction inequality in~\cite[Section~12.2, equation~(12.8)]{LPW17},
applied to $\mu/\pi$, together with Cauchy--Schwarz gives
\begin{align*}
  \|\mu P^t-\pi\|_{\mathrm{TV}}
  \leq
  \frac12\sqrt{\chi^2(\mu P^t\Vert\pi)}
  \leq
  \frac12 e^{-t/\trel(P)}\sqrt{\chi^2(\mu\Vert\pi)}.
\end{align*}
Consequently, if $\chi^2(\mu\Vert\pi)\leq D$ for some $D\geq0$,
then $\tmix(P,\eta;\mu)\leq1+\trel(P)\log((1+D)/\eta)$.
We call $\mu$ a \emph{warm start} for $\pi$ in the $\chi^2$ sense if $\chi^2(\mu\Vert\pi)=O(1)$,
in which case $\tmix(P;\mu)=O(\trel(P))$.
For any set $S\subseteq\Omega$ with $\pi(S)>0$, we have $1+\chi^2(\pi(\,\cdot\mid S)\Vert\pi)=1/\pi(S)$,
so conditioning on a set whose stationary probability is bounded below by a positive constant gives a warm start.

\section{\texorpdfstring{$\widetilde{O}(n^5)$}{Õ(n⁵)}-time algorithm: The two-phase framework}
\label{sec:overview}
We begin with a $\widetilde O(n^5)$-time \fpras{}, which already improves the previously best known result of \cite{CVY26} by an $\widetilde{O}(n)$-factor.
More importantly, it sets up the basic algorithmic framework for our algorithms.
Subsequently, we build upon this same framework to obtain further improvements.

Let $A$ be an $n\times n$ $0/1$ matrix and let $G=(V_1\cup V_2,E)$ be its bipartite graph.
Recall that our goal is to estimate $\per{A}=|\mathcal P^*|$.
We first find a perfect matching $M^*$;
if none exist then we output $\per{A}=0$.
The arbitrarily chosen matching $M^*$ is fixed throughout the algorithm.
From now on, we work in the complete bipartite graph on $(V_1,V_2)$, with the perfect matchings $\mathcal P$,
the near-perfect matchings $\mathcal N(u,v)$, and the state space $\Omega$ defined in \cref{sub:high-level}.

For a \emph{temperature} parameter $\lambda\in[0,1]$, assign each edge $e\in V_1\times V_2$ of the complete bipartite graph the activity
\begin{align*}
  \lambda(e):=
  \begin{cases}
    1,&e\in E,\\
    \lambda,&e\notin E.
  \end{cases}
\end{align*}
We extend this notation to the \emph{activity} of a matching $M\in\Omega$ and the \emph{partition function} of a family $S\subseteq\Omega$ by setting
\begin{align*}
  \lambda(M):=\prod_{e\in M}\lambda(e)=\lambda^{|M\setminus E|}, \qquad \lambda(S):=\sum_{M\in S}\lambda(M),
\end{align*}
with the convention $0^0=1$.
Thus each nonedge contributes a factor $\lambda$, while edges of $G$ have activity one;
in particular, $\lambda(M^*)=1$ at every temperature.

At $\lambda=1$, every matching has activity one, so $\lambda(\mathcal{P})=n!$ and $\lambda(\mathcal{N}(u,v))=(n-1)!$ for every $(u,v)\in V_1\times V_2$.
At $\lambda=0$, only matchings contained in $E$ have nonzero activity,
and hence $\lambda(\mathcal P)=|\mathcal P^*|=\per{A}$.
This suggests a \emph{simulated annealing} approach.
We start from the known partition functions at temperature one and gradually decrease the activity of the nonedges.

Instead of reaching zero, we stop at a sufficiently small positive temperature.
A \emph{cooling schedule} is a sequence $\lambda_0=1>\lambda_1>\cdots>\lambda_L>0$ with $\lambda_L\leq1/n!$.
Every matching in $\mathcal P^*$ contributes one to $\lambda_L(\mathcal P)$,
while every matching in $\mathcal P\setminus\mathcal P^*$ contains a nonedge and therefore contributes at most $\lambda_L$.
Hence
\begin{align*}
  \per{A}\leq\lambda_L(\mathcal P) \leq\per{A}+n!\lambda_L \leq2\per{A}.
\end{align*}

For $0\leq i\leq L$, write $Z_i:=\lambda_i(\mathcal P)$, and let $\mu_i$ be the distribution on $\Omega$ given by $\mu_i(M):=\lambda_i(M)/Z_i$ for $M\in\mathcal P$ and $\mu_i(M):=0$ otherwise.
Samples from $\mu_i$ allow us to estimate the ratio of consecutive partition functions ${Z_{i+1}}/{Z_i}$.
The product of these ratios is $Z_L/Z_0$, where $Z_0=n!$.
To recover $\per{A}$ from $Z_L$, we also estimate the probability $p^*:=\Pr[M\sim\mu_L]{M\in\mathcal P^*}\ge 1/2$ that a sample at the final temperature uses only input edges.

The algorithm implements this reduction in two phases.
In Phase~1, we learn the hole weights needed to sample approximately from each $\mu_i$ along a suitable cooling schedule.
In Phase~2, we keep the learned weights fixed and use independent randomness to estimate the telescoping product and the final-temperature probability $p^*$.
Building on~\cite{JSV04,BSVV08,CVY26}, we first implement this framework with the JSV chain defined below.
At temperature $\lambda_i$, this chain runs on the enlarged state space $\Omega$,
and its stationary distribution conditioned on $\mathcal P$ is exactly $\mu_i$.

\subsection{The JSV chain}
\label{sec:JSV-chain}

Throughout this subsection, all edge activities are positive.
To compensate for the potentially unequal total activities of the near-perfect blocks, Jerrum, Sinclair,
and Vigoda \cite{JSV04} introduced a positive \emph{hole weight} $w(u,v)$ for each $(u,v)\in V_1\times V_2$.
Define the unnormalized weights
\begin{align*}
  \omega_{\lambda,w}(M):=
  \begin{cases}
    n^2\lambda(M), & M\in\mathcal P,\\
    w(u,v)\lambda(M), & M\in\mathcal N(u,v),
  \end{cases}
\end{align*}
and let $\pi_{\lambda,w}$ be the corresponding probability distribution on $\Omega$.
The factor $w(u,v)$ adjusts the total mass of $\mathcal N(u,v)$ without changing the relative probabilities of matchings within that block.

The \emph{JSV chain}, denoted by $P_{\mathrm{JSV}}(\lambda,w)$, evolves from a matching $M$ as follows.
With probability $1/2$, stay at $M$ and end the transition.
Otherwise, choose $z$ uniformly at random from $V_1\cup V_2$ and perform the following:
\begin{enumerate}
  \item If $M\in\mathcal P$, then:
    \begin{enumerate}
      \item if $z\in V_1$, let $(z,y)$ be the edge of $M$ incident to $z$,
        and set $M'=M\setminus\{(z,y)\}\in \mathcal N(z,y)$;
      \item otherwise, set $M'=M$.
    \end{enumerate}

  \item If $M\in\mathcal N(u,v)$, then:
    \begin{enumerate}
      \item if $z=u$, set $M'=M\cup\{(u,v)\}$;
      \item if $z\in V_2\setminus\{v\}$, let $(x,z)\in M$ and set $M'=\bigl(M\setminus\{(x,z)\}\bigr)\cup\{(u,z)\}\in \mathcal N(x,v)$;
      \item if $z\in V_1\setminus\{u\}$, let $(z,y)\in M$ and set $M'=\bigl(M\setminus\{(z,y)\}\bigr)\cup\{(z,v)\}\in \mathcal N(u,y)$;
      \item otherwise, set $M'=M$.
    \end{enumerate}

  \item With probability $\min\{1,\omega_{\lambda,w}(M')/\omega_{\lambda,w}(M)\}$, move to $M'$;
    otherwise, stay at $M$.
\end{enumerate}

Write $M\sim M'$ if $M\ne M'$ and one can be obtained from the other by an add/delete or slide move above.
Conditional on passing the initial holding step, each such neighbor is proposed with probability $1/(2n)$.
Thus, for fixed $\lambda,w$ and every pair of adjacent states $M\sim M'$,
\begin{align*}
  P_{\mathrm{JSV}}(M,M')
  =
  \frac{1}{4n}
  \min\left\{
  1,
  \frac{\omega_{\lambda,w}(M')}{\omega_{\lambda,w}(M)}
  \right\}.
\end{align*}
Distinct nonadjacent states have transition probability zero,
and the remaining probability is placed on the self-loop.
The chain is lazy because each state has at most $2n$ neighbors,
and the Metropolis rule makes it reversible with respect to $\pi_{\lambda,w}$.
The local moves connect $\Omega$, so positivity of the activities and hole weights also gives irreducibility.

\begin{remark}
  \label{rem:boosted-JSV}
  This is a boosted version of the original JSV chain~\cite{JSV04}.
  Each perfect matching receives an additional factor $n^2$ in $\omega_{\lambda,w}$,
  while the near-perfect weights and local moves are unchanged from the original JSV chain.
  The $n^2$-factor gives the perfect-matching block constant stationary mass;
  this is important in the initialization algorithm presented in \cref{sec:burn-in}.
  Throughout this paper, we refer to this boosted version simply as the JSV chain.
\end{remark}

The \emph{ideal hole weights} are
\begin{align}
  w^*(u,v):=
  \frac{\lambda(\mathcal P)}{\lambda(\mathcal N(u,v))}.
  \label{eq:ideal-weights}
\end{align}
With these weights, each near-perfect block has total unnormalized weight $\lambda(\mathcal P)$,
while the perfect-matching block has total weight $n^2\lambda(\mathcal P)$.
Hence
\begin{align*}
  \pi_{\lambda,w^*}(\mathcal P)=\frac12,
  \qquad
  \pi_{\lambda,w^*}(\mathcal N(u,v))=\frac1{2n^2}.
\end{align*}
The ideal weights depend on unknown partition functions and are therefore not available to the algorithm.
We work instead with approximations, whose accuracy is quantified as follows.

\begin{definition}
  \label{def:rough-hole-weights}
  For $c\geq1$, the hole weights $w$ are \emph{$c$-rough} if, for every $(u,v)\in V_1\times V_2$,
  \begin{align*}
    \frac{1}{c}\leq\frac{w(u,v)}{w^*(u,v)}\leq c.
  \end{align*}
  In particular, they are \emph{rough} when $c=2$, and \emph{accurate} when $c=\sqrt2$.
\end{definition}

For any positive hole weights $w$, the conditional distribution on perfect matchings does not depend on $w$.
Specifically, for every $M\in\mathcal P$,
\begin{align*}
  \frac{\pi_{\lambda,w}(M)}{\pi_{\lambda,w}(\mathcal P)}
  =\frac{\lambda(M)}{\lambda(\mathcal P)}.
\end{align*}
Thus, at temperature $\lambda_i$, conditioning $\pi_{\lambda_i,w}$ on $\mathcal P$ gives exactly $\mu_i$.

\subsection{Phase~1: learning hole weights}
\label{subsec:phase-one}

Phase~1 learns accurate hole weights at every temperature of a cooling schedule;
write $w_i^*$ for the ideal hole weights at temperature $\lambda_i$.
It combines three ingredients: an update that turns rough weights into accurate ones at a fixed temperature,
an efficient way to estimate the probabilities that this update needs,
and a cooling schedule along which the ideal weights change slowly.

\paragraph{Refining the weights.}
Fix positive edge activities $\lambda$ and positive hole weights $w$, and write $\pi=\pi_{\lambda,w}$.
Let $\mathcal H:=\{\mathcal P\}\cup\{\mathcal N(u,v):u\in V_1,\ v\in V_2\}$ be the partition of $\Omega$ into blocks,
write $p_S:=\pi(S)$ for $S\in\mathcal H$, and abbreviate $p_\emptyset:=p_{\mathcal P}$ and $p_{u,v}:=p_{\mathcal N(u,v)}$.
By the definition of the stationary distribution,
\begin{align}
  w^*(u,v)=\frac{\lambda(\mathcal P)}{\lambda(\mathcal N(u,v))}=\frac{w(u,v)}{n^2}\frac{p_\emptyset}{p_{u,v}}.
  \label{eq:hole-weights-identity}
\end{align}

Given estimates $\widehat p_S$ of all $p_S$, we use \cref{eq:hole-weights-identity} to define
\begin{align}
  w'(u,v):=\frac{w(u,v)}{n^2}\frac{\widehat p_\emptyset}{\widehat p_{u,v}}.
  \label{eq:jsv-weight-update}
\end{align}
If the estimates $\widehat p_S$ have relative error at most $1/10$, then the updated weights $w'(u,v)$ satisfy
\begin{align*}
  \frac{w'(u,v)}{w^*(u,v)}=\frac{\widehat p_\emptyset/p_\emptyset}{\widehat p_{u,v}/p_{u,v}}
  \in\left[\frac9{11},\frac{11}{9}\right]\subseteq\left[\frac1{\sqrt2},\sqrt2\right].
\end{align*}
Thus, the updated weights $w'(u,v)$ are accurate in the sense of \cref{def:rough-hole-weights},
and their relative errors depend only on the frequency estimates,
not on the accuracy of the original weights $w(u,v)$.

\paragraph{Estimating the block probabilities.}
We estimate each $p_S$ by its empirical frequency along a trajectory of the JSV chain;
to start near stationarity, the chain first runs for a \emph{burn-in} period from the fixed perfect matching $M^*$.
Two lemmas determine how long each stage must be.
\Cref{lem:occupation-estimation} bounds the variance of the empirical frequencies,
and \cref{lem:jsv-relaxation} bounds the relaxation time, which controls the burn-in through \cref{eq:mixing-from-relaxation}.

\begin{lemma}
  \label{lem:occupation-estimation}
  Let $(X_t)_{t\geq0}$ evolve according to $P_{\mathrm{JSV}}(\lambda,w)$ with rough hole weights $w$ and $X_0\sim\pi$.
  There is a universal constant $C>0$ such that, for every integer $T\geq1$ and every $S\in\mathcal H$,
  \begin{align}
    \Var{
      \frac1T\sum_{t=0}^{T-1}\one_{\{X_t\in S\}}
    }
    \leq
    \frac{Cn^3\log n}{T}\,p_S^2.
    \label{eq:JSV-occupation-estimation}
  \end{align}
\end{lemma}

\begin{lemma}
  \label{lem:jsv-relaxation}
  For every fixed $\alpha\geq1$, there is a constant $C_\alpha$ such that,
  for every choice of positive edge activities and every set of $\alpha$-rough hole weights,
  \begin{align*}
    \trel\bigl(P_{\mathrm{JSV}}(\lambda,w)\bigr)
    \leq C_\alpha n^3\log n.
  \end{align*}
\end{lemma}

For every temperature $\lambda_i$ and all rough weights $w$, we have $\pi_{\lambda_i,w}(M^*)\geq1/(3n!)$.
Hence, by \cref{lem:jsv-relaxation,eq:mixing-from-relaxation},
a burn-in of $O(n^4\log^2 n)$ transitions brings the chain within a small constant total variation distance of stationarity.
A \emph{learning batch} consists of such a burn-in followed by $T:=\lceil C_0n^3\log n\rceil$ \emph{learning steps},
along which the frequencies of all blocks are recorded.
By \cref{lem:occupation-estimation} and Chebyshev's inequality,
if $C_0$ is a sufficiently large universal constant, then each batch estimates each fixed $p_S$ to relative error at most $1/10$ with probability at least $7/8$,
even after accounting for the burn-in error.

Both lemmas rest on a refined multicommodity-flow analysis, developed in \cref{sec:energy-flows-overview}.
Applying Cauchy--Schwarz after averaging over the paths of a randomized routing, instead of path by path,
replaces the usual path-length factor by the \emph{routing energy},
which measures how the routing spreads its flow over the transitions.
Our routings have energy $O(\log n)$ even though their paths have length $\Theta(n)$,
which improves the previous flow-analysis bound of~\cite{BSVV08} by a factor of $\widetilde O(n)$.

\paragraph{The cooling schedule.}
Inspired by~\cite{GLYYZ26}, we use a non-adaptive schedule and analyze the whole product of the partition ratios at once.
Since the expected number of nonedges decreases during cooling,
the temperature steps can grow along the schedule while each relevant partition function changes by at most a fixed constant factor between consecutive temperatures.
The schedule also bounds the relative second moment of the product.
Recall that, for $\lambda_i>0$, $Z_i:=\lambda_i(\mathcal P)$ is the associated partition function,
and $\mu_i$ is the corresponding distribution on $\mathcal{P}$.

\begin{lemma}
  \label{lem:cooling-schedule}
  There is an explicit schedule
  \begin{align}
    1=\lambda_0>\lambda_1>\cdots>\lambda_L
    =
    \exp\bigl(-\lceil\log_2(n!)\rceil\bigr)
    \leq\frac1{n!},
    \qquad
    L=O(n\log^2 n),
    \label{schedule}
  \end{align}
  such that, for every $0\leq i<L$ and every $(u,v)\in V_1\times V_2$, the ideal hole weights satisfy
  \begin{align}
    \frac34
    \leq
    \frac{w_{i+1}^*(u,v)}{w_i^*(u,v)}
    \leq
    \frac43.
    \label{eq:cooling-ratio}
  \end{align}
  If $M_i\sim\mu_i$ independently for $0\leq i<L$, then $Y:=\prod_{i=0}^{L-1}\lambda_{i+1}(M_i)/\lambda_i(M_i)$ satisfies
  \begin{align}
    \E{Y}=\frac{Z_L}{Z_0},
    \qquad
    \frac{\E{Y^2}}{(\E{Y})^2}\leq\frac{64}{27}.
    \label{eq:cooling-variance}
  \end{align}
\end{lemma}

We prove \cref{lem:cooling-schedule} in \cref{sub:cooling-proofs};
the second-moment bound \cref{eq:cooling-variance} is used in Phase~2.
In Phase~1, \cref{eq:cooling-ratio} ensures that weights accurate at $\lambda_i$ remain rough at $\lambda_{i+1}$,
so they can be used to run the JSV chain at the next temperature and can be refined there.

\paragraph{The learning procedure.}
We start from the known ideal weights $w_0(u,v)=n$ at $\lambda_0=1$.
For $i=0,1,\dots,L-1$, we run $R=O(\log{n})$ independent learning batches at temperature $\lambda_{i+1}$ with the weights $w_i$.
In each batch, we run the JSV chain, starting from $M^*$, for its mixing time of $O(n^4\log^2{n})$ to obtain an initial sample.
Then we run the JSV chain for $O(n^3\log{n})$ steps to estimate the stationary probabilities of having holes $(u,v)$,
using \cref{lem:occupation-estimation}.
Applying \cref{eq:jsv-weight-update}, we obtain accurate hole weights at $\lambda_{i+1}$,
and we proceed to the next $i$.

Phase~1 uses $O(LR\,n^4\log^2 n)=O(n^5\log^5 n)$ transitions in total, which is dominated by the burn-ins,
since the learning steps contribute only $O(LRn^3\log{n})=\widetilde O(n^4)$ transitions;
this motivates the faster initialization procedure of \cref{sec:burn-in} to bypass the mixing time required for the burn-in.

\subsection{Phase~2: averaging products across temperatures}
\label{subsec:phase-two}

Assume we successfully completed Phase~1 and obtained accurate hole weights $w_i(u,v)$ for all $u\in V_1, v\in V_2$ at every temperature $\lambda_i, 0\leq i\leq L$ in the cooling schedule of \cref{lem:cooling-schedule}.
The basic observation is that a sample from $\mu_i$, which is the distribution prescribed by $\lambda_i$ on $\mathcal{P}$,
gives an unbiased estimate of the ratio of the perfect-matching partition functions, since
\begin{align*}
  \E[M\sim\mu_i]{\frac{\lambda_{i+1}(M)}{\lambda_i(M)}}
  =\sum_{M\in\mathcal P}\frac{\lambda_i(M)}{Z_i}
  \frac{\lambda_{i+1}(M)}{\lambda_i(M)}
  =\frac{Z_{i+1}}{Z_i}.
\end{align*}
For independent samples across temperatures, the expectation of their product therefore telescopes to $Z_L/Z_0$.
For each repetition $t$, take one perfect matching $M_{i,t}$ at each temperature $\lambda_i$ and form the product of the ratio estimates and the indicator that the final sample uses only input edges,
obtaining
\begin{align*}
  X_t:=n!\,\one_{\{M_{L,t}\in\mathcal P^*\}}
  \prod_{i=0}^{L-1}
  \frac{\lambda_{i+1}(M_{i,t})}{\lambda_i(M_{i,t})},
\end{align*}
and return the average $\widehat p$ of these products.
We call this the whole-product estimator.
The indicator that $M_{L,t}\in\mathcal P^*$ supplies the factor $p^*$ in $\per{A}=n!\,(Z_L/Z_0)\,p^*$.
In the ideal stationary setting, $\E{X_t}=\per{A}$, while \cref{eq:cooling-variance} and $p^*\geq1/2$ give constant relative variance,
independently of the number of temperatures.

To generate the samples, we run one independent JSV chain at each temperature.
We initialize approximately from $\mu_i$ using burn-ins of $\Tmix=O(n^4\log^2 n)$ transitions from $M^*$,
repeated independently until the endpoint is a perfect matching.
We then inspect each chain every $s=\Theta(n^3\log n)$ transitions,
retaining only the perfect matchings observed at these inspection times,
until we have collected $N=\lceil C\varepsilon^{-2}\rceil$ samples at every temperature,
for a sufficiently large constant $C$.
Successive samples from one chain may be correlated, while chains at different temperatures are independent.
The relaxation bound in \cref{lem:jsv-relaxation} controls the variance of the average of their products,
so Chebyshev's inequality gives the desired relative accuracy.

Controlling the initialization error and stopping with output zero if the total number of transitions reaches a prescribed limit gives overall success probability at least $3/4$,
including the learning and implementation errors.
Since perfect matchings have stationary probability at least $1/3$,
Phase~2 uses $O((L+1)(\Tmix+s\varepsilon^{-2}))$ transitions, and the total, including learning, is
\begin{align*}
  O\left(n^5\log^5 n+n^4\log^3 n\,\varepsilon^{-2}\right)
  =\widetilde O(n^5+n^4\varepsilon^{-2}).
\end{align*}

\section{\texorpdfstring{$\widetilde{O}(n^3)$}{Õ(n³)}-time algorithm:  Fast sampling and the FAR chain}
\label{sec:burn-in}

Phase~1 of the \fpras{}  learns accurate hole weights across the $L$ intermediate temperatures of the cooling schedule prescribed by \cref{lem:cooling-schedule}.
As shown in \cref{subsec:phase-one}, the primary computational bottleneck in Phase~1 lies in initializing the chain at each temperature,
which requires running the JSV chain for its mixing time of $O(n^4\log^2{n})$ steps.

To overcome this bottleneck within the same framework, a natural direction is to design a Markov chain with a sharper variance bound and a shorter relaxation time (cf.\ \cref{lem:occupation-estimation,lem:jsv-relaxation}).
In \cref{sec:algorithm-FAR}, we introduce the \emph{fixed anchor reweighting} (FAR) chain,
which improves both parameters by a factor of $\widetilde{O}(n)$ (see \cref{lem:FAR-occupation-estimation} for an example).
Simply substituting the FAR chain for the JSV chain in the framework of \cref{sec:overview} immediately gives an $\widetilde{O}(n^4)$-time \fpras{}.

In this section, we go a step further and shave another factor of $\widetilde{O}(n)$ by reducing the initialization overhead at each temperature from $\widetilde{O}(n^4)$ to $\widetilde{O}(n^2)$.
Specifically, at temperature $\lambda_i$, our new procedure generates a perfect matching from a distribution that is close in total variation distance to $\mu_i$,
the stationary distribution of the JSV chain conditioned on~$\mathcal{P}$.
As a byproduct, this procedure can also be used directly as the sampler in Phase~2.
Notably, the procedure relies on only $n$ rough hole weights $\{w_i(r,v)\}_{v\in V_2}$ incident to a single fixed \emph{anchor} vertex $r\in V_1$.
Combining this efficient sampler with the FAR chain ultimately leads to our $\widetilde{O}(n^3)$-time \fpras.

The guarantees of this sampler are as follows.

\begin{lemma}
  \label{lem:fast-burnin}
  For every fixed $\alpha>1$, there exists a constant $C_\alpha>0$ such that the following holds.
  Let $\lambda_i$, $0\leq i\leq L$, be a temperature in the cooling schedule \cref{schedule},
  with perfect-matching distribution $\mu_i$.
  For any $r\in V_1$, given $\alpha$-rough hole weights $w_i(r,v)$ at temperature $\lambda_i$ for all $v\in V_2$,
  there is a randomized algorithm that, in time $C_\alpha n^2\log^3 n$,
  outputs a perfect matching with distribution $\widehat\mu_i$ satisfying
  \begin{align*}
    \|\widehat\mu_i-\mu_i\|_{\mathrm{TV}}\leq \frac{1}{16}.
  \end{align*}
\end{lemma}

\subsection{Rescaling and auxiliary distributions}
\label{subsec:burnin-rescaling}

To construct the fast sampler, we first rescale the edge activities without changing the distribution on perfect matchings.
We then enlarge the state space to allow unassigned vertices and conflicts,
and define the auxiliary distributions used by the sampler.

\paragraph{Rescaling the edge activities.}
Fix positive edge activities $\lambda_i$, an anchor $r\in V_1$,
and the $\alpha$-rough weights $w_i(r,v)$ supplied to \cref{lem:fast-burnin}.
For $u\in V_1$ and $v\in V_2$, set
\begin{align}
  B(u,v):=
  \frac{\lambda_i(u,v)/w_i(r,v)}
  {\displaystyle\sum_{v'\in V_2}\lambda_i(u,v')/w_i(r,v')}.
  \label{eq:burnin-rescaling}
\end{align}
That is, we divide column $v$ of the activity matrix by $w_i(r,v)$ and then normalize each row.
The resulting matrix is positive and satisfies $\sum_{v\in V_2}B(u,v)=1$ for every $u$.
For a matching $M$ and a family of matchings $\mathcal S$, write
\begin{align*}
  B(M):=\prod_{(u,v)\in M}B(u,v),
  \qquad B(\mathcal S):=\sum_{M\in\mathcal S}B(M).
\end{align*}
Every perfect matching uses each row and each column exactly once,
so the rescaling multiplies the weights of all perfect matchings by the same constant.
In particular, for every $M\in\mathcal P$,
\begin{align}
  \frac{B(M)}{B(\mathcal P)}=\mu_i(M).
  \label{eq:burnin-perfect-distribution}
\end{align}

The column scaling is chosen to balance the near-perfect matchings whose left hole is $r$.
With the ideal weights $w_i^*(r,v)$ defined in \cref{eq:ideal-weights},
the resulting matrix $B^*$ satisfies $B^*(\mathcal N(r,v))=B^*(\mathcal P)$ for every $v$.
Rough weights preserve this balance up to constant factors.

\paragraph{Auxiliary distributions.}
The auxiliary distributions live on the expanded state space
\begin{align*}
  \Omega_{\mathrm{aux}} := \left\{ \sigma=(\sigma_u)_{u\in V_1}: \sigma_u\in V_2\cup\{\ast\} \right\}.
\end{align*}
Here $\sigma_u=v$ means that the left vertex $u$ is assigned to the right vertex $v$,
while $\sigma_u=\ast$ means that $u$ is unassigned.
Several left vertices may be assigned to the same right vertex;
allowing such conflicts lets us update one assignment at a time without requiring every intermediate configuration to be a matching.

For a configuration $\sigma\in\Omega_{\mathrm{aux}}$ and each $v\in V_2$,
let $d_v(\sigma):=|\{u\in V_1:\sigma_u=v\}|$ be the number of left vertices assigned to $v$.
We call a right vertex $v\in V_2$ \emph{crowded} in $\sigma$ if $d_v(\sigma)\geq2$.
A left vertex $u\in V_1$ is \emph{conflicted} in $\sigma$ if it is assigned to a crowded vertex,
i.e., if $\sigma_u\neq\ast$ and $d_{\sigma_u}(\sigma)\geq2$.
Let
\begin{align*}
  \mathsf{un}(\sigma):=\bigl|\{u\in V_1:\sigma_u=\ast\}\bigr|,\qquad
  \mathsf{cr}(\sigma):= \sum_{v\in V_2}\max\{d_v(\sigma)-1,0\}.
\end{align*}
The configurations with $\mathsf{cr}(\sigma)=0$ correspond exactly to matchings between $V_1$ and $V_2$.
In particular, we identify $\mathcal P$ with the set of configurations satisfying $\mathsf{un}(\sigma)=\mathsf{cr}(\sigma)=0$.
Let $\sigma^\ast$ be the configuration in which every left vertex is unassigned.

Fix $J:=J_\alpha n^2$, which determines the penalty for assignments sharing a right vertex,
where $J_\alpha\geq8\alpha^2$ is a constant depending only on $\alpha$.
(Recall that $\alpha$ is the roughness parameter of the hole weights.) For each auxiliary activity $a>0$,
define the distribution $\nu_{a}$ on $\Omega_{\mathrm{aux}}$ by
\begin{align}
  \nu_{a}(\sigma)=\frac{1}{F_{\mathsf{aux}}(a)} \frac{a^{n-\mathsf{un}(\sigma)}}{(1+J)^{\mathsf{cr}(\sigma)}} \prod_{u\in V_1:\sigma_u\neq\ast}B(u,\sigma_u),
  \label{eq:burnin-auxiliary-distribution}
\end{align}
where $F_{\mathsf{aux}}(a)$ is the normalizing factor.
For the rescaled matrix at temperature $\lambda_i$, we also write $\nu_{i,a}=\nu_{a}$.
Each assigned left vertex $u$ contributes a factor $aB(u,\sigma_u)$,
and each assignment beyond the first at a right vertex contributes a factor $(1+J)^{-1}$.
Increasing $a$ therefore favors assigned vertices, while the fixed penalty keeps suppressing crowding.
A perfect matching $M\in\mathcal P$ has unnormalized weight $a^nB(M)$,
so by \cref{eq:burnin-perfect-distribution}, conditioning on $\mathcal P$ gives $\nu_a(\cdot\mid \mathcal{P})=\mu_i$.

\begin{remark}
  Let $\mathcal{M}_j=(V_1\times V_2,\mathcal I_j)$, $j=1,2$, be the partition matroids corresponding to the left and right vertices,
  with
  \begin{align*}
    \mathcal I_1&:=\{I\subseteq V_1\times V_2:
    |I\cap(\{u\}\times V_2)|\leq1\text{ for every }u\in V_1\},\\
    \mathcal I_2&:=\{I\subseteq V_1\times V_2:
    |I\cap(V_1\times\{v\})|\leq1\text{ for every }v\in V_2\}.
  \end{align*}
  Identify each assignment $\sigma$ with $I:=\{(u,\sigma_u):u\in V_1,\ \sigma_u\ne\ast\}$.
  Under this identification,
  \begin{align*}
    \Omega_{\mathrm{aux}}=\mathcal I_1,
    \qquad
    \mathcal P=\{I\in\mathcal I_1\cap\mathcal I_2:|I|=n\}.
  \end{align*}
  Thus perfect matchings are common bases of $\mathcal{M}_1$ and $\mathcal{M}_2$.
  The rank $\mathsf{rk}_{\mathcal{M}_2}(I)$ in $\mathcal{M}_2$ counts the occupied right vertices,
  so $\mathsf{cr}(\sigma)=|I|-\mathsf{rk}_{\mathcal{M}_2}(I)$ and
  \begin{align*}
    \nu_a(I)\propto
    a^{|I|}(1+J)^{-(|I|-\mathsf{rk}_{\mathcal{M}_2}(I))}
    \prod_{(u,v)\in I}B(u,v),\qquad I\in\mathcal I_1.
  \end{align*}
  The distribution $\nu_a$ multiplies the weight by $(1+J)^{-1}$ for each assignment beyond the first at a right vertex,
  while increasing $a$ favors larger sets.
  The algorithm fixes $J=J_\alpha n^2$ and increases $a$ from $1/n$ to $J$.
  At $a=J$, \cref{lem:burnin-schedule} below and the conditional distribution above give
  \begin{align*}
    \nu_J(\mathcal P)\geq\frac12,
    \qquad \nu_J(\cdot\mid\mathcal P)=\mu_i.
  \end{align*}
  Thus common bases have constant probability at the end of the schedule.
  The sampler draws approximately from $\nu_J$ and retains perfect matchings to obtain approximate samples from $\mu_i$.
\end{remark}

\subsection{Sampling perfect matchings by increasing the assignment activity}

\paragraph{Choosing the activity values.}
Throughout this schedule, the edge activities $\lambda_i$ stay fixed;
only the assignment activity $a$ changes.

We call a distribution $\xi$ on $\Omega_{\mathrm{aux}}$ a \emph{$3$-warm start} for $\nu_a$ if
\begin{align*}
  \sum_{\sigma\in\Omega_{\mathrm{aux}}}
  \frac{\xi(\sigma)^2}{\nu_a(\sigma)}\leq3.
\end{align*}
Equivalently, $\chi^2(\xi\Vert\nu_a)\leq2$.
The following schedule keeps consecutive target distributions close in this sense,
with larger activity increments near the end.
The proof is deferred to \cref{sub:burnin-schedule-proofs}.

\begin{lemma}
  \label{lem:burnin-schedule}
  A deterministic algorithm constructs, in $O(n^2)$ time, a schedule
  \begin{align*}
    \frac{1}{n}=a_0<a_1<\cdots<a_m=J
  \end{align*}
  such that the point mass at $\sigma^\ast$ is $3$-warm for $\nu_{a_0}$ and $\nu_{a_{k-1}}$ is $3$-warm for $\nu_{a_k}$ for every $1\leq k\leq m$.
  Moreover, $\nu_J(\mathcal P)\geq1/2$, the number of stages is $m+1=O(n)$, and
  \begin{align}
    \sum_{k=0}^{m}(1+\sqrt{a_k})=O(n).
    \label{eq:burnin-schedule-cost}
  \end{align}
\end{lemma}

The next lemma gives the sampler used at each stage.

\begin{lemma}
  \label{lem:burnin-chain}
  After a shared $O(n^2)$-time preprocessing step for $B$, for every $a\in[1/n,J]$ there is a randomized algorithm that,
  from a $3$-warm start for $\nu_{i,a}$, outputs a configuration with distribution $\xi$ satisfying
  \begin{align*}
    \|\xi-\nu_{i,a}\|_{\mathrm{TV}} \leq\frac1{1000(m+1)}.
  \end{align*}
  Its running time is $O((1+\sqrt a)n\log^3 n)$ for every input configuration.
\end{lemma}

\cref{lem:burnin-schedule,lem:burnin-chain} fit together as follows.
Starting from the all-unassigned configuration $\sigma^\ast$,
we apply the sampler successively at $a_0,\ldots,a_m$, using each output to initialize the next stage.
Comparing each stage with one started from the exact preceding stationary distribution,
the warm-start guarantees ensure that the sampling errors accumulate only additively.

At the final activity, $\nu_{a_m}(\mathcal P)\geq1/2$ and $\nu_{a_m}(\,\cdot\mid\mathcal P)=\mu_i$.
Thus, rerunning the full activity schedule a constant number of times independently and returning the first output which is a perfect matching yields the required $1/16$ total variation accuracy.
If no trial succeeds, return an arbitrary fixed perfect matching.
Finally, \cref{eq:burnin-schedule-cost} bounds the total time for one run of the schedule by $O(n^2\log^3 n)$.

\subsection{Assignment Glauber dynamics}
\label{subsec:burnin-chain}

Fix positive edge activities $\lambda_i$ and an auxiliary activity $a>0$.
Let $B$, $\Omega_{\mathrm{aux}}$, and $J$ be as in \cref{subsec:burnin-rescaling},
and write $\nu:=\nu_{i,a}=\nu_a$ for the target distribution in \cref{eq:burnin-auxiliary-distribution}.

For a configuration $\sigma\in\Omega_{\mathrm{aux}}$, recall that a left vertex is conflicted if it is assigned to a right vertex shared with another left vertex.
The process updates these vertices at a higher rate.
For $u\in V_1$, define
\begin{align*}
  \theta_u(\sigma):=
  \begin{cases}
    1+J,&u\text{ is conflicted in }\sigma,\\
    1,&\text{otherwise}.
  \end{cases}
\end{align*}
At configuration $\sigma$, coordinate $u$ is resampled at rate $\theta_u(\sigma)$ from the distribution
\begin{align*}
  q_u(\ast):=\frac1{1+a},
  \qquad
  q_u(v):=\frac{aB(u,v)}{1+a}
  \quad(v\in V_2).
\end{align*}
Write $\sigma^{u\to v}$ for $\sigma$ with coordinate $u$ replaced by $v$.
The positive Laplacian is
\begin{align}
  (\mathcal L_a f)(\sigma)
  :=\sum_{u\in V_1}\theta_u(\sigma)
  \sum_{v\in V_2\cup\{\ast\}}q_u(v)
  \bigl(f(\sigma)-f(\sigma^{u\to v})\bigr).
  \label{eq:burnin-laplacian}
\end{align}
We call this continuous-time process the \emph{assignment Glauber dynamics}.
Its generator is $-\mathcal L_a$, and its transition operator after time $t$ is $e^{-t\mathcal L_a}$.

\begin{remark}
  The assignment Glauber dynamics is a continuous-time Glauber dynamics on $\Omega_{\mathrm{aux}}$.
  Its sequence of resampling events forms a discrete-time chain.
  At configuration $\sigma$, choose a coordinate $u$ with probability proportional to $\theta_u(\sigma)$ and draw its new assignment from $q_u$:
  \begin{align*}
    \Pr{\text{choose }u\mid\sigma}
    =\frac{\theta_u(\sigma)}{\sum_{z\in V_1}\theta_z(\sigma)},
    \qquad v\sim q_u,\quad \sigma_u\leftarrow v.
  \end{align*}
\end{remark}

\begin{claim}
  \label{clm:burnin-stationarity}
  The assignment Glauber dynamics is irreducible and has stationary distribution~$\nu_a$.
\end{claim}

\begin{lemma}
  \label{clm:burnin-relaxation}
  The assignment Glauber dynamics has spectral gap at least $1$.
\end{lemma}
The proofs of \cref{clm:burnin-stationarity,clm:burnin-relaxation} are given in \cref{sub:burnin-proof}.
In particular, for any $3$-warm initial distribution $\xi$,
\begin{align*}
  \|\xi e^{-t\mathcal L_a}-\nu_a\|_{\mathrm{TV}}
  \leq\frac12e^{-t}\sqrt{\chi^2(\xi\Vert\nu_a)}
  \leq\frac{\sqrt2}{2}e^{-t}.
\end{align*}

We use the standard exponential-clock construction for continuous-time Markov chains~\cite[Section~20.1]{LPW17}.
To bound the simulation time, we control the cost per resampling and the number of resamplings needed.
Let $Q_T$ count the number of resampling events up to time $T$,
including those that leave the assignment unchanged.

\begin{proposition}
  \label{prop:burnin-simulation}
  After a shared $O(n^2)$-time preprocessing step for $B$, the assignment Glauber dynamics can be simulated exactly up to any fixed time $T\geq0$,
  from any initial configuration, in time $O(n+Q_T\log{n})$.
  The initialization takes $O(n)$ time and each resampling takes $O(\log{n})$ time.
\end{proposition}

\begin{proposition}
  \label{prop:burnin-event-count}
  There is a constant $C_\alpha>0$, depending only on $\alpha$, such that, for every $a\in[1/n,J]$ and $T\geq0$,
  from any $3$-warm start for~$\nu_a$, we have
  \begin{align*}
    \Pr{Q_T>C_\alpha n(1+\sqrt a)\bigl(T+\log(2n)\bigr)^2}
    \leq e^{-T}.
  \end{align*}
\end{proposition}

The proofs of these two propositions are deferred to \cref{sub:burnin-simulation-proof}.
Together with the mixing bound above, they give the stage sampler.

\begin{proof}[Proof of \cref{lem:burnin-chain}]
  Put $\varepsilon:=1/[1000(m+1)]$ and set
  \begin{align*}
    T:=\log(2/\varepsilon),
    \qquad R:=\left\lceil C_\alpha n(1+\sqrt a)\bigl(T+\log(2n)\bigr)^2\right\rceil.
  \end{align*}
  Use \cref{prop:burnin-simulation} to simulate up to time $T$,
  returning $\sigma^\ast$ if another resampling would exceed the budget $R$.
  Let $\xi_T$ and $\widehat\xi$ be the output distributions of the simulations without and with the $R$-resampling limit,
  respectively, from the supplied $3$-warm start.
  Using the same randomness, their outputs agree unless $Q_T>R$.
  Thus the mixing bound and \cref{prop:burnin-event-count} give
  \begin{align*}
    \|\widehat\xi-\nu_a\|_{\mathrm{TV}}
    &\leq\|\widehat\xi-\xi_T\|_{\mathrm{TV}}
    +\|\xi_T-\nu_a\|_{\mathrm{TV}}\\
    &\leq\Pr{Q_T>R}+\frac{\sqrt2}{2}e^{-T}
    \leq\frac{\varepsilon}{2}+\frac{\varepsilon}{2}
    =\varepsilon.
  \end{align*}
  Since $m+1=O(n)$, we have $T=O(\log(2n))$.
  By \cref{prop:burnin-simulation}, the running time is therefore $O_\alpha(n(1+\sqrt a)\log^3(2n))$ for every input configuration,
  after the shared preprocessing.
\end{proof}

\subsection{Choosing the assignment activity values}

In this subsection, we state the estimates used to choose $a_0,\ldots,a_m$ and explain how they give \cref{lem:burnin-schedule};
the detailed proofs are in \cref{sub:burnin-schedule-proofs}.

\begin{lemma}
  \label{lem:burnin-balance}
  For every $u\in V_1$ and $v\in V_2$, the rescaled matrix $B$ satisfies
  \begin{align*}
    B(\mathcal N(u,v))\leq2\alpha^2B(\mathcal P)
    \qquad\text{ and } \qquad
    B(\mathcal N(r,v))\geq\alpha^{-2}B(\mathcal P).
  \end{align*}
\end{lemma}

These bounds follow from a weight-preserving switching argument and the $\alpha$-roughness of the supplied weights.
The upper bound ensures that perfect matchings have constant probability at the final activity,
and both bounds enter the estimate on unmatched vertices below.

For a matching $M$ sampled with probability proportional to $a^{|M|}B(M)$,
the expected number of unmatched left vertices is
\begin{align*}
  O_\alpha\left(1+\min\left\{n,\frac n{\sqrt a}\right\}\right).
\end{align*}
Allowing crowded assignments changes the partition function by only a constant factor and gives a comparable tail bound on the number of unassigned vertices.
The precise statements and proofs appear as \cref{lem:burnin-partition-comparison,lem:burnin-matching-estimates} in \cref{sub:burnin-schedule-proofs}.

We increase $a$ by a relative amount of order $(1+\min\{n,n/\sqrt a\})^{-1/2}$.
These estimates let us choose the activities so that $\nu_{a_{k-1}}$ is a $3$-warm start for $\nu_{a_k}$,
while also bounding the sampling cost at each activity.
As $a$ grows, each stage becomes more expensive, but the permitted increments grow as well.
Summing over the resulting schedule gives $\sum_k(1+\sqrt{a_k})=O_\alpha(n)$,
which keeps the total cost of initialization nearly quadratic.

Finally, the all-unassigned configuration has probability at least $e^{-1}$ at $a_0=1/n$,
while perfect matchings have total probability at least $1/2$ at $a_m=J$.
Together, these observations give all the properties in \cref{lem:burnin-schedule}.

\subsection{The FAR chain and the full algorithm}
\label{sec:algorithm-FAR}

Fix a vertex $r\in V_1$, positive edge activities $\lambda_i$,
and supplied $\alpha$-rough anchor weights $w_i(r,v)$ for all $v\in V_2$.
Let $B$ be the rescaled activities from \cref{eq:burnin-rescaling}.
The fixed anchor reweighting (FAR) chain has state space
\begin{align*}
  \Omega=\mathcal P\cup
  \bigcup_{u\in V_1,\ v\in V_2}\mathcal N(u,v).
\end{align*}
The vertex $r$ remains fixed throughout the algorithm, but both holes may move.

For $v\in V_2$, let $s_v:=\sum_{u\in V_1}B(u,v)$ be its column sum, and define
\begin{align*}
  W(u,v):=
  1+nB(u,v)+s_v(1-B(u,v))+(s_v-B(u,v)).
\end{align*}
The multiplier $W(u,v)$ is one plus the total unnormalized weight of the addition and slide moves available from a matching with holes $(u,v)$.
Dividing these move weights by $2W(u,v)$ makes the chain lazy.
Including the same factor $W(u,v)$ in the stationary weight cancels this normalization in the detailed-balance equations.

Let $\pi_B$ be the distribution on $\Omega$ given by
\begin{align*}
  \pi_B(M)\propto
  \begin{cases}
    n^2B(M),&M\in\mathcal P,\\
    W(u,v)B(M),&M\in\mathcal N(u,v).
  \end{cases}
\end{align*}

Write $P_{\mathrm{FAR}}$ for the transition matrix.
From a perfect matching $M$, the probability of deleting any given edge $e\in M$ is
\begin{align*}
  P_{\mathrm{FAR}}(M,M\setminus\{e\})=\frac1{2n}.
\end{align*}
From $N\in\mathcal N(u,v)$, the addition probability is
\begin{align*}
  P_{\mathrm{FAR}}(N,N\cup\{(u,v)\})
  =\frac{nB(u,v)}{2W(u,v)}.
\end{align*}
For each $(y,z)\in N$, the two slide probabilities are
\begin{align*}
  P_{\mathrm{FAR}}\bigl(
  N,N\setminus\{(y,z)\}\cup\{(u,z)\}
  \bigr)
  &=\frac{B(u,z)s_v}{2W(u,v)},\\
  P_{\mathrm{FAR}}\bigl(
  N,N\setminus\{(y,z)\}\cup\{(y,v)\}
  \bigr)
  &=\frac{B(y,v)}{2W(u,v)}.
\end{align*}
The first slide fixes the right hole $v$ and moves the left hole to $y$;
the second fixes the left hole $u$ and moves the right hole to $z$.
All other non-loop probabilities are zero, and the remaining probability is assigned to the self-loop.

The chain is lazy, irreducible, and reversible with stationary distribution $\pi_B$.
Moreover,
\begin{align*}
  \pi_B(\,\cdot\mid\mathcal P)=\mu_i,
  \qquad
  \pi_B(\mathcal P)\geq\frac1{1+8\alpha^2}.
\end{align*}
Indeed, $\sum_{u,v}W(u,v)\leq4n^2$ and $B(\mathcal N(u,v))\leq2\alpha^2B(\mathcal P)$.

The column sums, the values $W(u,v)$, and the row and column prefix sums of $B$ can be computed in $O(n^2)$ time.
For a left slide, sample $z\ne v$ from row $u$;
for a right slide, sample $y\ne u$ from column $v$, with probabilities proportional to the corresponding entries of $B$.
Thus each transition takes $\widetilde O(1)$ time.

\begin{remark}[Continuous-time interpretation of FAR]
  \label{rem:FAR-continuous}
  Let $\mu$ be the distribution on $\Omega$ with weights
  \begin{align*}
    \mu(M)\propto
    \begin{cases}
      n^2B(M),&M\in\mathcal P,\\
      B(M),&M\in\mathcal N(u,v).
    \end{cases}
  \end{align*}
  Define a continuous-time JSV chain whose state is a matching $M\in\Omega$.
  An action $b$ specifies one of the following operations, including the chosen edge where applicable:
  \begin{itemize}
    \item \emph{Perfect matching $M\in\mathcal P$.}
      \begin{itemize}
        \item \emph{Deletion.} For each edge $e\in M$, delete $e$ at rate $1/n$.
      \end{itemize}
    \item \emph{Near-perfect matching $M\in\mathcal N(u,v)$.}
      \begin{itemize}
        \item \emph{Addition.} Add the edge $(u,v)$ at rate $nB(u,v)$.
        \item \emph{Left slide.} For each edge $(y,z)\in M$, replace it by $(u,z)$ at rate $B(u,z)s_v$,
          moving the left hole from $u$ to $y$.
        \item \emph{Right slide.} For each edge $(y,z)\in M$, replace it by $(y,v)$ at rate $B(y,v)$,
          moving the right hole from $v$ to $z$.
        \item \emph{Holding.} Include an action of rate one that leaves the state unchanged.
      \end{itemize}
  \end{itemize}
  All other rates between distinct states are zero.
  These rates make the chain reversible with respect to $\mu$.
  The choice of $s_v$ is motivated by the flow analysis.
  Reversibility would also hold with $s_v$ replaced by any nonnegative factor depending only on $v$.

  Let $q_M(b)$ be the rate of action $b$ at state $M$.
  The total action rate is
  \begin{align*}
    R(M):=\sum_bq_M(b)=
    \begin{cases}
      1,&M\in\mathcal P,\\
      W(u,v),&M\in\mathcal N(u,v).
    \end{cases}
  \end{align*}
  A naive way to obtain a discrete-time chain is to divide all action rates by the common bound $\Lambda:=\max_{M\in\Omega}R(M)$,
  assigning the remaining probability to holding.
  This preserves $\mu$, but $\Lambda$ can be $\Theta(n)$, whereas $R(M)=1$ at every perfect matching.
  It can therefore take an expected $\Theta(n)$ steps just to leave a perfect matching.
  To avoid these extra holding steps, we normalize by $R(M)$ at each state,
  allowing the stationary weights to change.

  To obtain the FAR chain, define weights on pairs and their sums over actions by
  \begin{align*}
    w\bigl(\underbrace{M}_{\text{state}},\underbrace{b}_{\text{action}}\bigr)
    &:=\mu(M)q_M(b),\\
    w(M,\ast)&:=\sum_b w(M,b)=\mu(M)R(M).
  \end{align*}
  The state weights and conditional action probabilities satisfy
  \begin{align*}
    \pi_B(M)\propto w(M,\ast),
    \qquad
    \Pr{\text{choose }b\mid M}
    =\frac{w(M,b)}{w(M,\ast)}=\frac{q_M(b)}{R(M)}.
  \end{align*}
  Staying put with probability $1/2$ and otherwise choosing and performing an action by this rule gives the FAR chain.
  The same hole-weight learning procedure applies after dividing each observed frequency of holes $(u,v)$ by the known factor $W(u,v)$.
\end{remark}

\paragraph{Estimating the probabilities of holes.}

For the FAR chain, we have the following analog of \cref{lem:occupation-estimation}.
The key improvement is from $O(n^3\log{n})$ in \cref{lem:occupation-estimation} to $O(n^2\log{n})$ as follows.

\begin{lemma}
  \label{lem:FAR-occupation-estimation}
  Fix positive edge activities $\lambda_i$, a vertex $r\in V_1$, and $\alpha$-rough anchor weights $w_i(r,v)$.
  Let $X_0\sim\pi_B$, and let $X_1,X_2,\ldots$ evolve according to the corresponding FAR chain.
  For every $v\in V_2$, write $p_v:=\pi_B(\mathcal N(r,v))$.
  There is a constant $C_\alpha>0$, depending only on $\alpha$, such that, for every integer $T\geq1$,
  \begin{align*}
    \operatorname{Var}\left[
      \frac1T\sum_{t=0}^{T-1}
      \mathbf{1}_{\{X_t\in\mathcal N(r,v)\}}
      \right]
    \leq
    \frac{C_\alpha n^2\log n}{T}\,p_v^2.
  \end{align*}
\end{lemma}

\Cref{lem:FAR-occupation-estimation} is proved in \cref{subsec:far-occupation}.
Essentially, the FAR chain transition probabilities are designed to optimize the flow analysis there,
which yields an improvement by a factor of $n$ over the bound in \cref{lem:occupation-estimation}.

The distribution $\mu_i$ has density at most $1+8\alpha^2$ relative to $\pi_B$.
Hence \cref{lem:FAR-occupation-estimation} and Chebyshev's inequality imply that $C_\alpha n^2\log n$ transitions suffice for constant-relative-error estimation of each $p_v$,
with constant success probability, when the initial perfect matching is drawn from $\mu_i$.

An initial distribution within total variation distance $1/16$ of $\mu_i$ changes the probability of any trajectory event by at most $1/16$.
Thus, after increasing the constant in the trajectory length,
each independent batch estimates any specified probability of having holes $(r,v)$ to relative error $1/10$ with probability at least $3/4$.
For an internal confidence parameter $\delta\in(0,1)$, keep the weights fixed and take coordinatewise medians over $O(\log(n/\delta))$ independent batches.
This gives these estimates simultaneously for all pairs $(r,v)$ with probability at least $1-\delta$.

\paragraph{Updating the anchor weights.}
To update the weights, note that $w_i(r,v)W(r,v)/p_v$ equals $w_i^*(r,v)$ up to a common factor independent of $v$,
because of the stationary distribution $\pi_B$ and \cref{eq:burnin-rescaling}.
For positive frequency estimates $\widehat p_v$, we therefore set
\begin{align}
  \widetilde w_i(r,v)&:=\frac{w_i(r,v)W(r,v)}{\widehat p_v},\notag\\
  w_i'(r,v)&:=\widetilde w_i(r,v)
  \sum_{z\in V_2}\frac{\lambda_i(r,z)}{\widetilde w_i(r,z)}.
  \label{eq:far-weight-update}
\end{align}
Expanding the permanent along row $r$ gives $\sum_z\lambda_i(r,z)/w_i^*(r,z)=1$,
so the normalization recovers $w_i^*$ from exact frequencies.
As in the analysis of \cref{eq:jsv-weight-update}, relative frequency errors of at most $1/10$ give $w_i'(r,v)/w_i^*(r,v)\in[9/11,11/9]$.
Thus the updated anchor weights are accurate, with error depending only on the current frequency estimates.

\paragraph{The whole \texorpdfstring{$\widetilde{O}(n^3)$}{Õ(n3)}-time algorithm.}
The algorithm follows the two-phase framework of \cref{sec:overview} along the cooling schedule of \cref{lem:cooling-schedule},
with two substitutions.
Every perfect matching that was obtained from the JSV chain in \cref{sec:overview},
whether as the endpoint of a burn-in or as a Phase-2 sample,
is now produced by an independent call to the sampler at the current temperature,
and the learning steps of the JSV chain are replaced by FAR transitions.
Consequently, for an anchor $r\in V_1$ fixed once and for all,
only the $n$ anchor weights are learned and stored at each temperature,
and no fixed starting matching is needed, since every sampler call starts from the all-unassigned configuration $\sigma^\ast$.
In Phase~1, at each temperature $\lambda_i$ with $i\geq1$, the supplied weights are the accurate anchor weights learned at $\lambda_{i-1}$,
which are rough at $\lambda_i$ by \cref{eq:cooling-ratio}.
\paragraph{Phase~1.}
Fix a confidence parameter $\delta\in(0,1)$ for the learning procedure.
At $\lambda_0=1$, start with the exact anchor weights $w_0(r,v)=n$.
At each of the $L$ subsequent temperatures in \cref{schedule},
keep the supplied weights fixed and use an odd number $R=O(\log(nL/\delta))$ of independent learning batches.
Initialize each batch by a separate run of the sampler in \cref{lem:fast-burnin},
followed by $O(n^2\log n)$ FAR transitions, and take coordinatewise median frequencies.
If any required median is zero, return zero;
otherwise, apply \cref{eq:far-weight-update}.
With a sufficiently large constant in $R$, conditioned on all preceding refinements succeeding,
the failure probability at each temperature is at most $\delta/(2L)$.
Since $4\sqrt2/3<2$, accurate weights remain rough at the next temperature by \cref{eq:cooling-ratio}.
Summing over the possible temperatures of the first failed refinement bounds the probability of any failure by $L\cdot\delta/(2L)=\delta/2$.
Thus the weights are accurate at all temperatures with probability at least $1-\delta/2$.

Each batch costs $O(n^2\log^3 n)$ time, including initialization and the $O(\log n)$ cost per FAR transition.
The total Phase~1 cost is therefore
\begin{align*}
  O(LR\,n^2\log^3 n)
  =O\left(n^3\log^5 n\,\log(2n/\delta)\right).
\end{align*}

\paragraph{Phase~2.}
We use the whole-product estimator from \cref{subsec:phase-two},
but generate its samples by independent runs of the perfect-matching sampler in \cref{lem:fast-burnin}.
With the learned weights fixed, we obtain $O(\varepsilon^{-2})$ samples at every temperature,
using the sampler with an adjustable error tolerance in \cref{lem:fast-burnin-general} to control the total sampling error.
The second-moment bound in \cref{eq:cooling-variance} and Chebyshev's inequality then give the required approximation.
The total Phase~2 running time is $O(n^3\varepsilon^{-2}\log^6(n/\varepsilon))$;
the detailed error analysis is given in \cref{subsec:improved-main-proof}.

\section{Energy-weighted multicommodity flows}
\label{sec:energy-flows-overview}

This section develops the multicommodity flow method used for the occupation bounds of the JSV and FAR chains and the relaxation bound of the JSV chain.

The common tool in this section is a comparison between weighted squared differences of function values at pairs of states and the Dirichlet form,
which measures differences across single-step transitions.
To obtain it, we connect the endpoint pairs by paths of allowed moves,
expand each endpoint difference along its path, and control the weight counted at each transition.
We build on classical canonical-path and multicommodity-flow bounds~\cite{DS91,Sinclair92}.
Coupled and transport flows~\cite{CVY26,CFJ25,CCFV2025} give additional freedom in choosing the endpoint pairs.
Our contribution refines the routing analysis by replacing the usual path-length cost with \emph{routing energy},
the sum of the squared probabilities that a sampled path uses each transition.
For our routings, paths may have length $\Theta(n)$ while the routing energy is only $O(\log n)$.

\subsection{Classical multicommodity-flow bounds}
\label{sec:classical-flows}

Let $P$ be a finite lazy irreducible reversible Markov chain on $\Omega$ with stationary distribution $\pi$.
Its undirected transition graph has edge set $\mathcal T:=\{\{x,y\}:x\neq y,\ P(x,y)>0\}$,
whose elements we call \emph{transitions}.
For $I,F\in\Omega$, a path from $I$ to $F$ is a sequence $\gamma=(I=x_0,x_1,\ldots,x_k=F)$ such that $\{x_{j-1},x_j\}\in\mathcal T$ for every $1\leq j\leq k$;
its length is $|\gamma|:=k$.
Let $\Paths(I,F)$ be the set of paths from $I$ to $F$ with no repeated states,
including the empty path when $I=F$.

\begin{definition}
  \label{def:fractional-routing}
  A \emph{routing} from $I$ to $F$ is a function $\theta: \Paths(I,F) \to [0,1]$ such that
  \begin{equation*}
    \sum_{\gamma\in\Paths(I,F)}\theta(\gamma)=1.
  \end{equation*}
\end{definition}

Routing $d$ units of flow according to $\theta$ sends $d\,\theta(\gamma)$ units along each path $\gamma$.
This generalizes the classical \emph{canonical paths}, which are recovered when the routing of each pair $(I,F)$ puts all its weight on a single path.
Recall that the variance of a function $f$ can be written as
\begin{equation}
  \Var[\pi]{f} = \frac12 \sum_{I,F\in\Omega} \pi(I)\pi(F)(f(I)-f(F))^2.
  \label{eq:flow-variance-identity}
\end{equation}
The standard multicommodity flow therefore sends $\pi(I)\pi(F)$ units of flow from $I$ to $F$ for every ordered pair $(I,F)$.
Since it also sends $\pi(I)\pi(F)$ units from $F$ to $I$, the factor $1/2$ in \cref{eq:flow-variance-identity} accounts for this double counting.
The \emph{capacity} of a transition $e=\{x,y\}\in\mathcal T$ is
\begin{equation*}
  Q(e):=\pi(x)P(x,y)=\pi(y)P(y,x),
\end{equation*}
so that the Dirichlet form of $f$ can be written as
\begin{equation*}
  \Dir_P(f,f)=\sum_{e=\{x,y\}\in\mathcal T} Q(e)(f(x)-f(y))^2.
\end{equation*}
Fix a routing $\theta_{I,F}$ for every ordered pair $(I,F)$.
The \emph{congestion} and the \emph{length-weighted congestion} are
\begin{align}
  \rho &:= \max_{e\in\mathcal T} \frac1{2Q(e)} \sum_{I,F\in\Omega}\pi(I)\pi(F) \sum_{\substack{\gamma\in\Paths(I,F)\\e\in\gamma}} \theta_{I,F}(\gamma), \label{eq:classical-congestion} \\
  \mathcal R_{\mathrm{len}} &:= \max_{e\in\mathcal T} \frac1{2Q(e)} \sum_{I,F\in\Omega}\pi(I)\pi(F) \sum_{\substack{\gamma\in\Paths(I,F)\\e\in\gamma}} \theta_{I,F}(\gamma)|\gamma|. \notag
\end{align}
Thus $\rho$ compares the total flow through each transition with twice its capacity,
while $\mathcal R_{\mathrm{len}}$ also weights each unit of flow by the length of its path.
Let $\ell:=\max\{|\gamma|:\theta_{I,F}(\gamma)>0\text{ for some }I,F\}$ be the maximum length of a path carrying flow.
The following bound is standard;
we include its proof for comparison with our new bound.

\begin{lemma}
  \label{lem:length-weighted-flow}
  For every $f:\Omega\to\mathbb R$,
  \begin{equation*}
    \Var[\pi]{f} \leq \mathcal R_{\mathrm{len}}\Dir_P(f,f) \leq \ell\rho\,\Dir_P(f,f).
  \end{equation*}
  Consequently, $\trel(P)\leq \mathcal R_{\mathrm{len}}\leq \ell\rho$.
\end{lemma}

\begin{proof}
  Fix $\gamma=(x_0,\ldots,x_k)\in\Paths(I,F)$, where $x_0=I$ and $x_k=F$.
  Telescoping and Cauchy--Schwarz give
  \begin{equation}
    (f(I)-f(F))^2 =\left(\sum_{j=0}^{k-1}(f(x_j)-f(x_{j+1}))\right)^2 \leq|\gamma|\sum_{e=\{x,y\}\in\gamma}(f(x)-f(y))^2.
    \label{eq:classical-path-bound}
  \end{equation}
  Average this inequality over the routing, multiply by $\pi(I)\pi(F)/2$, and sum over endpoint pairs.
  Grouping the resulting contributions by transition yields
  \begin{align}
    \Var[\pi]{f} &\leq\frac12\sum_{I,F}\pi(I)\pi(F) \E[\gamma\sim\theta_{I,F}]{|\gamma|\sum_{e=\{x,y\}\in\gamma}(f(x)-f(y))^2} \label{eq:classical-path-bound-expectation} \\
    &=\sum_{e=\{x,y\}\in\mathcal T} \left[\frac12\sum_{I,F}\pi(I)\pi(F) \sum_{\substack{\gamma\in\Paths(I,F)\\e\in\gamma}}\theta_{I,F}(\gamma)|\gamma|\right](f(x)-f(y))^2 \notag \\
    &\leq\mathcal R_{\mathrm{len}}\Dir_P(f,f). \notag
  \end{align}
  Since every path carrying flow has length at most $\ell$, the definitions give $\mathcal R_{\mathrm{len}}\leq\ell\rho$.
  The bound on the relaxation time follows from \cref{eq:poincare-characterization}.
\end{proof}

Thus the classical flow analysis has two components: the maximum congestion over transitions,
measured by $\rho$, and the lengths of the paths carrying this flow.
For matching chains, natural paths may have length $\Theta(n)$,
which costs an additional factor $n$ even after the loads of the transitions have been bounded.

\subsection{An energy-based analysis of multicommodity flow}
\label{subsec:energy-flows}

For an energy-based analysis, we view routings probabilistically.
By \cref{def:fractional-routing}, a routing $\theta_{I,F}$ from $I$ to $F$ is a probability distribution on $\Paths(I,F)$.
Thus, for each transition $e$, the total weight of the paths containing $e$ is exactly the probability that a random path sampled from $\theta_{I,F}$ uses $e$.
\begin{equation*}
  \sum_{\substack{\gamma\in\Paths(I,F)\\e\in\gamma}} \theta_{I,F}(\gamma) = \Pr[\gamma\sim\theta_{I,F}]{e\in\gamma}.
\end{equation*}
The \emph{transition support} of a routing $\theta$ is the set of transitions it uses with positive probability,
given by
\begin{equation*}
  \supp(\theta) := \left\{ e\in\mathcal T: \Pr[\gamma\sim\theta]{e\in\gamma}>0 \right\}.
\end{equation*}
The routing energy measures how a routing spreads its flow over the transitions.

\begin{definition}
  \label{def:routing-energy}
  The \emph{energy} of a routing $\theta$ is defined by
  \begin{equation*}
    \En(\theta):=\sum_{e\in\mathcal T} \Pr[\gamma\sim\theta]{e\in\gamma}^2.
  \end{equation*}
\end{definition}

To compare the routing energy with the path-length factor of \cref{sec:classical-flows},
first consider a routing supported on a single path;
its energy is exactly the length of this path.
More generally,
\begin{equation*}
  \En(\theta) \leq \sum_e\Pr[\gamma\sim\theta]{e\in\gamma} = \E[\gamma\sim\theta]{|\gamma|} \leq \max_{\theta(\gamma)>0}|\gamma|.
\end{equation*}
Thus the routing energy never exceeds the maximum path length, but it can be much smaller.
We define the \emph{support congestion} as
\begin{equation*}
  \rho_{\mathrm{supp}} :=\max_{e\in\mathcal T}\frac1{2Q(e)} \sum_{\substack{I,F\in\Omega\\e\in\supp(\theta_{I,F})}} \pi(I)\pi(F).
\end{equation*}
Unlike the congestion $\rho$ in \cref{eq:classical-congestion},
the support congestion charges the full demand of each pair whose routing can use $e$,
instead of weighting it by $\Pr[\gamma\sim\theta_{I,F}]{e\in\gamma}$.
Since these probabilities are at most one, $\rho\leq\rho_{\mathrm{supp}}$.
Both $\rho$ and $\rho_{\mathrm{supp}}$ measure congestion alone,
whereas $\mathcal R_{\mathrm{len}}$ also accounts for path lengths.
Let $\En_{\max}:=\max_{I,F}\En(\theta_{I,F})$.
We have the following analogue of \cref{lem:length-weighted-flow}.

\begin{lemma}
  \label{lem:energy-flow}
  For every $f:\Omega\to\mathbb R$,
  \begin{equation*}
    \Var[\pi]{f} \leq\En_{\max}\rho_{\mathrm{supp}}\Dir_P(f,f).
  \end{equation*}
  Consequently, $\trel(P)\leq\En_{\max}\rho_{\mathrm{supp}}$.
\end{lemma}

\begin{proof}
  Fix an ordered pair $(I,F)$.
  Telescoping along a path $\gamma$ and taking absolute values gives $|f(I)-f(F)|\leq\sum_{e=\{x,y\}\in\gamma}|f(x)-f(y)|$.
  Averaging over $\gamma\sim\theta_{I,F}$ before applying Cauchy--Schwarz gives
  \begin{equation*}
    |f(I)-f(F)| \leq\sum_{e=\{x,y\}\in\supp(\theta_{I,F})} \Pr[\gamma\sim\theta_{I,F}]{e\in\gamma}|f(x)-f(y)|.
  \end{equation*}
  The sum of the squared coefficients is $\En(\theta_{I,F})$, so
  \begin{equation}
    (f(I)-f(F))^2 \leq\En(\theta_{I,F}) \sum_{e=\{x,y\}\in\supp(\theta_{I,F})}(f(x)-f(y))^2.
    \label{eq:routing-energy-bound}
  \end{equation}
  Multiply by $\pi(I)\pi(F)/2$, sum over endpoint pairs, and group by transition.
  Using $\En(\theta_{I,F})\leq\En_{\max}$ gives
  \begin{equation*}
    \Var[\pi]{f} \leq\En_{\max}\sum_{e=\{x,y\}\in\mathcal T} \left[\frac12\sum_{\substack{I,F\\e\in\supp(\theta_{I,F})}}\pi(I)\pi(F)\right](f(x)-f(y))^2 \leq\En_{\max}\rho_{\mathrm{supp}}\Dir_P(f,f).
  \end{equation*}
  The bound on the relaxation time follows from \cref{eq:poincare-characterization}.
\end{proof}

If each routing consists of a single path, then $\En_{\max}=\ell$ and $\rho_{\mathrm{supp}}=\rho$,
which recovers the classical bound $\trel(P)\leq\ell\rho$.
For a fractional routing, a smaller energy yields a better bound,
provided that the support congestion is also controlled.
In \cref{sub:routing-energy-idea}, we show that two-ended routing achieves logarithmic energy even when every path has linear length.

\begin{remark}
  The key difference between the proofs of \cref{lem:length-weighted-flow,lem:energy-flow} is the order in which we take the expectation and apply Cauchy--Schwarz.
  In the classical proof, we apply Cauchy--Schwarz to each path separately,
  which produces the path-length factor $|\gamma|$ in \cref{eq:classical-path-bound},
  and then average over the routing in \cref{eq:classical-path-bound-expectation}.
  In the energy-based proof, we first average the sum of the absolute increments along the path.
  This gives a sum over transitions whose coefficients are the usage probabilities $\Pr[\gamma\sim\theta_{I,F}]{e\in\gamma}$,
  and applying Cauchy--Schwarz to this averaged sum yields the factor $\En(\theta_{I,F})$ in \cref{eq:routing-energy-bound}.
  Taking the expectation first allows the bound to capture how the routing spreads its flow over the transitions.
\end{remark}

\subsection{Coupled flows}
\label{subsec:coupled-flows}

In the preceding two subsections, we assigned the demand $\pi(I)\pi(F)$ to every ordered pair $(I,F)$ in order to bound
\begin{equation*}
  \Var[\pi]{f} =\frac12\sum_{I,F\in\Omega}\pi(I)\pi(F)(f(I)-f(F))^2 =\frac12\E[I\sim\pi,F\sim\pi]{(f(I)-f(F))^2}.
\end{equation*}
To interpret the last expression, draw $I$ and $F$ independently from $\pi$,
and take half the expected squared difference between $f(I)$ and $f(F)$.
We now generalize this framework by allowing the two endpoints to be dependent and to have different marginal distributions.
This is the setting of coupling-based canonical paths and transport flows~\cite{CCFV2025,CFJ25,CVY26}.

Let $\mu$ and $\nu$ be probability distributions on $\Omega$.
Recall that a \emph{coupling} of $\mu$ and $\nu$ is a probability distribution $\kappa$ on $\Omega\times\Omega$ such that
\begin{equation*}
  \sum_{F\in\Omega}\kappa(I,F)=\mu(I), \qquad \sum_{I\in\Omega}\kappa(I,F)=\nu(F).
\end{equation*}
Thus, if $(I,F)\sim\kappa$, then $I\sim\mu$ and $F\sim\nu$, but the endpoints need not be independent.
We assign the demand $\kappa(I,F)$ to each ordered pair $(I,F)$.
The freedom in choosing the coupling allows us to place the demand on pairs that are easier to connect in the transition graph,
while preserving both marginal distributions.
The resulting flow comparison bounds the coupled squared difference
\begin{equation*}
  \E[(I,F)\sim\kappa]{(f(I)-f(F))^2} = \sum_{I,F\in\Omega} \kappa(I,F)(f(I)-f(F))^2.
\end{equation*}
Such a bound also controls the squared difference between the two expectations,
since the marginal conditions and Jensen's inequality give
\begin{equation*}
  \bigl(\E[\mu]{f}-\E[\nu]{f}\bigr)^2 = \bigl(\E[(I,F)\sim\kappa]{f(I)-f(F)}\bigr)^2 \leq \E[(I,F)\sim\kappa]{(f(I)-f(F))^2}.
\end{equation*}

Now fix a coupling $\kappa$ and a routing $\theta_{I,F}$ for every pair $(I,F)$ with $\kappa(I,F)>0$.
As in \cref{sec:classical-flows}, define the congestion and the length-weighted congestion by
\begin{align*}
  \rho(\kappa) &:=\max_{e\in\mathcal T}\frac1{Q(e)} \sum_{I,F\in\Omega}\kappa(I,F) \Pr[\gamma\sim\theta_{I,F}]{e\in\gamma}, \\
  \mathcal R_{\mathrm{len}}(\kappa) &:=\max_{e\in\mathcal T}\frac1{Q(e)} \sum_{I,F\in\Omega}\kappa(I,F) \sum_{\substack{\gamma\in\Paths(I,F)\\e\in\gamma}} \theta_{I,F}(\gamma)|\gamma|.
\end{align*}
Here and below, pairs with zero demand are omitted from the sums.
There is no factor $1/2$ in these definitions, because we bound the coupled squared difference itself instead of the variance.
As before, $\ell$ denotes the maximum length of a path carrying flow.
We have the following coupling-based analogue of \cref{lem:length-weighted-flow}.

\begin{lemma}
  \label{lem:coupled-length-flow}
  For every $f:\Omega\to\mathbb R$,
  \begin{equation*}
    \E[(I,F)\sim\kappa]{(f(I)-f(F))^2} \leq\mathcal R_{\mathrm{len}}(\kappa)\Dir_P(f,f) \leq\ell\rho(\kappa)\Dir_P(f,f).
  \end{equation*}
\end{lemma}

\begin{proof}
  Apply \cref{eq:classical-path-bound} to each path, multiply by $\kappa(I,F)\theta_{I,F}(\gamma)$,
  and sum over endpoint pairs and paths.
  The coefficient of $(f(x)-f(y))^2$ for $e=\{x,y\}$ is $\sum_{I,F}\kappa(I,F)\sum_{\gamma\ni e}\theta_{I,F}(\gamma)|\gamma|$,
  which is at most $\mathcal R_{\mathrm{len}}(\kappa)Q(e)$.
  Summing over transitions proves the first inequality, and $|\gamma|\leq\ell$ proves the second.
\end{proof}

\paragraph{Choosing a coupling for matchings.}
We now illustrate how the freedom in choosing a coupling simplifies the flow construction for our occupation estimates.
Consider the JSV chain with positive edge activities $\lambda$ and positive hole weights $w$,
write $\pi=\pi_{\lambda,w}$, and fix $u\in V_1$ and $v\in V_2$.
As explained in \cref{sec:overview}, we want to estimate $p_{u,v}=\pi(\mathcal N(u,v))$ by the frequency with which the chain visits $\mathcal N(u,v)$.

Let $\mu_{u,v}$ and $\mu_\emptyset$ be the distributions supported on $\mathcal N(u,v)$ and $\mathcal P$,
respectively, given by
\begin{equation*}
  \mu_{u,v}(I)=\frac{\lambda(I)}{\lambda(\mathcal N(u,v))}, \qquad \mu_\emptyset(F)=\frac{\lambda(F)}{\lambda(\mathcal P)}.
\end{equation*}
Since the hole weight is constant within each block, $\mu_{u,v}$ and $\mu_\emptyset$ are precisely the conditional distributions of $\pi$ on $\mathcal N(u,v)$ and $\mathcal P$.
Consequently, for every $f:\Omega\to\mathbb R$,
\begin{equation*}
  \E[\pi]{\bigl(\one_{\mathcal N(u,v)}-p_{u,v}\bigr)f} = p_{u,v} \bigl(\E[\mu_{u,v}]{f}-\E[\pi]{f}\bigr).
\end{equation*}
Thus a bound on the squared difference of the two expectations on the right in terms of the Dirichlet form automatically yields the factor $p_{u,v}^2$ needed for an occupation estimate with relative error.
The resulting variance bounds for trajectories are derived in \cref{subsec:JSV-occupation}.
We use $\mu_\emptyset$ as a common reference distribution.
Since $\pi$ is a mixture of the block-conditional distributions,
uniform comparisons with $\mu_\emptyset$ also control the difference from $\pi$, up to a constant factor.

The product coupling samples $I\sim\mu_{u,v}$ and $F\sim\mu_\emptyset$ independently.
Their symmetric difference $I\oplus F$ then consists of an alternating path $\mathcal A$ between the holes $u$ and $v$ of $I$,
which is augmenting for $I$, and a collection of alternating cycles.
Routing from $I$ all the way to $F$ would require changing the matching along all of these components.
Following~\cite{CVY26}, we instead use $F$ only as auxiliary randomness and define the new endpoint
\begin{equation*}
  J:=I\oplus\mathcal A,
\end{equation*}
which fills the two holes and agrees with $I$ outside $\mathcal A$.
Crucially, $J$ still has the required perfect-matching distribution.

\begin{lemma}[\cite{CVY26}]
  \label{lem:augmenting-path-coupling}
  The pair $(I,J)$ is a coupling of $\mu_{u,v}$ and $\mu_\emptyset$.
\end{lemma}

\begin{proof}
  By definition, the independent pair $(I,F)$ has probability
  \begin{equation*}
    \mu_{u,v}(I)\mu_\emptyset(F) =\frac{\lambda(I)\lambda(F)} {\lambda(\mathcal N(u,v))\lambda(\mathcal P)}.
  \end{equation*}
  Put $K:=F\oplus\mathcal A$ and recall that $J:=I\oplus\mathcal A$.
  Then $K\in\mathcal N(u,v)$ and $J\in\mathcal P$.
  The map $(I,F)\mapsto(K,J)$ is an involution of $\mathcal N(u,v)\times\mathcal P$,
  since the symmetric difference $K\oplus J=I\oplus F$ has the same unique augmenting path $\mathcal A$,
  and switching it again recovers $(I,F)$.
  Moreover, switching only exchanges ownership of the edges of $\mathcal A$, so
  \begin{equation*}
    \lambda(I)\lambda(F)=\lambda(K)\lambda(J).
  \end{equation*}
  The involution therefore preserves the product distribution $\mu_{u,v}\otimes\mu_\emptyset$.
  In particular, $J$ has distribution $\mu_\emptyset$;
  since $I$ has distribution $\mu_{u,v}$ by construction, $(I,J)$ has the asserted marginals.
\end{proof}

Switching the augmenting path $\mathcal A$ may still require $\Theta(n)$ moves.
To combine this coupling with the energy bound of \cref{subsec:energy-flows},
we next extend the flow framework so that it retains the auxiliary information used to construct the routing.

\subsection{Energy bounds for weighted flows}
\label{subsec:energy-comparison}

In the matching coupling of the preceding subsection, the endpoints of the routing are $(I,J)$,
but the demand is most conveniently counted using the original independent pair $(I,F)$.
Different auxiliary matchings $F$ may produce the same endpoints,
and retaining $F$ keeps the explicit product weight $\mu_{u,v}(I)\mu_\emptyset(F)$ available in the load calculation.
We now extend the energy comparison to accommodate this distinction between the endpoints and the auxiliary information used to construct them.

The extension also allows arbitrary nonnegative demands.
The endpoints need not be independent, and the total demand need not be normalized.
The only requirement is that, after summing over the auxiliary information,
the total weight assigned to each pair of endpoints equals its prescribed demand.

\begin{definition}
  \label{def:routing-family}
  Let $\kappa:\Omega\times\Omega\to\mathbb R_{\geq0}$ be an \emph{endpoint demand},
  and let $\mathcal L$ be a finite set of auxiliary labels.
  A \emph{weighted routing family} $\mathcal F$ for $\kappa$ consists of weights $\widehat\kappa(I,a,J)\geq0$ such that,
  for every $I,J\in\Omega$,
  \begin{equation}
    \sum_{a\in\mathcal L}\widehat\kappa(I,a,J)=\kappa(I,J),
    \label{eq:routing-family-demand}
  \end{equation}
  together with a routing $\theta_{I,a,J}$ from $I$ to $J$ for each triple $(I,a,J)$ with positive weight.
  We call such a triple a \emph{routing occurrence}.
  The label $a$ records auxiliary information used to construct the pair of endpoints, to specify its routing,
  or to count its contribution to the load.
  The weights and routings do not depend on the function $f$.
\end{definition}

Throughout this subsection, sums involving routings omit triples with $\widehat\kappa(I,a,J)=0$.
The occurrence $(I,a,J)$ sends $\widehat\kappa(I,a,J)\theta_{I,a,J}(\gamma)$ units of flow along each path $\gamma\in\Paths(I,J)$.
Since each routing is a probability distribution on paths,
\begin{equation*}
  \sum_{a\in\mathcal L} \sum_{\gamma\in\Paths(I,J)} \widehat\kappa(I,a,J)\theta_{I,a,J}(\gamma) = \kappa(I,J).
\end{equation*}
Thus splitting the demand into occurrences does not change the total flow sent from $I$ to $J$.
When $\kappa$ is a coupling, $\widehat\kappa$ is a joint distribution on triples whose marginal on $(I,J)$ is $\kappa$.
If the routing depends only on the endpoints, the auxiliary label may be summed out,
although retaining it can simplify the load calculation.

For a weighted routing family $\mathcal F$, we reuse the notation $\En_{\max}:=\max_{\widehat\kappa(I,a,J)>0}\En(\theta_{I,a,J})$ for its maximum routing energy,
and define its support congestion and its energy-weighted congestion by
\begin{align*}
  \rho_{\mathrm{supp}} &:= \max_{e\in\mathcal T}\frac1{Q(e)} \sum_{\substack{I,a,J:\\ e\in\supp(\theta_{I,a,J})}} \widehat\kappa(I,a,J), \\
  \mathcal R_{\En} &:= \max_{e\in\mathcal T}\frac1{Q(e)} \sum_{\substack{I,a,J:\\ e\in\supp(\theta_{I,a,J})}} \widehat\kappa(I,a,J)\En(\theta_{I,a,J}).
\end{align*}
If $\kappa\equiv0$, we set $\En_{\max}=\rho_{\mathrm{supp}}=\mathcal R_{\En}=0$.
Since each routing has energy at most $\En_{\max}$, we have
\begin{equation}
  \mathcal R_{\En} \leq \En_{\max}\rho_{\mathrm{supp}}.
  \label{eq:energy-factorization}
\end{equation}
This is the energy analogue of the factorization $\mathcal R_{\mathrm{len}}(\kappa) \leq\ell\rho(\kappa)$,
with the routing energy in place of the path length and the support congestion in place of the ordinary congestion.

\begin{lemma}
  \label{lem:energy-MCF}
  For every weighted routing family $\mathcal F$ with endpoint demand $\kappa$ and every $f:\Omega\to\mathbb R$,
  \begin{equation}
    \sum_{I,J\in\Omega}\kappa(I,J)(f(I)-f(J))^2 \leq \mathcal R_{\En}\Dir_P(f,f) \leq \En_{\max}\rho_{\mathrm{supp}}\Dir_P(f,f).
    \label{eq:coupled-energy-flow-bound}
  \end{equation}
\end{lemma}

\begin{proof}
  Apply \cref{eq:routing-energy-bound} to the routing of each occurrence, multiply by $\widehat\kappa(I,a,J)$,
  and sum.
  The demand identity \cref{eq:routing-family-demand} and regrouping by transition give
  \begin{align*}
    \sum_{I,J}\kappa(I,J)(f(I)-f(J))^2 &\leq\sum_{e=\{x,y\}\in\mathcal T} \left[\sum_{\substack{I,a,J\\e\in\supp(\theta_{I,a,J})}} \widehat\kappa(I,a,J)\En(\theta_{I,a,J})\right](f(x)-f(y))^2 \\
    &\leq\mathcal R_{\En}\Dir_P(f,f).
  \end{align*}
  The second inequality in \cref{eq:coupled-energy-flow-bound} follows from \cref{eq:energy-factorization}.
\end{proof}

When $\kappa$ is a coupling of $\mu$ and $\nu$, the left-hand side of \cref{eq:coupled-energy-flow-bound} is $\E[(I,J)\sim\kappa]{(f(I)-f(J))^2}$.
Together with Jensen's inequality, the lemma therefore gives
\begin{equation*}
  \bigl(\E[\mu]{f}-\E[\nu]{f}\bigr)^2 \leq \E[(I,J)\sim\kappa]{(f(I)-f(J))^2} \leq \mathcal R_{\En}\Dir_P(f,f).
\end{equation*}
This is the energy analogue of \cref{lem:coupled-length-flow},
and it is the form used in our occupation estimates.
Alternatively, taking $\kappa(I,J)=\pi(I)\pi(J)/2$ with one occurrence per ordered pair makes the left-hand side equal to $\Var[\pi]{f}$,
and the support congestion becomes exactly that of \cref{subsec:energy-flows}.
This recovers \cref{lem:energy-flow}, with the factor $1/2$ now carried by the demand.

\paragraph{The matching coupling as a weighted routing family.}
We conclude by expressing the matching coupling of \cref{lem:augmenting-path-coupling} in this framework.
Fix a hole pattern $(u,v)$, and let $\kappa_{u,v}$ be the corresponding coupling.
The auxiliary labels are the perfect matchings, $\mathcal L=\mathcal P$,
and the label of an occurrence is the independently sampled matching $F$.
For $I\in\mathcal N(u,v)$ and $F\in\mathcal P$, let $\mathcal A(I,F)$ be the augmenting path in $I\oplus F$.
The occurrence $(I,F,I\oplus\mathcal A(I,F))$ receives the weight
\begin{equation*}
  \widehat\kappa(I,F,I\oplus\mathcal A(I,F)) := \mu_{u,v}(I)\mu_\emptyset(F) = \frac{\lambda(I)\lambda(F)} {\lambda(\mathcal N(u,v))\lambda(\mathcal P)},
\end{equation*}
and all other occurrences receive weight zero.
By the construction of the coupling,
\begin{equation*}
  \sum_{F\in\mathcal P}\widehat\kappa(I,F,J) = \kappa_{u,v}(I,J).
\end{equation*}
For each occurrence of positive weight, we choose a routing $\theta_{I,F,J}$ from $I$ to $J$ that changes only the edges of $\mathcal A(I,F)$,
using slides followed by a final addition.
Together with the weights, these routings define the weighted routing family $\mathcal F_{u,v}$ for $\kappa_{u,v}$;
the specific randomized routing and its energy bound are given in \cref{sub:routing-energy-idea}.
Note that each occurrence carries the explicit product weight of $(I,F)$, but its routing connects $I$ to $J$,
not to $F$, and randomizes how the augmenting path is switched.
This separates the product weights used to count the demand from the routing probabilities used to bound the energy.

\subsection{Flow bounds for JSV and FAR chains}
\label{subsec:flow-bounds-overview}

We now state the resulting flow bounds for the JSV and FAR chains.
For the JSV chain, we have the following bounds, proved in \cref{subsec:JSV-combining-bounds}.

\begin{lemma}
  \label{lem:JSV-flow-summary}
  Fix a temperature $\lambda_i$, and consider the JSV chain defined using $\alpha$-rough hole weights $w_i(u,v)$ for all $u\in V_1$ and $v\in V_2$.
  For every hole pattern $(u,v)$, the coupling $\kappa_{u,v}$ admits a weighted routing family satisfying
  \begin{align*}
    \En_{\max}=O(\log n),
    \qquad
    \rho_{\mathrm{supp}}=O_\alpha(n^3),
    \qquad
    \mathcal R_{\En}=O_\alpha(n^3\log n).
  \end{align*}
\end{lemma}

Consequently, by \cref{lem:energy-MCF}, for every $f:\Omega\to\mathbb R$,
\begin{align*}
  \E[(I,J)\sim\kappa_{u,v}]{(f(I)-f(J))^2}
  \leq
  O_\alpha(n^3\log n)
  \Dir_{P_{\mathrm{JSV}}}(f,f).
\end{align*}

The factor $O_\alpha(n^3)$ comes from comparing each transition's capacity with the total weight of all triples whose routing can use it.
The factor $O(\log n)$ comes from the routing energy, even though the paths themselves may have length $\Theta(n)$.

We also need to control variation among perfect matchings.
Using the cycle routing from \cref{lem:cycle-energy-overview} and the same transition-capacity comparison gives the following bound.

\begin{lemma}
  \label{lem:perfect-matching-comparison}
  For every $f:\Omega\to\mathbb R$,
  \begin{align*}
    \Var[\mu_\emptyset]{f}
    \leq
    C_\alpha n^3
    \Dir_{P_{\mathrm{JSV}}}(f,f).
  \end{align*}
\end{lemma}

The proofs of \cref{lem:JSV-flow-summary,lem:perfect-matching-comparison} are given in \cref{sec:JSV-flow}.

The FAR chain has the same conditional distribution within each hole pattern as the JSV chain.
Indeed, the FAR reweighting depends only on the holes $(u,v)$ and not on the particular matching in $\mathcal N(u,v)$.
Thus, conditioned on a hole pattern, both chains have the $\lambda_i$-weighted distribution,
and we may use the same augmenting-path coupling and routing family.

\begin{lemma}
  \label{lem:FAR-flow-summary}
  Fix a temperature $\lambda_i$ and $r\in V_1$, and let $B$ be the matrix defined from $\alpha$-rough hole weights $w_i(r,v)$ for all $v\in V_2$.
  Consider the FAR chain defined using $B$, and let $\mathcal Z$ denote the normalizing constant of its stationary distribution $\pi_B$.
  For every hole pattern $(u,v)$, the coupling $\kappa_{u,v}$ admits a weighted routing family with
  \begin{align*}
    \En_{\max}=O(\log n)
  \end{align*}
  and, for every $f:\Omega\to\mathbb R$,
  \begin{align}
    \E[(I,J)\sim\kappa_{u,v}]{(f(I)-f(J))^2}
    \leq
    C_\alpha
    \frac{\mathcal Z}{B(\mathcal N(u,v))}
    \log n\,
    \Dir_{P_{\mathrm{FAR}}}(f,f).
    \label{eq:FAR-fixed-pattern-flow}
  \end{align}
\end{lemma}

For holes $(r,v)$, \cref{eq:FAR-fixed-pattern-flow}, $\mathcal Z=O_\alpha(n^2B(\mathcal P))$,
and $B(\mathcal N(r,v))\geq\alpha^{-2}B(\mathcal P)$ give
\begin{equation}
  \E[(I,J)\sim\kappa_{r,v}]{(f(I)-f(J))^2}
  \leq
  C_\alpha n^2\log n\, \Dir_{P_{\mathrm{FAR}}}(f,f).
  \label{eq:FAR-anchor-flow}
\end{equation}

We also need a bound averaged over all hole patterns.
Let
\begin{align*}
  p_{u,v}:=\pi_B(\mathcal N(u,v))
  =
  \frac{W(u,v)B(\mathcal N(u,v))}{\mathcal Z}.
\end{align*}
Multiplying \cref{eq:FAR-fixed-pattern-flow} by $p_{u,v}$ and summing gives
\begin{align}
  \sum_{u,v}p_{u,v}
  \E[(I,J)\sim\kappa_{u,v}]{(f(I)-f(J))^2}
  \leq
  C_\alpha n^2\log n\,
  \Dir_{P_{\mathrm{FAR}}}(f,f).
  \label{eq:FAR-averaged-flow}
\end{align}
Here the potentially small factor $B(\mathcal N(u,v))$ cancels because $p_{u,v}\mathcal Z/B(\mathcal N(u,v))=W(u,v)$,
and $\sum_{u,v}W(u,v)=O(n^2)$.

We prove the support-congestion bounds for the JSV and FAR chains,
as well as \cref{lem:perfect-matching-comparison}, in \cref{sec:JSV-flow}.
The next subsection explains the common routing constructions and their energy bounds.

\subsection{Routing energy on cycles and augmenting paths}
\label{sub:routing-energy-idea}

We first consider two perfect matchings that differ on one alternating cycle.
Randomizing where we start switching the cycle gives a routing of constant energy.
We then consider a near-perfect matching and a perfect matching that differ on one augmenting path.
Randomizing which end of the path to advance gives a routing of logarithmic energy.
Both constructions use moves available in the JSV and the FAR chains.

\begin{lemma}
  \label{lem:cycle-energy-overview}
  Let $I,F\in\mathcal{P}$ such that $I\oplus F$ consists of a single alternating cycle of length $2k$.
  There is a routing $\theta_{I,F}$ from $I$ to $F$ such that
  \begin{align*}
    \En(\theta_{I,F})
    \leq2.
  \end{align*}
\end{lemma}

\begin{proof}
  The cycle contains $k$ edges of $I$ and $k$ edges of $F$.
  Choose one of its $k$ edges of $I$ uniformly at random and delete it.
  Keep the resulting hole in $V_2$ fixed, and move the hole in $V_1$ around the cycle by $k-1$ successive slides,
  each replacing an edge of $I$ with an edge of $F$.
  Finally, add the remaining edge of $F$.
  Only the starting edge is randomized;
  the direction is fixed by keeping the hole in $V_2$ stationary.
  This defines a routing $\theta_{I,F}$ supported on $k$ paths, each with $k+1$ transitions: one deletion,
  $k-1$ slides, and one addition.

  Different starting edges leave different fixed holes in $V_2$.
  Consequently, the $k$ possible paths have no intermediate matching in common,
  and their transition sets are disjoint.
  Each path is chosen with probability $1/k$, so every transition in the support is used with probability $1/k$.
  There are $k(k+1)$ such transitions.
  By the definition of routing energy,
  \begin{align*}
    \En(\theta_{I,F})
    &=\sum_{e\in\supp(\theta_{I,F})}
    \Pr[\gamma\sim\theta_{I,F}]{e\in\gamma}^2
    =k(k+1)\left(\frac1k\right)^2
    =1+\frac1k\leq2. \qedhere
  \end{align*}
\end{proof}

Classical canonical paths fix the component order and the starting point of each switch~\cite{JS89}.
In the construction above, each path still has length $k+1$,
but spreading the flow among $k$ choices of the starting edge reduces the energy to at most two.

For an augmenting path, we instead spread the flow by choosing which endpoint to advance at each step.
Suppose that switching the path requires $k$ slides and one final addition.
After $p$ slides from one end and $q$ slides from the other, the total number of slides is $t=p+q$.
The pair $(p,q)$ records the numbers of slides made from the two ends.
There are $t+1$ possible pairs after $t$ slides;
see \cref{fig:two-ended-splits}, where each level groups pairs with the same total number of slides.
We will choose the next endpoint so that these pairs are equally likely.

\begin{figure}[htbp]
  \centering
  \begin{tikzpicture}[
    x=1.5cm,y=.95cm,
    split/.style={inner sep=2pt,font=\scriptsize},
    routing step/.style={-{Stealth[length=1.5mm]},thin},
    probability/.style={fill=white,inner sep=1pt,font=\scriptsize},
    row label/.style={font=\small}
  ]
    \foreach \t in {0,...,3}{
      \foreach \p in {0,...,\t}{
        \pgfmathtruncatemacro{\q}{\t-\p}
        \node[split] (split-\t-\p) at ({\p-\t/2},-\t)
          {$(\p,\q)$};
      }
      \pgfmathtruncatemacro{\countsplit}{\t+1}
      \node[row label,anchor=east] at (-2.6,-\t) {$t=\t$};
      \node[row label,anchor=west] at (2.6,-\t)
        {\ifnum\t=0 $1$ pair\else $\countsplit$ pairs\fi};
    }
    \foreach \t in {0,...,2}{
      \pgfmathtruncatemacro{\nextlevel}{\t+1}
      \foreach \p in {0,...,\t}{
        \pgfmathtruncatemacro{\nextp}{\p+1}
        \pgfmathtruncatemacro{\leftnumerator}{\t-\p+1}
        \pgfmathtruncatemacro{\denominator}{\t+2}
        \pgfmathtruncatemacro{\leftgcd}{gcd(\leftnumerator,\denominator)}
        \pgfmathtruncatemacro{\rightgcd}{gcd(\nextp,\denominator)}
        \pgfmathtruncatemacro{\leftnum}{\leftnumerator/\leftgcd}
        \pgfmathtruncatemacro{\leftden}{\denominator/\leftgcd}
        \pgfmathtruncatemacro{\rightnum}{\nextp/\rightgcd}
        \pgfmathtruncatemacro{\rightden}{\denominator/\rightgcd}
        \draw[routing step] (split-\t-\p) --
          node[probability,midway,sloped,above,fill=none] {$\leftnum/\leftden$}
          (split-\nextlevel-\p);
        \draw[routing step] (split-\t-\p) --
          node[probability,midway,sloped,above,fill=none] {$\rightnum/\rightden$}
          (split-\nextlevel-\nextp);
      }
    }
    \node at (0,-3.6) {$\vdots$};
  \end{tikzpicture}
  \caption{Numbers of slides from the two ends.
    Arrow labels are conditional probabilities for the next slide.
    The first four levels are shown ($k\geq3$).}
  \label{fig:two-ended-splits}
\end{figure}

After $p$ slides from one end and $q$ from the other, with $p+q=t<k$, choose the next slide according to
\begin{align*}
  (p,q)\longrightarrow
  \begin{cases}
    (p+1,q),&\text{with probability }(p+1)/(t+2),\\
    (p,q+1),&\text{with probability }(q+1)/(t+2).
  \end{cases}
\end{align*}
These probabilities sum to one.
Suppose that every pair after $t$ slides has probability $1/(t+1)$.
The pair $(r,t+1-r)$ after the next slide can be reached from $(r-1,t+1-r)$ or $(r,t-r)$.
Its probability is therefore
\begin{align*}
  \frac1{t+1}\frac{r}{t+2}
  +\frac1{t+1}\frac{t+1-r}{t+2}
  =\frac1{t+2},
\end{align*}
for $0\leq r\leq t+1$, with a missing predecessor contributing zero.
Thus the possible pairs are equally likely after the next slide.

\begin{lemma}
  \label{lem:two-ended-energy-overview}
  Let $I\in\mathcal N(u,v)$ and $J\in\mathcal P$ such that $I\oplus J$ consists of a single augmenting path between $u$ and $v$,
  containing $k$ edges of $I$ and $k+1$ edges of $J$.
  There is a routing $\theta_{I,J}$ from $I$ to $J$ such that
  \begin{align*}
    \En(\theta_{I,J})=O(\log(k+2)).
  \end{align*}
\end{lemma}

\begin{proof}
  Advance along the augmenting path from both ends.
  After $p$ slides from $u$ and $q$ from $v$, put $t=p+q$ and, if $t<k$,
  choose the next endpoint using the transition rule above.
  Each slide replaces an edge of $I$ by the adjacent edge of $J$.
  After $k$ slides, add the remaining edge of $J$.
  The two hole positions determine $(p,q)$ uniquely, so no state reached after $t$ slides, for $t=0,\ldots,k$,
  can be reached after a different number of slides.
  The final state is $J$.
  This defines the routing $\theta_{I,J}$.

  The preceding calculation shows inductively that the $t+1$ possible pairs after $t$ slides are equally likely,
  starting from $(0,0)$.

  Let $E_t$ be the set of transitions from states reached after $t$ slides,
  including the final additions when $t=k$.
  Every sampled path uses exactly one transition in $E_t$, and hence
  \begin{align*}
    \sum_{e\in E_t}\Pr[\gamma\sim\theta_{I,J}]{e\in\gamma}=1.
  \end{align*}
  To use a particular transition in $E_t$, the path must first visit its initial state.
  Since each possible state after $t$ slides has probability $1/(t+1)$,
  \begin{align*}
    \Pr[\gamma\sim\theta_{I,J}]{e\in\gamma}\leq\frac1{t+1}
    \qquad(e\in E_t).
  \end{align*}
  The sets $E_0,\ldots,E_k$ are disjoint.
  Summing their contributions to the routing energy gives
  \begin{align*}
    \En(\theta_{I,J})
    &=\sum_{t=0}^k\sum_{e\in E_t}
    \Pr[\gamma\sim\theta_{I,J}]{e\in\gamma}^2
    \leq\sum_{t=0}^k\frac1{t+1}
    \sum_{e\in E_t}\Pr[\gamma\sim\theta_{I,J}]{e\in\gamma}
    =\sum_{t=0}^k\frac1{t+1}
    \leq1+\log(k+1).
  \end{align*}
  Thus, although every path has length $k+1$, its routing contributes at most $1/(t+1)$ energy for transitions after $t$ slides,
  giving a logarithmic total.
\end{proof}

\subsection{Occupation estimates and relaxation bounds}
\label{subsec:JSV-occupation}

We next use the coupled flow comparison to estimate the probabilities of the blocks $\mathcal P$ and $\mathcal N(u,v)$ along a stationary trajectory.
The goal is a variance bound proportional to the square of each probability,
so that the trajectory gives a relative-error estimate.
We first compare the conditional means of a function on different blocks,
and then use this comparison to bound the sum of stationary covariances.

\paragraph{Occupation estimate for the JSV chain.}

\begin{proof}[Proof of \cref{lem:occupation-estimation}]
  Write $P:=P_{\mathrm{JSV}}$ and $\pi:=\pi_{\lambda,w}$.
  Fix a block $S\in\mathcal H$, and put $p_S:=\pi(S)$, $g:=\one_S-p_S$, and $K:=C_0n^3\log n$,
  where $C_0$ is chosen below.

  For each block $R\in\mathcal H$, write $p_R:=\pi(R)$ and $\mu_R:=\pi(\,\cdot\mid R)$.
  By \cref{lem:JSV-flow-summary,lem:energy-MCF}, choosing $C_0$ sufficiently large gives a coupling $\kappa_R$ of $\mu_R$ and $\mu_\emptyset$ such that
  \begin{align*}
    \E[(I,J)\sim\kappa_R]{(f(I)-f(J))^2}
    \leq K\Dir_P(f,f).
  \end{align*}
  For $R=\mathcal P$, take the diagonal coupling $I=J\sim\mu_\emptyset$.

  Fix a function $f$ with $\E[\pi]{f}=0$, and let $m_R:=\E[\mu_R]{f}$ be its conditional mean on $R$.
  Since $\kappa_R$ has marginals $\mu_R$ and $\mu_\emptyset$, Jensen's inequality gives
  \begin{align*}
    |m_R-m_{\mathcal P}|^2
    =\bigl|\E[(I,J)\sim\kappa_R]{f(I)-f(J)}\bigr|^2
    \leq\E[(I,J)\sim\kappa_R]{(f(I)-f(J))^2}
    \leq K\Dir_P(f,f).
  \end{align*}
  The blocks in $\mathcal H$ partition $\Omega$, so $\sum_Rp_R=1$ and $\sum_Rp_Rm_R=\E[\pi]{f}=0$.
  It follows that
  \begin{align*}
    |m_S|
    =\left|\sum_Rp_R(m_S-m_R)\right|
    \leq|m_S-m_{\mathcal P}|
    +\sum_Rp_R|m_R-m_{\mathcal P}|
    \leq2\sqrt{K\Dir_P(f,f)}.
  \end{align*}
  Moreover, since $g=\one_S-p_S$ and $f$ has mean zero,
  \begin{align*}
    \E[\pi]{gf}
    =\E[\pi]{\one_Sf}-p_S\E[\pi]{f}
    =p_Sm_S.
  \end{align*}
  Combining these two observations gives
  \begin{align}
    \bigl(\E[\pi]{gf}\bigr)^2
    \leq4Kp_S^2\Dir_P(f,f).
    \label{eq:occupation-indicator-comparison}
  \end{align}
  Thus the flow bound controls the correlation of the pattern indicator with any centered test function,
  at the required scale $p_S^2$.

  To apply this comparison to a stationary trajectory $(X_t)$,
  write $(P^rg)(x):=\E[X_0=x]{g(X_r)}$ for $r\geq0$, and put
  \begin{align*}
    c_r:=\Cov(g(X_0),g(X_r))=\E[\pi]{gP^rg},
    \qquad D_T:=\sum_{r=0}^{T-1}c_r.
  \end{align*}
  Laziness and reversibility imply that $P$ has nonnegative eigenvalues, so $c_r\geq0$ for every $r$.
  We will bound their partial sum $D_T$ by using the test function
  \begin{align*}
    f_T:=\sum_{r=0}^{T-1}P^rg.
  \end{align*}
  Stationarity gives $\E[\pi]{f_T}=0$ and $\E[\pi]{gf_T}=D_T$, while telescoping gives $f_T-Pf_T=g-P^Tg$.
  By reversibility, $\E[\pi]{(P^rg)(P^Tg)}=c_{T+r}$.
  Consequently,
  \begin{align*}
    \Dir_P(f_T,f_T)
    =\E[\pi]{f_T(f_T-Pf_T)}
    =\E[\pi]{f_Tg}-\E[\pi]{f_TP^Tg}
    =D_T-\sum_{r=0}^{T-1}c_{T+r}
    \leq D_T.
  \end{align*}
  Applying \cref{eq:occupation-indicator-comparison} to $f_T$ now yields
  \begin{align*}
    D_T^2=\bigl(\E[\pi]{gf_T}\bigr)^2
    \leq4Kp_S^2\Dir_P(f_T,f_T)
    \leq4Kp_S^2D_T.
  \end{align*}
  If $D_T>0$, dividing by $D_T$ gives $D_T\leq4Kp_S^2$;
  the same bound holds when $D_T=0$.
  Finally, stationarity and $c_r\geq0$ imply
  \begin{align*}
    \Var{\frac1T\sum_{t=0}^{T-1}\one_S(X_t)}
    =\frac1{T^2}\left(Tc_0+2\sum_{r=1}^{T-1}(T-r)c_r\right)
    \leq\frac{2D_T}{T}\leq\frac{8K}{T}p_S^2.
  \end{align*}
  This is \cref{eq:JSV-occupation-estimation} with $C=8C_0$.
\end{proof}

\paragraph{Occupation estimate for the FAR chain.}
\label{subsec:far-occupation}

For the FAR chain, the flow bound for prescribed source holes may depend on their locations.
Instead, we use a uniform bound for the sets $\mathcal N(r,v)$ and a bound averaged over all blocks with their stationary probabilities.
These suffice for the same conditional-mean and covariance argument as in the preceding proof.

\begin{proof}[Proof of \cref{lem:FAR-occupation-estimation}]
  Write $P:=P_{\mathrm{FAR}}$, whose stationary distribution is $\pi_B$.
  For each block $R\in\mathcal H$, put $p_R:=\pi_B(R)$ and $\mu_R:=\pi_B(\,\cdot\mid R)$.
  In particular, write $p_{u,v}:=\pi_B(\mathcal N(u,v))$.

  Fix $f$ with $\E[\pi_B]{f}=0$, and write $m_R:=\E[\mu_R]{f}$ and $m_{u,v}:=m_{\mathcal N(u,v)}$.
  Set
  \begin{align*}
    K_{\mathrm{FAR}}:=C_\alpha n^2\log n,
  \end{align*}
  where $C_\alpha$ is sufficiently large.

  By Jensen's inequality and \cref{eq:FAR-anchor-flow}, for every $v\in V_2$,
  \begin{align}
    (m_{r,v}-m_{\mathcal P})^2
    \leq
    K_{\mathrm{FAR}}\Dir_P(f,f).
    \label{eq:FAR-anchor-means}
  \end{align}
  Similarly, Jensen's inequality and \cref{eq:FAR-averaged-flow} give
  \begin{align}
    \sum_{u,v}p_{u,v}
    (m_{u,v}-m_{\mathcal P})^2
    \leq
    K_{\mathrm{FAR}}\Dir_P(f,f).
    \label{eq:FAR-averaged-means}
  \end{align}

  Since the blocks in $\mathcal H$ partition $\Omega$ and $\E[\pi_B]{f}=0$,
  we have $m_{\mathcal P} =\sum_{u,v}p_{u,v}(m_{\mathcal P}-m_{u,v})$.
  Cauchy--Schwarz and \cref{eq:FAR-averaged-means} give $m_{\mathcal P}^2\leq K_{\mathrm{FAR}}\Dir_P(f,f)$.
  Together with \cref{eq:FAR-anchor-means}, this yields
  \begin{align*}
    m_{r,v}^2
    \leq
    2(m_{r,v}-m_{\mathcal P})^2+2m_{\mathcal P}^2
    \leq4K_{\mathrm{FAR}}\Dir_P(f,f).
  \end{align*}

  Fix $v\in V_2$, put $S:=\mathcal N(r,v)$, and let $g:=\one_S-p_S$.
  Since $\E[\pi_B]{gf}=p_Sm_{r,v}$, we have
  \begin{align*}
    \bigl(\E[\pi_B]{gf}\bigr)^2
    \leq4K_{\mathrm{FAR}}p_S^2\Dir_P(f,f).
  \end{align*}
  The FAR chain is lazy and reversible, so the stationary covariance calculation in the preceding proof applies unchanged.
  Thus, for every integer $T\geq1$,
  \begin{align*}
    \Var{\frac1T\sum_{t=0}^{T-1}\one_{\{X_t\in S\}}}
    \leq
    \frac{8K_{\mathrm{FAR}}}{T}\,p_S^2
    =
    \frac{C_\alpha n^2\log n}{T}\,p_S^2,
  \end{align*}
  after adjusting $C_\alpha$.
  This proves the estimate for the probability of holes $(r,v)$, for every $v\in V_2$.
\end{proof}

\paragraph{Relaxation time of the JSV chain.}
\label{subsec:JSV-relaxation}

We now combine the flow comparisons above to prove the relaxation-time bound for the JSV chain in \cref{lem:jsv-relaxation}.

\begin{proof}[Proof of \cref{lem:jsv-relaxation}]
  Write
  \begin{align*}
    P:=P_{\mathrm{JSV}},
    \qquad
    \pi:=\pi_{\lambda,w},
    \qquad
    m:=\E[\mu_\emptyset]{f}.
  \end{align*}
  For each hole pattern $(u,v)$, put $p_{u,v}:=\pi(\mathcal N(u,v))$.
  Since $\kappa_{u,v}$ couples $\mu_{u,v}$ with $\mu_\emptyset$,
  \begin{align*}
    \E[\mu_{u,v}]{(f-m)^2}
    \leq
    2\E[(I,J)\sim\kappa_{u,v}]{(f(I)-f(J))^2}
    +
    2\Var[\mu_\emptyset]{f}.
  \end{align*}
  Averaging over the hole patterns gives
  \begin{align*}
    \Var[\pi]{f}
    \leq
    2\sum_{u,v}p_{u,v}
    \E[(I,J)\sim\kappa_{u,v}]{(f(I)-f(J))^2}
    +
    2\Var[\mu_\emptyset]{f}.
  \end{align*}
  Since $\sum_{u,v}p_{u,v}\leq1$, \cref{lem:JSV-flow-summary,lem:energy-MCF} imply that the first term is $O_\alpha(n^3\log n)\Dir_P(f,f)$.
  By \cref{lem:perfect-matching-comparison}, the second term is $O_\alpha(n^3)\Dir_P(f,f)$.
  Therefore
  \begin{align*}
    \Var[\pi]{f}
    \leq
    C_\alpha n^3\log n\,\Dir_P(f,f).
  \end{align*}
  The Poincar\'e characterization gives
  \begin{align*}
    \tau_{\mathrm{rel}}
    \bigl(P_{\mathrm{JSV}}\bigr)
    & \leq C_\alpha n^3\log n.\qedhere
  \end{align*}
\end{proof}

\section{Proofs of the flow bounds}
\label{sec:JSV-flow}

We prove \cref{lem:JSV-flow-summary,lem:FAR-flow-summary,lem:perfect-matching-comparison}.
We first use the same two-ended routing to establish flow bounds for prescribed source holes in both chains,
and then compare perfect matchings using the cycle routing.

Throughout this section, fix a temperature $\lambda_i$.
For the JSV chain, let $w_i(u,v)$ be $\alpha$-rough hole weights for all $u\in V_1,v\in V_2$.
For the FAR chain, fix $r\in V_1$ and let $B$ be the matrix defined from $\alpha$-rough anchor weights $w_i(r,v)$,
$v\in V_2$.

The chains $P_{\mathrm{JSV}}$ and $P_{\mathrm{FAR}}$ have stationary distributions $\pi_{\lambda_i,w_i}$ and $\pi_B$,
respectively.
For a transition $T=\{M,M'\}$, define $Q_{\mathrm{JSV}}(T):=\pi_{\lambda_i,w_i}(M)P_{\mathrm{JSV}}(M,M')$ and $Q_{\mathrm{FAR}}(T):=\pi_B(M)P_{\mathrm{FAR}}(M,M')$.
By reversibility, both expressions are unchanged when $M$ and $M'$ are interchanged.

\subsection{Matching identities and recovery of the original matchings}
\label{subsec:matching-identities}

Abbreviate $Z_\emptyset:=\lambda_i(\mathcal P)$ and $Z_{u,v}:=\lambda_i(\mathcal N(u,v))$.
Both chains have conditional distributions $\mu_\emptyset(M)=\lambda_i(M)/Z_\emptyset$ for $M\in\mathcal P$ and $\mu_{u,v}(M)=\lambda_i(M)/Z_{u,v}$ for $M\in\mathcal N(u,v)$.
In particular, these distributions do not depend on the supplied hole weights.

\paragraph{Recovering the original matchings from the remaining edges.}
Write $I\oplus F$ for symmetric difference and $I\uplus F$ for the multiset retaining both copies of every common edge.
To count the possible pairs $(I,F)$, we record the edges left in $I\uplus F$ after removing the current matching and,
for a slide, the two edges specified below.
Those two edges are recorded separately, so these data recover $I\uplus F$.
This follows the approach of~\cite{BSVV08}.
Once $I\uplus F$ and the original source holes are known, we can determine which edges of the augmenting path belong to $I$ and which belong to $F$.
On an alternating cycle where the current matching agrees with $I$, it identifies the edges belonging to $I$.
Thus reconstruction is unique when every cycle is unchanged.
If one cycle is being switched and all other cycles are unchanged, there are at most two reconstructions,
corresponding to its two possible assignments of alternating edges to $I$ and $F$.
Common edges belong to both matchings.

\begin{lemma}
  \label{lem:cofactor}
  For every $u\in V_1$ and $v\in V_2$,
  \begin{align}
    \lambda_i(u,v)Z_{u,v}\leq Z_\emptyset,
    \label{eq:cofactor-ineq}
  \end{align}
  and, for every fixed $u\in V_1$, $\sum_{v\in V_2}\lambda_i(u,v)Z_{u,v}=Z_\emptyset$.
  The identity with the vertex classes interchanged also holds.
\end{lemma}

\begin{proof}
  Adding $(u,v)$ identifies $\mathcal N(u,v)$ with the perfect matchings containing that edge and multiplies the weight by $\lambda_i(u,v)$.
  For fixed $u$, these classes partition $\mathcal P$.
\end{proof}

For distinct $x,y\in V_1$ and distinct $v,a\in V_2$, let $\mathcal N^{(2)}(x,y;v,a)$ be the set of matchings with exactly these four holes,
and put $Z^{(2)}_{x,y;v,a}:=\sum_{K\in\mathcal N^{(2)}(x,y;v,a)}\lambda_i(K)$.
If two holes in the same vertex class coincide, the set is empty and its partition function is zero.

\begin{lemma}
  \label{lem:matching-switching}
  For distinct $x,y\in V_1$ and distinct $v,a\in V_2$,
  \begin{align}
    Z_{x,a}Z_{y,v}
    &\leq Z_{x,v}Z_{y,a}+Z_\emptyset Z^{(2)}_{x,y;v,a},
    \label{eq:two-hole-switching}\\
    Z_\emptyset Z^{(2)}_{x,y;v,a}
    &\leq Z_{x,v}Z_{y,a}+Z_{x,a}Z_{y,v}.
    \label{eq:four-hole-switching}
  \end{align}
  Consequently, for every $x\in V_1$ and distinct $v,a\in V_2$,
  \begin{align}
    Z_{x,v}\sum_{y\in V_1}\lambda_i(y,v)Z_{y,a}
    \leq2Z_{x,a}Z_\emptyset.
    \label{eq:transposed-two-hole}
  \end{align}
  Finally, for $x\in V_1$ and pairwise distinct $s,v,a\in V_2$,
  \begin{align}
    Z_{x,v}\sum_{y\in V_1}\lambda_i(y,v)Z^{(2)}_{x,y;s,a}
    \leq2Z_{x,s}Z_{x,a}.
    \label{eq:one-row-four-hole}
  \end{align}
\end{lemma}

\begin{proof}
  For \cref{eq:two-hole-switching}, superpose a red matching in $\mathcal N(x,a)$ and a blue matching in $\mathcal N(y,v)$.
  The alternating path starting at $x$ ends at either $y$ or $a$, since bipartite parity excludes $v$.
  Exchanging the red and blue edges along this path gives, respectively,
  a pair in $\mathcal N(y,a)\times\mathcal N(x,v)$ or a pair in $\mathcal P\times\mathcal N^{(2)}(x,y;v,a)$.
  Each map is inverted by switching the path starting at $x$, and preserves every edge multiplicity.
  Summing product weights proves the inequality.

  For \cref{eq:four-hole-switching}, superpose a red perfect matching and a blue matching in $\mathcal N^{(2)}(x,y;v,a)$.
  The alternating path starting at $x$ ends at either $v$ or $a$.
  Switching its colors gives, respectively, a pair in $\mathcal N(x,v)\times\mathcal N(y,a)$ or $\mathcal N(x,a)\times\mathcal N(y,v)$.
  The switch is invertible and preserves product weights, proving the bound.

  To prove \cref{eq:transposed-two-hole}, apply \cref{eq:two-hole-switching} with $a$ and $v$ interchanged to obtain $Z_{x,v}Z_{y,a}\leq Z_{x,a}Z_{y,v}+Z_\emptyset Z^{(2)}_{x,y;a,v}$ for $y\neq x$.
  We multiply by $\lambda_i(y,v)$ and sum over $y\neq x$.
  Partitioning perfect matchings and matchings in $\mathcal N(x,a)$ according to the vertex matched to $v$ gives,
  respectively, $\sum_{y\neq x}\lambda_i(y,v)Z_{y,v}=Z_\emptyset-\lambda_i(x,v)Z_{x,v}$ and $\sum_{y\neq x}\lambda_i(y,v)Z^{(2)}_{x,y;a,v}=Z_{x,a}$.
  Including the $y=x$ term separately, we therefore obtain
  \begin{align*}
    Z_{x,v}\sum_y\lambda_i(y,v)Z_{y,a}
    &=\lambda_i(x,v)Z_{x,v}Z_{x,a}
    +Z_{x,v}\sum_{y\neq x}\lambda_i(y,v)Z_{y,a}\\
    &\leq\lambda_i(x,v)Z_{x,v}Z_{x,a}
    +Z_{x,a}\sum_{y\neq x}\lambda_i(y,v)Z_{y,v}
    +Z_\emptyset\sum_{y\neq x}\lambda_i(y,v)Z^{(2)}_{x,y;a,v}\\
    &=\lambda_i(x,v)Z_{x,v}Z_{x,a}
    +Z_{x,a}\bigl(Z_\emptyset-\lambda_i(x,v)Z_{x,v}\bigr)
    +Z_\emptyset Z_{x,a}\\
    &=2Z_{x,a}Z_\emptyset.
  \end{align*}

  For \cref{eq:one-row-four-hole}, replace row $x$ by a row whose only nonzero entry is a one in column $s$,
  and apply \cref{eq:transposed-two-hole} to the resulting matrix.
  The total weight of its perfect matchings is $Z_{x,s}$;
  its cofactors deleting row $x$ are unchanged;
  and its cofactor deleting row $y\neq x$ and column $a$ is $Z^{(2)}_{x,y;s,a}$.
  The $y=x$ summand vanishes since $v\neq s$.
  These substitutions give the result.
  Zero activities in this auxiliary matrix are justified by continuity.
\end{proof}

\subsection{Transition capacities}
\label{subsubsec:JSV-capacity-target}

We now bound the transition capacities $Q_{\mathrm{JSV}}$ and $Q_{\mathrm{FAR}}$ defined above.

For the JSV chain, use the exact hole weights $w_i^*$ and omit the factor $n^2$ on perfect matchings to define $\omega^*$ by
\begin{align*}
  \omega^*(M):=
  \begin{cases}
    \lambda_i(M),&M\in\mathcal P,\\
    w_i^*(u,v)\lambda_i(M),&M\in\mathcal N(u,v),
  \end{cases}
\end{align*}
and, for an allowed transition $T=\{M,M'\}$, write $m_T:=\min\{\omega^*(M),\omega^*(M')\}$.
The JSV chain's normalizing constant satisfies
\begin{align*}
  \mathcal Z_{\mathrm J}
  =
  \left(n^2+\sum_{u,v}\frac{w_i(u,v)}{w_i^*(u,v)}\right)
  Z_\emptyset
  \leq
  (1+\alpha)n^2Z_\emptyset.
\end{align*}
The Metropolis rule and $\alpha$-roughness therefore give
\begin{align}
  Q_{\mathrm{JSV}}(T)
  \geq
  \frac{m_T}{4\alpha(1+\alpha)n^3Z_\emptyset}.
  \label{eq:boosted-JSV-capacity}
\end{align}

We also record the capacities of the FAR chain.
Recall that $\pi_B$ uses unnormalized weights $n^2B(M)$ for $M\in\mathcal P$ and $W(u,v)B(M)$ for $M\in\mathcal N(u,v)$,
with normalizing constant $\mathcal Z$.
For a slide from $N\in\mathcal N(u,v)$ to $N'=N\setminus\{(y,a)\}\cup\{(u,a)\}\in\mathcal N(y,v)$ fixing the right hole,
the FAR transition rule gives
\begin{align}
  Q_{\mathrm{FAR}}(\{N,N'\})
  =
  \frac{B(N)B(u,a)s_v}{2\mathcal Z},
  \label{eq:FAR-row-slide-capacity}
\end{align}
where $s_v=\sum_xB(x,v)$.
For a slide from $N\in\mathcal N(u,v)$ to $N'=N\setminus\{(y,a)\}\cup\{(y,v)\}\in\mathcal N(u,a)$ fixing the left hole,
\begin{align}
  Q_{\mathrm{FAR}}(\{N,N'\})
  =
  \frac{B(N)B(y,v)}{2\mathcal Z}.
  \label{eq:FAR-column-slide-capacity}
\end{align}
Finally, if $M=N\cup\{(u,v)\}\in\mathcal P$, then
\begin{align}
  Q_{\mathrm{FAR}}(\{N,M\})
  =
  \frac{nB(M)}{2\mathcal Z}.
  \label{eq:FAR-boundary-capacity}
\end{align}

\subsection{Coupling and routing from prescribed source holes}
\label{subsec:augmenting-coupling}

Fix arbitrary source holes $(u_0,v_0)\in V_1\times V_2$.
For the FAR chain, the anchor $r$ and the matrix $B$ remain fixed throughout the construction.
We construct a weighted routing family for the coupling $\kappa_{u_0,v_0}$,
sending flow from $\mathcal N(u_0,v_0)$ to $\mathcal P$.
Independently sample $I\sim\mu_{u_0,v_0}$ and $F\sim\mu_\emptyset$,
and obtain $J$ by switching the augmenting path $\mathcal A$ in $I\oplus F$, that is, $J:=I\oplus\mathcal A$.
By \cref{lem:augmenting-path-coupling}, $(I,J)$ has distribution $\kappa_{u_0,v_0}$.

For each pair $(I,F)$, send $c_{I,F}:=\lambda_i(I)\lambda_i(F)/(Z_{u_0,v_0}Z_\emptyset)$ units of flow from $I$ to the corresponding $J$;
this is the probability of sampling $(I,F)$.
Starting at $I$, process $\mathcal A$ by slides from its two ends,
choosing the next end by the random rule in \cref{lem:two-ended-energy-overview}.
After all slides, add the remaining edge to reach $J$.
Write $\theta^{\mathrm{path}}_{I,F}$ for the resulting routing.
Taking all pairs $(I,F)\in\mathcal N(u_0,v_0)\times\mathcal P$ gives a weighted routing family $\mathcal F_{u_0,v_0}$ with endpoint demand $\kappa_{u_0,v_0}$.

\paragraph{Routing energy.}
Label the augmenting path as $a_0=u_0,b_0,a_1,b_1,\ldots,a_k,b_k=v_0$.
Its edges are $f_j=(a_j,b_j)\in F$ for $0\leq j\leq k$ and $e_j=(a_{j+1},b_j)\in I$ for $0\leq j<k$.
After $p$ slides from the left and $q$ from the right, the holes are $a_p,b_{k-q}$.
When $p+q=k$, add the remaining edge to obtain $J$.
Since $k\leq n-1$, \cref{lem:two-ended-energy-overview} gives $\En_{\max}(\mathcal F_{u_0,v_0})\leq1+\log n$.

\paragraph{Total weight of routings that can use a transition.}
For an undirected transition $T$, sum the full weights of all pairs $(I,F)$ whose routing can use $T$,
defining
\begin{align*}
  L_T^{u_0,v_0}
  :=
  \sum_{\substack{
      I\in\mathcal N(u_0,v_0),\,F\in\mathcal P\\
      T\in\supp(\theta^{\mathrm{path}}_{I,F})}}
  c_{I,F}.
\end{align*}

\begin{lemma}
  \label{lem:two-ended-load}
  For every source hole pattern $(u_0,v_0)$ and transition $T=\{M,M'\}$,
  \begin{align}
    L_T^{u_0,v_0}
    \leq
    C\frac{m_T}{Z_\emptyset}.
    \label{eq:two-ended-load}
  \end{align}
  Moreover, for the FAR chain,
  \begin{align}
    L_T^{u_0,v_0}
    \leq
    C_\alpha
    \frac{\mathcal Z}{B(\mathcal N(u_0,v_0))}
    Q_{\mathrm{FAR}}(T).
    \label{eq:FAR-two-ended-load}
  \end{align}
\end{lemma}

\begin{proof}
  For a terminal addition from $N$ to $M\in\mathcal P$, the complement $K=(I\uplus F)\setminus M$ lies in $\mathcal N(u_0,v_0)$ and determines the source pair uniquely.
  Hence
  \begin{align*}
    L_{\{N,M\}}^{u_0,v_0}
    \leq
    \frac{\lambda_i(M)}{Z_{u_0,v_0}Z_\emptyset}
    \sum_{K\in\mathcal N(u_0,v_0)}\lambda_i(K)
    =
    \frac{\lambda_i(M)}{Z_\emptyset}.
  \end{align*}
  By \cref{eq:cofactor-ineq}, $m_{\{N,M\}}=\lambda_i(M)$, proving \cref{eq:two-ended-load} for transitions between a near-perfect and a perfect matching.

  For the FAR chain, the rescaling from $\lambda_i$ to $B$ is constant within each hole pattern and within $\mathcal P$.
  Thus the same demand is $c_{I,F}=B(I)B(F)/(B(\mathcal N(u_0,v_0))B(\mathcal P))$.
  The preceding count and \cref{eq:FAR-boundary-capacity} give
  \begin{align*}
    L_{\{N,M\}}^{u_0,v_0}
    &\leq \frac{B(M)}{B(\mathcal P)}
    = \frac{2\mathcal Z}{nB(\mathcal P)}
    Q_{\mathrm{FAR}}(\{N,M\})\\
    &\leq
    C_\alpha
    \frac{\mathcal Z}{B(\mathcal N(u_0,v_0))}
    Q_{\mathrm{FAR}}(\{N,M\}),
  \end{align*}
  where the last inequality follows from $B(\mathcal N(u_0,v_0))\leq2\alpha^2B(\mathcal P)$.

  Now fix a directed slide $N\to N'$ fixing the right hole $v$,
  where $N\in\mathcal N(u,v)$ and $N'=N\setminus\{(y',a)\}\cup\{(u,a)\}\in\mathcal N(y',v)$.
  Put $T:=\{N,N'\}$, and let $L_{N\to N'}^{u_0,v_0}$ denote the total weight of pairs whose routing can use this directed slide.
  For such a pair $(I,F)$, let $(y,v)$ be the edge of $F$ incident to the other hole $v$ in the remaining augmenting path.
  Removing $N$ and the two edges $(u,a),(y,v)$ from $I\uplus F$ gives
  \begin{align*}
    I\uplus F
    &=N\uplus K\uplus\{(u,a),(y,v)\},
    \qquad K\in\mathcal N^{(2)}(u_0,y;v_0,a),\\
    \lambda_i(I)\lambda_i(F)
    &=\lambda_i(N)\lambda_i(K)\lambda_i(u,a)\lambda_i(y,v).
  \end{align*}
  Given $N,N',y,K$, we can uniquely recover $I$ and $F$.
  Therefore
  \begin{align}
    L_{N\to N'}^{u_0,v_0}
    \leq
    \frac{\lambda_i(N)\lambda_i(u,a)}
    {Z_{u_0,v_0}Z_\emptyset}
    \sum_y\lambda_i(y,v)Z^{(2)}_{u_0,y;v_0,a}.
    \label{eq:two-ended-directed-load}
  \end{align}

  Using \cref{eq:one-row-four-hole}, with the case $v=v_0$ following directly by cofactor expansion, gives
  \begin{align*}
    \frac1{Z_{u_0,v_0}}
    \sum_y\lambda_i(y,v)Z^{(2)}_{u_0,y;v_0,a}
    \leq2\frac{Z_{u_0,a}}{Z_{u_0,v}},\\
    L_{N\to N'}^{u_0,v_0}
    \leq2\frac{\lambda_i(N)\lambda_i(u,a)Z_{u_0,a}}
    {Z_\emptyset Z_{u_0,v}}.
  \end{align*}
  For $u\neq u_0$, applying \cref{eq:transposed-two-hole} to the transposed activity matrix and retaining only the summand indexed by $a$ gives $Z_{u,v}\lambda_i(u,a)Z_{u_0,a}\leq2Z_{u_0,v}Z_\emptyset$.
  For $u=u_0$, the same inequality follows from \cref{eq:cofactor-ineq}.
  Hence $L_{N\to N'}^{u_0,v_0}\leq4\omega^*(N)/Z_\emptyset$.
  The identity $\lambda_i(N)\lambda_i(u,a)=\lambda_i(N')\lambda_i(y',a)$ allows us to apply the same estimate with $y'$ in place of $u$,
  so this same directed load is also at most $4\omega^*(N')/Z_\emptyset$.
  Applying both bounds to the reverse slide and adding gives $L_T^{u_0,v_0}\leq8m_T/Z_\emptyset$,
  proving \cref{eq:two-ended-load}.
  The argument for slides fixing the left hole is symmetric.

  It remains to establish the FAR comparison.
  Applying \cref{eq:four-hole-switching} to $B$ and then \cref{lem:burnin-balance} gives $B(\mathcal N^{(2)}(x,y;s,a))\leq8\alpha^4B(\mathcal P)$ for every four-hole pattern.
  Since the demands can be written using $B$ as above, \cref{eq:two-ended-directed-load} applies with $B$ in place of $\lambda_i$.
  Inserting the four-hole bound gives
  \begin{align*}
    L_{N\to N'}^{u_0,v_0}
    \leq
    C_\alpha
    \frac{B(N)B(u,a)s_v}
    {B(\mathcal N(u_0,v_0))},
  \end{align*}
  where $s_v=\sum_yB(y,v)$.
  The reverse directed load satisfies the same bound, since $B(N)B(u,a)=B(N')B(y',a)$.
  Together with \cref{eq:FAR-row-slide-capacity}, this proves \cref{eq:FAR-two-ended-load} for slides fixing the right hole.
  For a slide fixing the left hole, the same count uses $\sum_{b\in V_2}B(u,b)=1$ in place of $s_v$ and gives
  \begin{align*}
    L_T^{u_0,v_0}
    \leq
    C_\alpha
    \frac{B(N)B(y,v)}
    {B(\mathcal N(u_0,v_0))}.
  \end{align*}
  Now \cref{eq:FAR-column-slide-capacity} gives the required capacity comparison.
  Together with the final edge additions treated above, this establishes \cref{eq:FAR-two-ended-load}.
\end{proof}

\subsection{Flow bounds for prescribed source holes}
\label{subsec:JSV-combining-bounds}

We prove \cref{lem:JSV-flow-summary,lem:FAR-flow-summary}.

\begin{proof}[Proof of \cref{lem:JSV-flow-summary,lem:FAR-flow-summary}]
  The routing family $\mathcal F_{u_0,v_0}$ has endpoint demand $\kappa_{u_0,v_0}$ and maximum energy at most $1+\log n$.
  For the JSV chain, \cref{eq:two-ended-load,eq:boosted-JSV-capacity} give support congestion $O_\alpha(n^3)$.
  By \cref{eq:energy-factorization}, we therefore have $\mathcal R_{\En}(\mathcal F_{u_0,v_0})=O_\alpha(n^3\log n)$,
  which proves \cref{lem:JSV-flow-summary}.

  For the FAR chain, \cref{eq:FAR-two-ended-load} gives support congestion at most $C_\alpha\mathcal Z/B(\mathcal N(u_0,v_0))$.
  Combining this with the same energy bound and \cref{lem:energy-MCF} yields \cref{eq:FAR-fixed-pattern-flow},
  proving \cref{lem:FAR-flow-summary}.
\end{proof}

\subsection{Flow comparison for perfect matchings}
\label{subsec:perfect-cycle-comparison}

We now prove \cref{lem:perfect-matching-comparison}, which controls variation among perfect matchings for the JSV chain.

\begin{proof}[Proof of \cref{lem:perfect-matching-comparison}]
  Let $I,F\sim\mu_\emptyset$ independently, and let $\mathcal C(I,F)$ be the alternating cycles in $I\oplus F$.
  Conditioned on $I\uplus F$, each cycle has two equally likely assignments of its alternating edges to $I$ and $F$,
  independently across cycles.
  Encode these choices by independent fair bits.
  For a function $g$ of these bits, variance tensorization gives
  \begin{align*}
    \E{(g(x)-g(\bar x))^2}
    \leq4\Var{g(x)}
    \leq\sum_j\E{(g(x)-g(x^{(j)}))^2},
  \end{align*}
  where $\bar x$ complements every bit and $x^{(j)}$ flips only bit $j$.
  Applying this to the matching specified by these assignments, and averaging over $I\uplus F$, gives
  \begin{align}
    \Var[\mu_\emptyset]{f}
    \leq
    \frac12\E[I,F\sim\mu_\emptyset]{
    \sum_{C\in\mathcal C(I,F)}
    (f(I)-f(I\oplus C))^2}.
    \label{eq:perfect-cycle-variance-reduction}
  \end{align}

  For each triple $(I,F,C)$, route from $I$ to $I\oplus C$ with coefficient $\lambda_i(I)\lambda_i(F)/(2Z_\emptyset^2)$.
  Following \cref{lem:cycle-energy-overview}, delete a uniformly chosen edge of $I$ on $C$,
  slide the left hole around $C$ while keeping the right hole fixed, and add the remaining edge of $F$ on $C$.
  Each routing has energy at most two.

  For each transition, we bound the total weight of all triples whose routing can use it.
  Fix a directed slide $N\to N'$ with $N\in\mathcal N(u,v)$ and $N'\in\mathcal N(u',v)$,
  where $N'=N\setminus\{(u',a)\}\cup\{(u,a)\}$.
  For a triple $(I,F,C)$ whose routing can use $N\to N'$, let $(y,v)$ be the edge of $F$ incident to the fixed hole $v$ on the cycle $C$ being switched.
  Removing $N$ and the two edges $(u,a),(y,v)$ from $I\uplus F$ leaves $K\in\mathcal N(y,a)$ satisfying $I\uplus F=N\uplus K\uplus\{(u,a),(y,v)\}$.
  This decomposition preserves the product of edge weights.
  Given $y,K$, the inserted edge identifies $C$, and $N$ determines which edges belong to $I$ on all other components.
  There are at most two possible assignments of alternating edges to $I$ and $F$ on $C$.
  Thus the total weight of triples whose routing can use $N\to N'$ satisfies
  \begin{align*}
    L_{N\to N'}
    \leq
    C_0\frac{\lambda_i(N)\lambda_i(u,a)}{Z_\emptyset^2}
    \sum_y\lambda_i(y,v)Z_{y,a}.
  \end{align*}

  By \cref{eq:transposed-two-hole,eq:cofactor-ineq},
  \begin{align*}
    Z_{u,v}\lambda_i(u,a)\sum_y\lambda_i(y,v)Z_{y,a}
    \leq
    2\lambda_i(u,a)Z_{u,a}Z_\emptyset
    \leq2Z_\emptyset^2.
  \end{align*}
  Consequently, $L_{N\to N'}\leq C_0\lambda_i(N)/Z_{u,v} =C_0\omega^*(N)/Z_\emptyset$.
  Using $\lambda_i(N)\lambda_i(u,a)=\lambda_i(N')\lambda_i(u',a)$ and applying the same inequalities with $u'$ in place of $u$ also bounds this load by $C_0\omega^*(N')/Z_\emptyset$.
  Taking the smaller bound and including the reverse direction gives $L_T=O(m_T/Z_\emptyset)$ for $T=\{N,N'\}$.

  For a transition between a perfect matching $M$ and a near-perfect matching,
  the complement $K=(I\uplus F)\setminus M$ is perfect.
  The graph edge added or deleted by this transition identifies $C$.
  For a fixed transition and complement $K$, at most a constant number of triples $(I,F,C)$ give these data.
  Each possible complement contributes at most $C_0\lambda_i(M)\lambda_i(K)/Z_\emptyset^2$ to the total weight of these triples.
  Summing over $K\in\mathcal P$ gives $C_0\lambda_i(M)/Z_\emptyset$,
  since $\sum_{K\in\mathcal P}\lambda_i(K)=Z_\emptyset$.
  By \cref{eq:cofactor-ineq}, $\lambda_i(M)=m_T$, so this bound is $C_0m_T/Z_\emptyset$.

  Thus the total weight of triples whose routing can use any transition $T$ is $O(m_T/Z_\emptyset)$.
  Together with \cref{eq:boosted-JSV-capacity}, this gives support congestion $O_\alpha(n^3)$.
  Since each routing has energy at most two, \cref{lem:energy-MCF,eq:perfect-cycle-variance-reduction} yield
  \begin{equation*}
    \Var[\mu_\emptyset]{f}
    \leq C_\alpha n^3
    \Dir_{P_{\mathrm{JSV}}}(f,f).\qedhere
  \end{equation*}
\end{proof}

\section{Fast perfect-matching sampling}
\label{sec:fast-sampling-proofs}

For fixed positive edge activities $\lambda_i$, our goal is to sample approximately from $\mu_i$ using only the supplied rough anchor weights,
with the running time stated in \cref{lem:fast-burnin}.
The proofs of \cref{lem:fast-burnin,lem:fast-burnin-general} apply to arbitrary strictly positive edge activities,
given $\alpha$-rough anchor weights, with the same guarantees.
We first establish bounds for the stationary distributions and choose the assignment activities $a_0,\ldots,a_m$ while keeping $\lambda_i$ fixed.
We then prove the particle-chain spectral-gap criterion, apply it to assignment Glauber dynamics,
and bound the cost of simulating that dynamics.
The sampler's guarantee for any error tolerance is proved at the end of the section.
The outer temperature schedule and permanent estimation are handled separately in \cref{sec:missing}.

\subsection{Bounds for the stationary distributions and the activity schedule}
\label{sub:burnin-schedule-proofs}

Throughout this subsection, fix positive edge activities $\lambda_i$, the roughness parameter $\alpha$,
an anchor $r$, and the supplied $\alpha$-rough weights.
Recall that $J=J_\alpha n^2$, where $J_\alpha\geq8\alpha^2$.

\paragraph{Near-perfect matching weights after rescaling.}

\begin{proof}[Proof of \cref{lem:burnin-balance}]
  Comparing the factors contributed by perfect and near-perfect matchings under the rescaling gives
  \begin{align*}
    \frac{B(\mathcal N(u,v))}{B(\mathcal P)}
    =
    \frac{\lambda_i(\mathcal N(u,v))}
    {\lambda_i(\mathcal P)}
    \,w_i(r,v)
    \sum_{b\in V_2}\frac{\lambda_i(u,b)}{w_i(r,b)}.
  \end{align*}
  Using $\alpha$-roughness and $w_i^*(r,v)=\lambda_i(\mathcal P)/ \lambda_i(\mathcal N(r,v))$, this is at most
  \begin{align*}
    \alpha^2
    \frac{
      \lambda_i(\mathcal N(u,v))
      \sum_b\lambda_i(u,b)\lambda_i(\mathcal N(r,b))
    }{
      \lambda_i(\mathcal P)\lambda_i(\mathcal N(r,v))
    }.
  \end{align*}
  The switching inequality in \cref{lem:matching-switching}, applied also after interchanging the two vertex classes,
  bounds the numerator by $2\lambda_i(\mathcal P)\lambda_i(\mathcal N(r,v))$.
  This proves the upper bound.

  For $u=r$, the same identity, together with $w_i(r,v)\geq w_i^*(r,v)/\alpha$ and $1/w_i(r,b)\geq1/(\alpha w_i^*(r,b))$,
  gives
  \begin{align*}
    \frac{B(\mathcal N(r,v))}{B(\mathcal P)}
    \geq
    \alpha^{-2}
    \frac{\sum_b
      \lambda_i(r,b)\lambda_i(\mathcal N(r,b))}
    {\lambda_i(\mathcal P)}
    =
    \alpha^{-2}.
  \end{align*}
  The last equality follows by partitioning perfect matchings according to the partner of $r$.
\end{proof}

\paragraph{Partition functions.}
Write
\begin{align*}
  F(a):=\sum_{M\text{ a matching}}a^{|M|}B(M),
  \qquad
  F_{\mathsf{aux}}(a)
  =
  \sum_{\sigma\in\Omega_{\mathrm{aux}}}
  \frac{a^{n-\mathsf{un}(\sigma)}}{(1+J)^{\mathsf{cr}(\sigma)}}
  \prod_{u:\sigma_u\neq\ast}B(u,\sigma_u).
\end{align*}
\begin{proposition}
  \label{prop:burnin-factorization}
  There exist positive numbers $\rho_1,\ldots,\rho_n$, depending only on $B$ and independent of $a$, such that
  \begin{align*}
    F(a)=B(\mathcal P)\prod_{j=1}^n(a+\rho_j)
    \qquad\text{for every }a.
  \end{align*}
\end{proposition}

\begin{proof}
  Heilmann and Lieb~\cite[Theorem~4.2]{HL72} show that the weighted matching polynomial on our graph with $2n$ vertices,
  \begin{align*}
    H_B(x):=\sum_{M\text{ a matching}}(-1)^{|M|}B(M)x^{2n-2|M|},
  \end{align*}
  has only real roots.
  Only even powers of $x$ occur, so $H_B(-x)=H_B(x)$.
  Thus, whenever $r$ is a root, $-r$ is also a root with the same multiplicity.
  Moreover, $H_B(0)=(-1)^nB(\mathcal P)\neq0$, so zero is not a root.
  Consequently, the roots form pairs consisting of one positive and one negative root.
  The coefficient of $x^{2n}$ is $1$, so $H_B$ has degree $2n$ and hence exactly $n$ such pairs,
  counted with multiplicity.
  We can therefore list its roots as $\pm r_1,\ldots,\pm r_n$ with $r_j>0$.
  For $x\neq0$, the change of variables gives
  \begin{align*}
    F(-x^{-2})=x^{-2n}H_B(x)
    =\prod_{j=1}^n(1-r_j^2x^{-2}).
  \end{align*}
  Thus $F(a)=\prod_{j=1}^n(1+r_j^2a)$ as a polynomial identity.
  Its leading coefficient is $B(\mathcal P)=\prod_{j=1}^nr_j^2$,
  which proves the claimed factorization with $\rho_j=r_j^{-2}>0$.
\end{proof}

Comparing the coefficients of $a^{n-1}$ gives
\begin{align}
  \sum_{j=1}^n\rho_j
  =
  \frac{\sum_{u,v}B(\mathcal N(u,v))}{B(\mathcal P)}
  \leq2\alpha^2n^2
  \leq J/4.
  \label{eq:burnin-root-sum}
\end{align}

\begin{lemma}
  \label{lem:burnin-partition-comparison}
  For every $a>0$,
  \begin{align*}
    F(a)\leq F_{\mathsf{aux}}(a)\leq\frac43F(a).
  \end{align*}
\end{lemma}

\begin{proof}[Proof of \cref{lem:burnin-partition-comparison}]
  Ordinary matchings contribute exactly $F(a)$ to $F_{\mathsf{aux}}(a)$, proving the lower bound.
  For the upper bound, order the left vertices and retain only the first assignment at each occupied right vertex.
  This maps each auxiliary configuration to a matching $M$, removing exactly $\mathsf{cr}(\sigma)$ assignments.
  Each removed assignment contributes $aB(u,v)/(1+J)$.

  For a fixed retained matching $M$, allowing each unmatched left vertex to be either unassigned or assigned to any right vertex only increases the total contribution.
  Since the rows of $B$ sum to one, this gives
  \begin{align*}
    F_{\mathsf{aux}}(a)
    \leq
    \sum_M a^{|M|}B(M)
    \left(1+\frac{a}{1+J}\right)^{n-|M|}
    =
    \left(1+\frac{a}{1+J}\right)^n
    F\left(\frac{a}{1+a/(1+J)}\right).
  \end{align*}
  The factorization of $F$ and \cref{eq:burnin-root-sum} therefore imply
  \begin{equation*}
    \frac{F_{\mathsf{aux}}(a)}{F(a)}
    \leq
    \prod_j
    \left(1+\frac{a\rho_j}{(1+J)(a+\rho_j)}\right)
    \leq
    \exp\left(\frac{\sum_j\rho_j}{1+J}\right)
    \leq e^{1/4}<\frac43.\qedhere
  \end{equation*}
\end{proof}

\paragraph{A mixed-derivative inequality.}
We use the following polynomial inequality to bound the total weights of matchings after deleting row $r$.
For a vector $b$, write
\begin{align} \label{eq:def-dd}
  D_b:=\sum_v b_v\partial_{z_v}.
\end{align}

\begin{definition}
  \label{def:burnin-hyperbolicity}
  A homogeneous polynomial $p:\mathbb R^m\to\mathbb R$ is \emph{hyperbolic with respect to $e$} if $p(e)>0$ and $s\mapsto p(z+se)$ has only real roots for every $z\in\mathbb R^m$.

  The \emph{hyperbolicity cone associated with $e$} is the connected component of $\{z:p(z)\neq0\}$ containing $e$~\cite[Proposition~1]{Ren06}.
  Equivalently, it consists of the points reachable from $e$ by a continuous path along which $p$ never vanishes.

  We say that $p$ is \emph{hyperbolic with cone $\Gamma$} if $p$ is hyperbolic with respect to some $e$ and $\Gamma$ is the hyperbolicity cone associated with $e$.
\end{definition}

Here, a \emph{cone} means a set $C$ such that $tx\in C$ for every $x\in C$ and $t>0$.
Since $p$ is homogeneous of degree $d$, we have $p(tx)=t^d p(x)$ for every $x$ and $t>0$.
Thus the positive ray through any $x\in\Gamma$ stays in $\{z:p(z)\neq0\}$,
and hence in the same connected component $\Gamma$.

We use the following standard facts in the forms needed below.

\begin{proposition}
  \label{prop:burnin-hyperbolicity-facts}
  Let $p$ be homogeneous of degree $d\geq1$ and hyperbolic with cone $\Gamma$.
  \begin{enumerate}[label=(\roman*),ref=\roman*,leftmargin=*]
    \item\label[propositionitem]{prop:burnin-hyperbolicity-facts-roots}
      \emph{Positive roots.}
      For $x,b\in\Gamma$, we have $p(x)>0$, and all $d$ roots of $s\mapsto p(x-sb)$ are positive,
      counted with multiplicity \cite[Section~2, Proposition~1 and Theorem~3]{Ren06}.

    \item\label[propositionitem]{prop:burnin-hyperbolicity-facts-derivative}
      \emph{One directional derivative.}
      If $d\geq2$ and $b\in\Gamma$, then $q:=D_b p$ is homogeneous of degree $d-1$ and hyperbolic with respect to $b$.
      Its hyperbolicity cone $\Gamma_q$ contains $\Gamma$;
      in particular, $q(x)>0$ for every $x\in\Gamma$ \cite[Theorem~3 and Section~4]{Ren06}.

    \item\label[propositionitem]{prop:burnin-hyperbolicity-facts-successive}
      \emph{Successive directional derivatives.}
      For $1\leq k\leq d-1$ and $b_1,\ldots,b_k\in\Gamma$, the polynomial $p_k:=D_{b_1}\cdots D_{b_k}p$ is homogeneous of degree $d-k$ and hyperbolic with a cone containing $\Gamma$.
      Consequently, $p_k(x)>0$ for every $x\in\Gamma$ \cite[Sections~4--5, especially the discussion before Proposition~16]{Ren06}.
  \end{enumerate}
\end{proposition}

Renegar \cite{Ren06} writes $\Lambda_{++}$ for our cone $\Gamma$.
For part~\ref{prop:burnin-hyperbolicity-facts-roots}, his cone is defined by positivity of the characteristic roots;
Proposition~1 identifies it with the connected component in \cref{def:burnin-hyperbolicity},
and Theorem~3 allows the reference direction $e$ to be replaced by $b$.
For part~\ref{prop:burnin-hyperbolicity-facts-derivative}, make this change of direction and use the derivative-cone inclusion at the start of Section~4.
Section~5 explicitly allows successive derivatives in different directions,
provided each direction lies in the current cone.
Here this condition holds because every $b_j$ lies in the original cone $\Gamma$,
which is contained in all the subsequent cones.

\begin{lemma}
  \label{lem:burnin-mixed-derivative}
  Let $p$ be homogeneous of degree $d\geq1$ and hyperbolic with respect to $e$,
  and let $\Gamma$ be its associated hyperbolicity cone.
  For $x,b_1,\ldots,b_{d-1}\in\Gamma$, write $L:=D_{b_1}\cdots D_{b_{d-1}}p$.
  Then
  \begin{align*}
    d!\,p(x)\prod_{j=1}^{d-1}L(b_j)\leq L(x)^d.
  \end{align*}
\end{lemma}

\begin{proof}
  We induct on $d$;
  the case $d=1$ is equality.
  For $d\geq2$, put $q=D_{b_1}p$, which is homogeneous of degree $d-1$.
  By \cref{prop:burnin-hyperbolicity-facts}\cref{prop:burnin-hyperbolicity-facts-derivative},
  $q$ is hyperbolic with respect to $b_1$, and its associated cone $\Gamma_q$ contains $\Gamma$.

  By \cref{prop:burnin-hyperbolicity-facts}\cref{prop:burnin-hyperbolicity-facts-roots},
  the degree-$d$ polynomial $s\mapsto p(x-sb_1)$ has only positive roots;
  let $\tau$ be the smallest one.
  The path $s\mapsto x-sb_1$, $0\leq s<\tau$, starts at $x\in\Gamma$, and $p$ never vanishes along it.
  It therefore stays in the same connected component $\Gamma$ of $\{p\neq0\}$.
  In particular, $x-sb_1\in\Gamma\subseteq\Gamma_q$.
  The directions $b_2,\ldots,b_{d-1}$ also belong to $\Gamma_q$, and $D_{b_2}\cdots D_{b_{d-1}}q=L$.
  Thus the induction hypothesis applies to $q$ at $x-sb_1$ with these directions, giving
  \begin{align*}
    q(x-sb_1)
    \leq
    \frac{L(x-sb_1)^{d-1}}
    {(d-1)!\prod_{j=2}^{d-1}L(b_j)}.
  \end{align*}
  The polynomial $L$ is linear, and \cref{prop:burnin-hyperbolicity-facts}\cref{prop:burnin-hyperbolicity-facts-successive} makes it positive on $\Gamma$.
  In particular, each $L(b_j)>0$, and $L(x-sb_1)=L(x)-sL(b_1)>0$ for $0\leq s<\tau$.
  Letting $s\uparrow\tau$ gives $\tau\leq L(x)/L(b_1)$.
  We now integrate, using $p(x-\tau b_1)=0$ and $\frac{d}{ds}p(x-sb_1)=-q(x-sb_1)$.
  Since $L(x)-sL(b_1)\geq0$ for $0\leq s\leq L(x)/L(b_1)$, we can extend the integration interval to obtain
  \begin{equation*}
    \begin{aligned}[b]
      p(x)
      &=\int_0^\tau q(x-sb_1)\,ds\\
      &\leq
      \frac{1}{(d-1)!\prod_{j=2}^{d-1}L(b_j)}
      \int_0^{L(x)/L(b_1)}
      \bigl(L(x)-sL(b_1)\bigr)^{d-1}\,ds\\
      &=
      \frac{L(x)^d}{d!\prod_{j=1}^{d-1}L(b_j)}.
    \end{aligned}
    \qedhere
  \end{equation*}
\end{proof}

\paragraph{The unmatched-vertex estimate.}
Recall that $F(a)=\sum_{M\text{ a matching}}a^{|M|}B(M)$ is the partition function for ordinary matchings,
and $\nu_a$ is the distribution on auxiliary assignments from \cref{eq:burnin-auxiliary-distribution},
given by
\begin{align*}
  \nu_a(\sigma)
  =\frac{a^{n-\mathsf{un}(\sigma)}}
  {F_{\mathsf{aux}}(a)(1+J)^{\mathsf{cr}(\sigma)}}
  \prod_{u:\sigma_u\neq\ast}B(u,\sigma_u).
\end{align*}
Here $F_{\mathsf{aux}}(a)$ is the normalizing factor, $\mathsf{un}(\sigma)$ counts unassigned left vertices,
and $\mathsf{cr}(\sigma)$ counts assignments beyond the first at each right vertex.

\begin{lemma}
  \label{lem:burnin-matching-estimates}
  There is a constant $A_\alpha\geq1$, depending only on $\alpha$, such that, for every $a>0$,
  a matching $M$ drawn with probability $a^{|M|}B(M)/F(a)$ satisfies
  \begin{align}
    \E{n-|M|}\leq
    A_\alpha\left(1+\min\left\{n,\frac n{\sqrt a}\right\}\right).
    \label{eq:burnin-unmatched-mean}
  \end{align}
  Moreover, for every $v\geq0$,
  \begin{align}
    \nu_a(\mathsf{un}\geq v)
    \leq\frac43\exp\left[
      A_\alpha\left(1+\min\left\{n,\frac n{\sqrt a}\right\}\right)
      \right]2^{-v}.
    \label{eq:burnin-unmatched-tail}
  \end{align}
\end{lemma}

\begin{proof}[Proof of \cref{lem:burnin-matching-estimates}]
  The case $n=1$ is immediate, so assume $n\geq2$.
  Put $K:=2\alpha^2$.
  We first bound the weights of matchings after deleting row $r$,
  then deduce the expectation bound and the auxiliary-assignment tail.

  \emph{1. Matching weights after deleting row $r$.}
  Let $s_q$ be the total $B$-weight of matchings that avoid row $r$ and have $n-1-q$ edges,
  for $0\leq q\leq n-1$.
  Write $b_v:=B(\mathcal N(r,v))$, so $s_0=\sum_v b_v$.
  By \cref{lem:burnin-balance},
  \begin{align*}
    B(\mathcal P)/K\leq b_v\leq K B(\mathcal P).
  \end{align*}

  Fix $I\subseteq V_1\setminus\{r\}$ with $|I|=q\geq1$, and put $U:=V_1\setminus(I\cup\{r\})$.
  Writing $B_u$ for row $u$ of $B$, define
  \begin{align*}
    p_I(z):=
    \left(\prod_{u\in U}D_{B_u}\right)\prod_{v\in V_2}z_v.
  \end{align*}
  Each derivative chooses one column for its row.
  Choosing the same column twice gives zero, since each variable initially has degree one.
  By the definition of directional derivative (see \cref{eq:def-dd}),
  the surviving choices are therefore exactly the matchings covering the rows in $U$, so
  \begin{align*}
    p_I(z)=
    \sum_{\substack{f:U\to V_2\\f\text{ injective}}}
    \left(\prod_{u\in U}B(u,f(u))\right)
    \prod_{v\in V_2\setminus f(U)}z_v.
  \end{align*}
  Thus $p_I(\one)$ sums the weights of all matchings whose unmatched rows are exactly $I\cup\{r\}$,
  regardless of their unmatched columns.
  In particular,
  \begin{align*}
    s_q=
    \sum_{\substack{I\subseteq V_1\setminus\{r\}\\|I|=q}}
    p_I(\one).
  \end{align*}

  The product $\prod_vz_v$ is hyperbolic with cone the positive orthant,
  and every row vector $B_u$ lies in that orthant.
  By \cref{prop:burnin-hyperbolicity-facts}\cref{prop:burnin-hyperbolicity-facts-successive},
  successive differentiation preserves hyperbolicity and gives a cone containing the original orthant.
  Hence $p_I$ is hyperbolic of degree $q+1$, with a cone containing $\one$ and every $B_u$.
  When $U$ is empty, this assertion holds directly for the original product polynomial.

  Applying the remaining derivatives indexed by $I$ gives one derivative for every row except $r$.
  The remaining variable records the single unmatched column, so
  \begin{align*}
    L_I(z):=
    \left(\prod_{u\in I}D_{B_u}\right)p_I(z)
    =\sum_v b_vz_v.
  \end{align*}
  Consequently, $L_I(\one)=s_0$.
  Since each row of $B$ sums to one and $s_0\leq nK B(\mathcal P)$,
  \begin{align*}
    L_I(B_u)
    =\sum_v B(u,v)b_v
    \geq\frac{B(\mathcal P)}K
    \geq\frac{s_0}{nK^2}.
  \end{align*}
  Apply \cref{lem:burnin-mixed-derivative} to $p_I$, with degree $q+1$, point $x=\one$,
  and the $q$ directions $(B_u)_{u\in I}$.
  It gives
  \begin{align*}
    p_I(\one)
    \leq
    \frac{s_0^{q+1}}
    {(q+1)!\prod_{u\in I}L_I(B_u)}
    \leq s_0\frac{(nK^2)^q}{(q+1)!}.
  \end{align*}
  Summing over $I$, we obtain
  \begin{align}
    \frac{s_q}{s_0}
    \leq
    \binom{n-1}{q}\frac{(nK^2)^q}{(q+1)!}
    \leq
    \frac{[K^2n(n-1)]^q}{q!(q+1)!}.
    \label{eq:burnin-deleted-row-coefficients}
  \end{align}
  For $q=0$, the same bound holds with equality.

  \emph{2. The mean number of unmatched rows.}
  The partition function after deleting row $r$ is $F_{-r}(a)=\sum_{q=0}^{n-1}s_qa^{n-1-q}$.
  Define
  \begin{align*}
    R(a):=\frac{F_{-r}(a)}{s_0a^{n-1}}
    =\sum_{q=0}^{n-1}\frac{s_q}{s_0}a^{-q}.
  \end{align*}
  By \cref{eq:burnin-deleted-row-coefficients},
  \begin{align*}
    1\leq R(a)
    \leq
    \sum_{q\geq0}\frac{(Kn/\sqrt a)^{2q}}{(q!)^2}
    \leq
    \left(\sum_{q\geq0}\frac{(Kn/\sqrt a)^q}{q!}\right)^2
    =\exp\left(\frac{2Kn}{\sqrt a}\right).
  \end{align*}
  The second inequality keeps only the diagonal terms in the square of the nonnegative exponential series.

  Let $M_{-r}$ have probability proportional to $a^{|M_{-r}|}B(M_{-r})$ on the graph with row $r$ deleted,
  and put $Q:=(n-1)-|M_{-r}|$.
  Multiplying a matching's weight by $2^Q$ replaces $a$ by $a/2$ in the normalized polynomial $R$, so
  \begin{align*}
    \E{2^Q}
    =\frac{R(a/2)}{R(a)}
    \leq\exp\left(\frac{2\sqrt2Kn}{\sqrt a}\right).
  \end{align*}
  Jensen's inequality therefore gives
  \begin{align*}
    \E Q
    \leq\frac{\log\E{2^Q}}{\log2}
    \leq\frac{2\sqrt2Kn}{\sqrt a\log2}.
  \end{align*}

  To compare with the full graph, we use the following elementary coupling.
  At a fixed activity, the matching distributions on a weighted graph $G$ and $G-x$ admit a coupling $(M,M')$ with $|M|-|M'|\in\{0,1\}$.
  Induct on the number of vertices and first draw $M$ on $G$.
  If $x$ is unmatched, set $M'=M$.
  If $x$ is matched to $y$, then $M\setminus\{xy\}$ has the matching distribution on $G-x-y$.
  Apply induction to $G-x$ and its vertex $y$, using this matching as the smaller-graph marginal,
  to obtain $M'$.
  In each case $M'$ has the matching distribution on $G-x$, and the size difference differs by at most one.

  Applying this coupling with $x=r$ gives $n-|M|\leq1+Q$.
  Hence, writing $h(a):=\E{n-|M|}$,
  \begin{align*}
    h(a)
    \leq1+\E Q
    \leq1+\frac{2\sqrt2Kn}{\sqrt a\log2}.
  \end{align*}
  Together with $h(a)\leq n$, this proves \cref{eq:burnin-unmatched-mean} for a sufficiently large constant $A_\alpha$.

  \emph{3. The auxiliary-assignment tail.}
  By \cref{prop:burnin-factorization},
  \begin{align}
    h(a)
    =n-\frac{aF'(a)}{F(a)}
    =\sum_j\frac{\rho_j}{a+\rho_j}.
    \label{eq:burnin-deficit-identity}
  \end{align}
  The definition of $\nu_a$ and \cref{lem:burnin-partition-comparison} now give
  \begin{align*}
    \E[\nu_a]{2^{\mathsf{un}}}
    =2^n\frac{F_{\mathsf{aux}}(a/2)}{F_{\mathsf{aux}}(a)}
    \leq\frac43\,2^n\frac{F(a/2)}{F(a)}
    =\frac43\prod_j\left(1+\frac{\rho_j}{a+\rho_j}\right)
    \leq\frac43e^{h(a)}.
  \end{align*}
  Markov's inequality and the expectation bound therefore yield
  \begin{align*}
    \nu_a(\mathsf{un}\geq v)
    \leq\frac43e^{h(a)}2^{-v}
    \leq\frac43
    \exp\left[
      A_\alpha\left(1+\min\left\{n,\frac n{\sqrt a}\right\}\right)
    \right]2^{-v}. &\qedhere
  \end{align*}
\end{proof}

\paragraph{Construction and analysis of the schedule.}
\begin{proof}[Proof of \cref{lem:burnin-schedule}]
  Recall that $J=J_\alpha n^2$, where $J_\alpha\geq8\alpha^2$ is a fixed constant depending only on $\alpha$.
  We must construct, in $O(n^2)$ time, an increasing schedule $1/n=a_0<\cdots<a_m=J$.
  The point mass at the all-unassigned configuration $\sigma^\ast$ must be $3$-warm for $\nu_{a_0}$,
  and $\nu_{a_k}$ must be $3$-warm for $\nu_{a_{k+1}}$ for every $k<m$.
  We also need
  \begin{align*}
    \nu_J(\mathcal P)\geq\frac12,
    \qquad
    \sum_{k=0}^m(1+\sqrt{a_k})=O(n).
  \end{align*}
  The summed bound also gives $m+1=O(n)$, since each summand is at least one.
  We first find a sufficient bound on each relative activity increase,
  then choose the schedule and verify the remaining requirements.

  \emph{1. A sufficient condition for consecutive warm starts.}
  Put $K:=\max\{A_\alpha,J_\alpha\}$, where $A_\alpha$ is the constant in \cref{lem:burnin-matching-estimates}.
  For activities $a$ and $\gamma a$ with $\gamma>1$, we claim that $\nu_a$ is $3$-warm for $\nu_{\gamma a}$ whenever
  \begin{align*}
    (\gamma-1)^2K
    \left(1+\min\left\{n,\frac{n}{\sqrt a}\right\}\right)
    \leq\frac1{16}.
  \end{align*}
  Indeed, the definition of $\nu_a$ and \cref{lem:burnin-partition-comparison} give
  \begin{align*}
    \sum_{\sigma}
    \frac{\nu_a(\sigma)^2}{\nu_{\gamma a}(\sigma)}
    &=
    \frac{F_{\mathsf{aux}}(\gamma a)
      F_{\mathsf{aux}}(a/\gamma)}
    {F_{\mathsf{aux}}(a)^2}
    \leq
    \frac{16}{9}\frac{F(\gamma a)F(a/\gamma)}{F(a)^2}\\
    &\leq
    \frac{16}{9}
    \exp\left((\gamma-1)^2
    \sum_j\frac{\rho_j}{a+\rho_j}\right)
    \leq
    \frac{16}{9}e^{1/16}<3.
  \end{align*}
  For the second inequality, apply \cref{prop:burnin-factorization} and the identity
  \begin{align*}
    \frac{(\gamma a+\rho)(a/\gamma+\rho)}{(a+\rho)^2}
    =
    1+\frac{a\rho(\gamma-1)^2}{\gamma(a+\rho)^2}
    \leq
    \exp\left((\gamma-1)^2\frac{\rho}{a+\rho}\right).
  \end{align*}
  For the last inequality, \cref{eq:burnin-deficit-identity} identifies $\sum_j\rho_j/(a+\rho_j)$ with the expected number of unmatched rows in an ordinary matching.
  By \cref{eq:burnin-unmatched-mean}, this is at most $K(1+\min\{n,n/\sqrt a\})$.
  Our assumed bound on $\gamma-1$ therefore makes the exponent at most $1/16$, proving the claim.

  \emph{2. Choosing the schedule and checking its endpoints.}
  To satisfy the condition from Step 1, let $q(a)$ be the least power of two satisfying
  \begin{align*}
    q(a)^2
    \geq
    16K\left(1+\min\left\{n,\frac{n}{\sqrt a}\right\}\right).
  \end{align*}
  Starting at $a_0=1/n$, set
  \begin{align*}
    a_{k+1}:=
    \min\left\{J,\left(1+\frac1{q(a_k)}\right)a_k\right\}
  \end{align*}
  until $a_m=J$.
  For every step, the ratio $\gamma=a_{k+1}/a_k$ satisfies $0<\gamma-1\leq1/q(a_k)$,
  including the final step truncated at $J$.
  Thus Step 1 gives all the consecutive warm-start guarantees.

  For the initial warm start, omitting the factors $(1+J)^{-\mathsf{cr}(\sigma)}$ and using row normalization gives $F_{\mathsf{aux}}(1/n)\leq(1+1/n)^n\leq e$.
  Since $\sigma^\ast$ has unnormalized weight one,
  \begin{align*}
    \sum_\sigma
    \frac{\delta_{\sigma^\ast}(\sigma)^2}{\nu_{a_0}(\sigma)}
    =
    \frac1{\nu_{a_0}(\sigma^\ast)}
    =
    F_{\mathsf{aux}}(1/n)
    \leq e<3.
  \end{align*}
  At the terminal activity, \cref{eq:burnin-root-sum} and the choice of $J$ give $\sum_j\rho_j/J\leq2\alpha^2/J_\alpha\leq1/4$.
  Together with \cref{lem:burnin-partition-comparison,prop:burnin-factorization}, this yields
  \begin{align*}
    \nu_J(\mathcal P)
    =\frac{J^nB(\mathcal P)}{F_{\mathsf{aux}}(J)}
    \geq
    \frac34\prod_j(1+\rho_j/J)^{-1}
    \geq
    \frac34\exp\left(-\frac{\sum_j\rho_j}{J}\right)
    \geq
    \frac34e^{-1/4}>\frac12.
  \end{align*}

  \emph{3. Bounding the total stage cost.}
  It remains to bound $\sum_k(1+\sqrt{a_k})$, which will also bound the number of stages,
  since each summand is at least one.
  The schedule starts at $a_0=1/n$ and ends at $a_m=J=J_\alpha n^2$.
  Group its activities into intervals $[b,2b)$, where $b$ ranges over the powers of two satisfying $1/(2n)<b\leq J$.
  Since $q$ is nonincreasing, every step starting in such an interval multiplies the activity by at least $1+1/q(b)$,
  unless it is truncated at $J$.
  Consequently, each interval contains at most
  \begin{align*}
    O\left(1+\frac{\log2}{\log(1+1/q(b))}\right)=O(q(b))
  \end{align*}
  stages;
  the possible final truncated step adds at most one.

  For $b<1$, there are $O_K(\sqrt n)$ stages per interval, and each contributes $O(1)$ to $\sum_k(1+\sqrt{a_k})$.
  There are $O(\log(2n))$ such intervals that meet the schedule.
  For $1\leq b\leq J=O_K(n^2)$, $q(b)=O_K(\sqrt n\,b^{-1/4})$, and each stage contributes $O(\sqrt b)$.
  Here $b\leq J\leq Kn^2$ lets us absorb the additive one in the defining bound for $q(b)$ into the term $n/\sqrt b$.
  Therefore
  \begin{align*}
    \sum_{k=0}^m(1+\sqrt{a_k})
    =
    O_K\left(
    \sqrt n\log(2n)+
    \sqrt n\sum_{\substack{b=2^\ell\\1\leq b\leq J}}b^{1/4}
    \right)
    =
    O_K\left(\sqrt n\log(2n)+\sqrt nJ^{1/4}\right)
    =
    O_K(n).
  \end{align*}
  This proves \cref{eq:burnin-schedule-cost} and, as noted at the start, also gives $m+1=O_K(n)$.

  \emph{4. Constructing the schedule efficiently.}
  To finish, we show that the defining rule can be evaluated within the claimed $O(n^2)$ construction time.
  Find $q(a_0)$ by doubling and maintain subsequent values by halving whenever the defining inequality still holds.
  For $a\leq1$, the test is $q^2\geq16K(1+n)$.
  For $a>1$, it is equivalent to
  \begin{align*}
    q^2>16K,
    \qquad
    a(q^2-16K)^2\geq(16K n)^2.
  \end{align*}
  Thus no partition functions or polynomial roots are computed.
  The value of $q$ never increases, so there are only $O_K(\log(2n))$ successful halvings over the whole schedule,
  in addition to one final failed halving test per stage.
  Thus there are $O_K(m+\log(2n))=O_K(n)$ tests in total.
  The construction is therefore within the claimed $C_K n^2$ time.
\end{proof}

\subsection{Particle-chain identity and spectral gap}
\label{sub:particle-gap}

Let $\Omega_1,\ldots,\Omega_n$ be finite nonempty sets, and consider the state space
\begin{align*}
  \Omega:=\prod_{i=1}^n\Omega_i.
\end{align*}
For $A\subseteq[n]$, write $x_A=(x_i)_{i\in A}$ for the partial configuration with its coordinate labels retained.
Write $x_{-i}$ and $x_{-i,-j}$ for the configurations obtained by deleting the indicated coordinates,
and $x^{i\to u}$ for $x$ with coordinate $i$ replaced by $u\in \Omega_i$.

For each $i$, fix a strictly positive probability distribution $p_i$ on $\Omega_i$.
Assign a positive weight $w(x_A)$ to every labelled partial configuration.
These weights define the target distribution
\begin{align}
  \pi(x)=\frac1Z w(x)\prod_{i=1}^n p_i(x_i),
  \qquad x\in\Omega,
  \label{eq:particle-stationary}
\end{align}
where $Z$ is the normalizing constant.

To sample from $\pi$, consider the continuous-time chain that resamples coordinate $i$ from $p_i$ at rate $w(x_{-i})/w(x)$,
with positive Laplacian
\begin{align}
  (\mathcal L f)(x)
  :=\sum_{i=1}^n\frac{w(x_{-i})}{w(x)}
  \sum_{u\in \Omega_i}p_i(u)
  \bigl(f(x)-f(x^{i\to u})\bigr).
  \label{eq:particle-laplacian}
\end{align}
This chain is irreducible and reversible with stationary distribution $\pi$.

To bound its spectral gap, associate to each $x\in\Omega$ the symmetric $n\times n$ matrix
\begin{align}
  K_{ij}(x):=
  \begin{cases}
    \dfrac{w(x_{-i})}{w(x)}
    \left(\dfrac{w(x_{-i})}{w(x)}-1\right),&i=j,\\[1.5ex]
    \dfrac{w(x_{-i})w(x_{-j})}{w(x)^2}
    -\dfrac{w(x_{-i,-j})}{w(x)},&i\ne j.
  \end{cases}
  \label{eq:particle-curvature-matrix}
\end{align}
\begin{lemma}
  \label{lem:particle-gap}
  If $K(x)\succeq0$ for every $x\in\Omega$, the chain with Laplacian $\mathcal L$ has spectral gap at least $1$.
\end{lemma}

\begin{remark}
  For fixed weights $w(x)$ on complete configurations and distributions $p_i$ for the individual coordinates,
  any positive choices of weights on labelled partial configurations preserve the stationary distribution $\pi$;
  these weights need not be marginals of the weights on complete configurations.
  The weights $w(x_{-i})$ determine transition rates and can affect the spectral gap.
  The weights $w(x_{-i,-j})$ affect only the matrix $K$, without changing the chain.
  To obtain the gap bound in \cref{lem:particle-gap}, the chosen weights must satisfy $K(x)\succeq0$ for every $x\in\Omega$.
  No weights of smaller configurations enter the chain or this criterion.
\end{remark}

We prove the spectral-gap criterion by first establishing a Bochner identity,
and then apply the criterion to occupancy-dependent rates.

\subsubsection{A Bochner identity}
\label{sub:particle-bochner-proof}

We use the stationary distribution $\pi$, Laplacian $\mathcal L$, and matrix $K$ defined above.

Define the discrete gradient, with one block for each particle, by
\begin{align*}
  (\nabla_i f)(x)_u:=f(x^{i\to u})-f(x)
  \quad(u\in \Omega_i),\qquad
  (\nabla f)(x):=\bigl((\nabla_i f)(x)\bigr)_{i=1}^n.
\end{align*}
Regard $p_i$ as a vector on $\Omega_i$ and use the ordinary Euclidean inner product, so
\begin{align*}
  \langle p_i,(\nabla_i f)(x)\rangle
  =\sum_u p_i(u)(\nabla_i f)(x)_u.
\end{align*}
All sums over proposed sites range over the corresponding $\Omega_i$.
Replacements of distinct particles commute, so their mixed second differences are
\begin{align*}
  \Delta_{ij}f(x;u,z)
  &:=(\nabla_i f)(x^{j\to z})_u-(\nabla_i f)(x)_u\\
  &=f(x^{i\to u,j\to z})-f(x^{i\to u})
  -f(x^{j\to z})+f(x),\qquad i\ne j.
\end{align*}

\emph{The sum-of-squares identity.}
The identity uses the matrix $K$ from \cref{eq:particle-curvature-matrix}.

\begin{proposition}
  \label{prop:particle-bochner-identity}
  For every function $f:\Omega\to\mathbb R$,
  \begin{align}
    &\langle\mathcal L f,\mathcal L f\rangle_\pi-\Dir(f,f)\notag\\
    &=\frac14\E[\pi]{\sum_{i\ne j}\frac{w(x_{-i,-j})}{w(x)}
      \sum_{u,z}p_i(u)p_j(z)\bigl(\Delta_{ij}f(x;u,z)\bigr)^2}\notag\\
    &\quad+\E[\pi]{\sum_{i,j=1}^n K_{ij}(x)
      \langle p_i,\nabla_i f\rangle\langle p_j,\nabla_j f\rangle}.
    \label{eq:particle-bochner-identity}
  \end{align}
  All quantities inside the expectations are evaluated at $x\sim\pi$.
  The first sum is over ordered pairs $i\ne j$.
  The identity does not require $K(x)\succeq0$.
\end{proposition}

\begin{claim}
  \label{clm:particle-quadratic-forms}
  For every function $f:\Omega\to\mathbb R$,
  \begin{align}
    \langle\mathcal L f,\mathcal L f\rangle_\pi
    =\E[\pi]{\sum_{i,j}\frac{w(x_{-i})w(x_{-j})}{w(x)^2}
      \langle p_i,\nabla_i f\rangle\langle p_j,\nabla_j f\rangle}
    \label{eq:particle-laplacian-square}
  \end{align}
  and
  \begin{align}
    \Dir(f,f)=\langle f,\mathcal L f\rangle_\pi
    =\E[\pi]{\sum_i\frac{w(x_{-i})}{w(x)}\langle p_i,\nabla_i f\rangle^2}.
    \label{eq:particle-dirichlet-gradient}
  \end{align}
\end{claim}

\begin{proof}
  By \cref{eq:particle-laplacian}, the rate from $x$ to $x^{i\to u}$ is
  \begin{align*}
    r(x,x^{i\to u})=\frac{w(x_{-i})}{w(x)}p_i(u)\quad(u\ne x_i).
  \end{align*}

  The corresponding Laplacian is
  \begin{align*}
    \mathcal L f
    =-\sum_i\sum_{u\ne x_i}r(x,x^{i\to u})(\nabla_i f)(x)_u
    =-\sum_i\frac{w(x_{-i})}{w(x)}\langle p_i,\nabla_i f\rangle,
  \end{align*}
  so squaring and taking expectation under $\pi$ proves \cref{eq:particle-laplacian-square}.
  To prove \cref{eq:particle-dirichlet-gradient}, we use the definitions to write
  \begin{align*}
    \Dir(f,f)
    &=\langle f,\mathcal L f\rangle_\pi
    =\sum_x\pi(x)f(x)(\mathcal L f)(x)\\
    &=\sum_i\sum_x\pi(x)\frac{w(x_{-i})}{w(x)}f(x)
    \left(f(x)-\sum_u p_i(u)f(x^{i\to u})\right).
  \end{align*}
  By \cref{eq:particle-stationary},
  \begin{align*}
    \pi(x)\frac{w(x_{-i})}{w(x)}
    =\frac1Z w(x)\left(\prod_k p_k(x_k)\right)
    \frac{w(x_{-i})}{w(x)}
    =\frac{w(x_{-i})}Z\left(\prod_{k\ne i}p_k(x_k)\right)p_i(x_i).
  \end{align*}
  For fixed $i$ and $x_{-i}$, only the last factor depends on $x_i$.
  Hence
  \begin{align*}
    \Dir(f,f) = \sum_i \sum_{x_{-i}} \frac{w(x_{-i})}Z\left(\prod_{k\ne i}p_k(x_k)\right) \sum_{x_i}p_i(x_i)f(x)
    \left(f(x)-\sum_u p_i(u)f(x^{i\to u})\right).
  \end{align*}
  Moreover, $\sum_u p_i(u)f(x^{i\to u})$ is independent of $x_i$, and $\sum_{x_i}p_i(x_i)=1$.
  Hence
  \begin{align*}
    &\sum_{x_i}p_i(x_i)f(x)
    \left(f(x)-\sum_u p_i(u)f(x^{i\to u})\right)
    =\sum_{x_i}p_i(x_i)f(x)^2
    -\left(\sum_{x_i}p_i(x_i)f(x)\right)^2\\
    &\qquad=\sum_{x_i}p_i(x_i)
    \left(f(x)-\sum_u p_i(u)f(x^{i\to u})\right)^2
    =\sum_{x_i}p_i(x_i)\langle p_i,\nabla_i f\rangle(x)^2,
  \end{align*}
  where the last equality is the definition of $\nabla_i f$.
  Plugging this back, we have
  \begin{align*}
    \Dir(f,f)
    &=\frac1Z\sum_i\sum_x w(x_{-i})
    \left(\prod_k p_k(x_k)\right)\langle p_i,\nabla_i f\rangle(x)^2\\
    &=\sum_x\pi(x)\sum_i\frac{w(x_{-i})}{w(x)}\langle p_i,\nabla_i f\rangle(x)^2
    =\E[\pi]{\sum_i\frac{w(x_{-i})}{w(x)}\langle p_i,\nabla_i f\rangle^2}.\qedhere
  \end{align*}
\end{proof}

\begin{proof}[Proof of \cref{prop:particle-bochner-identity}]
  By \cref{clm:particle-quadratic-forms} and the definition of $K$ in \cref{eq:particle-curvature-matrix},
  \begin{align*}
    &\langle\mathcal L f,\mathcal L f\rangle_\pi-\Dir(f,f)\\
    &=\E[\pi]{\sum_{i,j}K_{ij}
    \langle p_i,\nabla_i f\rangle\langle p_j,\nabla_j f\rangle}
    +\sum_{i\ne j}\E[\pi]{\frac{w(x_{-i,-j})}{w(x)}
      \langle p_i,\nabla_i f\rangle\langle p_j,\nabla_j f\rangle}.
  \end{align*}
  The first term is already the $K$ term in \cref{eq:particle-bochner-identity}.
  It remains to show that, for each $i\ne j$,
  \begin{align}
    \E[\pi]{\frac{w(x_{-i,-j})}{w(x)}
      \langle p_i,\nabla_i f\rangle\,\langle p_j,\nabla_j f\rangle}=\frac14\E[\pi]{\frac{w(x_{-i,-j})}{w(x)}
      \sum_{u,z}p_i(u)p_j(z)\bigl(\Delta_{ij}f(x;u,z)\bigr)^2}.
    \label{eq:particle-hessian-square}
  \end{align}
  Summing this identity over ordered pairs in the opening decomposition will give \cref{eq:particle-bochner-identity}.

  Fix $i\ne j$.
  For each fixed $x$ and proposed sites $u,z$, group the contributions from $x$, $x^{i\to u}$, $x^{j\to z}$,
  and $x^{i\to u,j\to z}$.
  Expanding the inner products gives
  \begin{align}
    &\E[\pi]{\frac{w(x_{-i,-j})}{w(x)}
      \langle p_i,\nabla_i f\rangle\langle p_j,\nabla_j f\rangle}\notag\\
    &\qquad=\sum_x\sum_{u,z}\pi(x)\frac{w(x_{-i,-j})}{w(x)}p_i(u)p_j(z)
    (\nabla_i f)(x)_u(\nabla_j f)(x)_z.
    \label{eq:particle-mixed-expansion}
  \end{align}
  At the four configurations, the proposed site pairs are $(u,z),(x_i,z),(u,x_j),(x_i,x_j)$, respectively.
  All subscripts below refer to the original $x,u,z$.
  For example, at $x^{i\to u}$ the proposed sites are $x_i,z$, so the $i$-move returns to $x$.

  To justify the regrouping, the coefficient of each product is
  \begin{align*}
    \pi(x)\frac{w(x_{-i,-j})}{w(x)}p_i(u)p_j(z)
    =\frac{w(x_{-i,-j})\prod_{k\ne i,j}p_k(x_k)}Z
    p_i(x_i)p_i(u)p_j(x_j)p_j(z).
  \end{align*}
  This coefficient is unchanged by interchanging $x_i,u$, by interchanging $x_j,z$, or by doing both.
  These are bijections of the summation indices, so the four products evaluated at these configurations have the same total weighted sum.
  Taking their average explains the factor $1/4$.
  Thus
  \begin{align*}
    \text{LHS of }\eqref{eq:particle-mixed-expansion}
    &=\frac14\sum_x\sum_{u,z}\pi(x)\frac{w(x_{-i,-j})}{w(x)}p_i(u)p_j(z)\\
    &\quad\times\Bigl[
      (\nabla_i f)(x)_u(\nabla_j f)(x)_z
      +(\nabla_i f)(x^{i\to u})_{x_i}(\nabla_j f)(x^{i\to u})_z\\
      &\qquad\quad
      +(\nabla_i f)(x^{j\to z})_u(\nabla_j f)(x^{j\to z})_{x_j}
      +(\nabla_i f)(x^{i\to u,j\to z})_{x_i}
      (\nabla_j f)(x^{i\to u,j\to z})_{x_j}\Bigr].
  \end{align*}

  To simplify the bracket, reversing the two coordinate moves gives
  \begin{align*}
    (\nabla_i f)(x^{i\to u})_{x_i}
    &=f(x)-f(x^{i\to u})=-(\nabla_i f)(x)_u,\\
    (\nabla_j f)(x^{j\to z})_{x_j}
    &=f(x)-f(x^{j\to z})=-(\nabla_j f)(x)_z.
  \end{align*}
  The same identities hold after replacing the other coordinate.
  Thus the four products have signs $+,-,-,+$ when written using the differences for the moves $x_i\to u$ and $x_j\to z$,
  and the bracket factors as
  \begin{align*}
    \bigl((\nabla_i f)(x^{j\to z})_u-(\nabla_i f)(x)_u\bigr)
    \bigl((\nabla_j f)(x^{i\to u})_z-(\nabla_j f)(x)_z\bigr)
    =\bigl(\Delta_{ij}f(x;u,z)\bigr)^2.
  \end{align*}
  Indeed, each factor is $f(x^{i\to u,j\to z})-f(x^{i\to u})-f(x^{j\to z})+f(x)$,
  since the two replacements commute.
  Substituting this into the regrouped sum proves \cref{eq:particle-hessian-square} and hence \cref{eq:particle-bochner-identity}.
\end{proof}

\subsubsection{The spectral gap and occupancy-dependent rates}

We can now deduce the spectral-gap criterion.

\begin{proof}[Proof of \cref{lem:particle-gap}]
  The chain is reversible with respect to $\pi$ and irreducible on $\Omega$.
  In \cref{prop:particle-bochner-identity}, the first term is nonnegative because $w$ is positive.
  The second is the expected quadratic form of $K(x)\succeq0$ at the vector $(\langle p_i,\nabla_i f\rangle)_{i=1}^n$,
  so it is also nonnegative.
  Thus, for every $f:\Omega\to\mathbb R$, the identity gives the Bochner inequality
  \begin{equation*}
    \langle\mathcal L f,\mathcal L f\rangle_\pi\geq\langle f,\mathcal L f\rangle_\pi.\qedhere
  \end{equation*}
\end{proof}

The following specialization gives the particle chains used here.

\begin{corollary}
  \label{cor:occupancy-particle-gap}
  Let $S$ be a finite set of sites and let $p_i$ be a probability distribution on $S$ for each particle $i$.
  Set $\Omega_i:=\supp(p_i)$ and $\Omega:=\prod_i\Omega_i$;
  several particles may occupy the same site.
  For each $v\in S$, suppose
  \begin{align*}
    \theta_v(0)=1\leq\theta_v(1)\leq\cdots\leq\theta_v(n).
  \end{align*}
  For each labelled partial configuration $x_A$, let $d_v(x_A):=|\{i\in A:x_i=v\}|$ and set
  \begin{align*}
    w(x_A):=\prod_{v\in S}\prod_{k=1}^{d_v(x_A)}\frac1{\theta_v(k)},
  \end{align*}
  where an empty product equals $1$.
  Then the Laplacian $\mathcal L$ in \cref{eq:particle-laplacian} has spectral gap at least $1$.
\end{corollary}

\begin{proof}
  By \cref{lem:particle-gap}, it suffices to verify $K(x)\succeq0$.
  Deleting particle $i$ gives
  \begin{align*}
    \frac{w(x_{-i})}{w(x)}=\theta_{x_i}(d_{x_i}(x)).
  \end{align*}
  If $x_i\ne x_j$, the two deletions affect different site factors, so $K_{ij}=0$.
  Now fix an occupied site $v$ and put $m=d_v(x)$.
  For distinct particles $i,j$ with $x_i=x_j=v$, deleting them successively gives
  \begin{align*}
    \frac{w(x_{-i,-j})}{w(x)}
    =\frac{w(x_{-i,-j})}{w(x_{-i})}\frac{w(x_{-i})}{w(x)}
    =\theta_v(m-1)\theta_v(m).
  \end{align*}
  Substituting these ratios into \cref{eq:particle-curvature-matrix},
  the block indexed by $\{i:x_i=v\}$ has entries
  \begin{align*}
    K_{ij}(x)=
    \begin{cases}
      \theta_v(m)(\theta_v(m)-1),&i=j,\\
      \theta_v(m)^2-\theta_v(m)\theta_v(m-1)
      =\theta_v(m)(\theta_v(m)-\theta_v(m-1)),&i\ne j.
    \end{cases}
  \end{align*}
  Thus this block equals
  \begin{align*}
    &\theta_v(m)(\theta_v(m)-1)I_m
    +\theta_v(m)(\theta_v(m)-\theta_v(m-1))
    (\mathbf1_m\mathbf1_m^{\mathsf T}-I_m)\\
    &\qquad=\theta_v(m)(\theta_v(m-1)-1)I_m
    +\theta_v(m)(\theta_v(m)-\theta_v(m-1))
    \mathbf1_m\mathbf1_m^{\mathsf T}.
  \end{align*}
  The same expression holds for $m=1$, when there are no off-diagonal entries and $\theta_v(0)=1$.
  Both coefficients in the last line are nonnegative, so every block is positive semidefinite.
  Hence $K(x)\succeq0$, and the lemma applies.
\end{proof}

\begingroup
\subsection{Stationarity and spectral gap of assignment dynamics}
\label{sub:burnin-proof}

\begin{proof}[Proof of \cref{clm:burnin-stationarity}]
  It suffices to verify detailed balance for configurations differing at one coordinate $u$.
  Fix the remaining assignments $\sigma_{-u}$.
  If $v\in V_2$ is already occupied in $\sigma_{-u}$, assigning $u$ to $v$ increases the count $\mathsf{cr}$ by one and gives resampling rate $1+J$.
  Otherwise the count is unchanged and the resampling rate is $1$.
  The same rate $1$ applies to the unassigned state.
  Consequently, as $v$ varies over $V_2\cup\{\ast\}$,
  \begin{align*}
    \nu(\sigma^{u\to v})\theta_u(\sigma^{u\to v})
    =c(\sigma_{-u})q_u(v)
  \end{align*}
  for a positive factor $c(\sigma_{-u})$ independent of $v$.
  Multiplying by $q_u(z)$ makes the right-hand side symmetric in $v,z$,
  which proves detailed balance for the rates in \cref{eq:burnin-laplacian}.
  All rates and entries of $q_u$ are positive, so successive coordinate replacements connect every pair of configurations.
  This also proves irreducibility.
\end{proof}

\begin{proof}[Proof of \cref{clm:burnin-relaxation}]
  We identify $\mathcal L_a$ with the particle Laplacian in \cref{cor:occupancy-particle-gap}.
  Besides the sites $V_2$, give each $u\in V_1$ a separate unassigned site $\ast_u$,
  which only particle $u$ can occupy, and take
  \begin{align*}
    p_u(\ast_u)=\frac1{1+a},
    \qquad
    p_u(v)=\frac{aB(u,v)}{1+a}
    \quad(v\in V_2).
  \end{align*}
  The resampling rate is $1$ at the unassigned site.
  At a right vertex, the resampling rate is $1$ at occupancy one and $1+J$ at every larger occupancy.
  Taking the rate at occupancy zero to be $1$, these rates are at least one and nondecreasing in the occupancy.
  Thus \cref{cor:occupancy-particle-gap} gives the claimed spectral gap.
\end{proof}
\endgroup

\subsection{Fast simulation of assignment Glauber dynamics}
\label{sub:burnin-simulation-proof}

\begin{claim}
  \label{clm:weighted-sampling-tree}
  Let $N\geq1$ and let $b_1,\ldots,b_N$ be nonnegative weights with positive sum $S:=\sum_{j=1}^N b_j$.
  A segment tree can be built in $O(N)$ time and space to sample an index $j$ with probability $b_j/S$ in $O(\log(2N))$ time.
  Replacing one weight by a nonnegative value, while keeping the total positive, also takes $O(\log(2N))$ time.
  The current total weight can be read in $O(1)$ time.
\end{claim}

\begin{proof}
  Store the weights at the leaves of a segment tree and the sum of the weights below each internal node at that node.
  The tree has $O(N)$ nodes and height $O(\log(2N))$;
  computing its sums from the leaves upward takes $O(N)$ time.
  To sample, start at the root and repeatedly choose a child with probability equal to its sum divided by its parent's sum.
  Zero-weight branches are never entered, and the probabilities along the path to leaf $j$ multiply to $b_j/S$.
  Replacing a weight requires updating only the sums along its path to the root.
  Thus both operations take $O(\log(2N))$ time, and the root stores the current total weight.
\end{proof}

\begin{proof}[Proof of \cref{prop:burnin-simulation}]
  Let $C(\sigma)$ be the set of conflicted left vertices, namely those assigned to a right vertex shared with another left vertex,
  and let
  \begin{align*}
    W(\sigma):=\sum_{u\in V_1}\theta_u(\sigma)=n+J|C(\sigma)|
  \end{align*}
  be the total resampling rate.
  To simulate up to a fixed time $T$, start at time $t=0$.
  At time $t$ in configuration $\sigma$, start a fresh independent Poisson clock with rate $W(\sigma)$.
  If it does not ring by time $T$, return $\sigma$.
  Otherwise, set $t$ to the ringing time, choose $u\in V_1$ with probability $\theta_u(\sigma)/W(\sigma)$,
  resample $\sigma_u$ from $q_u$, and repeat with a new clock at the updated rate.
  These choices give each resampling $\sigma\to\sigma^{u\to v}$ rate $\theta_u(\sigma)q_u(v)$,
  as required by \cref{eq:burnin-laplacian}.
  A ring may leave the assignment unchanged, but time still advances.

  \paragraph{Preprocessing.}
  This preprocessing is done once before all simulations with the same matrix $B$,
  including those at different activities $a$.
  First compute $B$ from \cref{eq:burnin-rescaling} in $O(n^2)$ time.
  For each row $u$, maintain a segment tree $\mathbb{T}_u$ with weights $B(u,\cdot)$.
  By \cref{clm:weighted-sampling-tree}, building all these trees takes $O(n^2)$ time.

  \paragraph{Simulation.}
  For each simulation, maintain a segment tree $\mathbb{T}_\theta$ with weights $\theta_u(\sigma)$;
  its root stores $W(\sigma)$.
  The initialization takes $O(n)$ time.
  At each ring, sample $u$ using $\mathbb{T}_\theta$.
  To resample $\sigma_u$, choose $\ast$ with probability $1/(1+a)$;
  otherwise sample $v$ using $\mathbb{T}_u$.
  After resampling, update $\mathbb{T}_\theta$.
  Only the selected vertex and at most one other vertex at each of its old and new sites can change whether they share their assigned right vertex,
  so only $O(1)$ weights in $\mathbb{T}_\theta$ need updating.
  By \cref{clm:weighted-sampling-tree}, sampling and updating therefore take $O(\log(2n))$ time per resampling.
  The total simulation time is therefore $O(n+Q_T\log(2n))$ after this $O(n^2)$ setup.
\end{proof}

\begin{proof}[Proof of \cref{prop:burnin-event-count}]
  We reduce the desired bound on $Q_T$ to a bound on the maximum number of unassigned vertices.
  Put
  \begin{align*}
    h:=1+\min\{n,n/\sqrt a\},
    \qquad s:=T+\log(2n).
  \end{align*}
  It suffices to prove the stronger estimate
  \begin{align}
    \Pr{Q_T>C_\alpha(1+a)(h+s)s}\leq e^{-T}.
    \label{eq:burnin-proposal-count}
  \end{align}
  Indeed, $a\in[1/n,J]$ and $J=O_\alpha(n^2)$ imply $(1+a)h=O_\alpha(n(1+\sqrt a))$.
  Since $h\geq1$ and $s\geq\log2$,
  \begin{align*}
    (1+a)(h+s)s
    \leq3(1+a)hs^2
    =O_\alpha\bigl(n(1+\sqrt a)s^2\bigr).
  \end{align*}

  To obtain \cref{eq:burnin-proposal-count}, we bound the number of resamplings using the maximum number of unassigned vertices during the simulation.
  Write $(X_t)_{0\leq t\leq T}$ for the trajectory and $U_t:=\mathsf{un}(X_t)$.
  We will show that, for every integer $K\geq1$,
  \begin{align*}
    \Pr{Q_T\geq\lceil16(1+a)Ks\rceil}
    \leq\Pr{\sup_{0\leq t\leq T}U_t>K}+\frac23e^{-T}.
  \end{align*}
  It will then suffice to choose $K=O_\alpha(h+s)$ so that the first probability is at most $e^{-T}/3$.

  To bound $Q_T$, we first bound the cumulative number of resamplings that leave the selected coordinate unassigned.
  For $0\leq t\leq T$, let $N_{\mathrm{start}}(t)$ be the number of times \emph{up to time $t$} that a coordinate is selected for resampling while it is unassigned.
  Let $N_{\mathrm{end}}(t)$ be the number of resamplings \emph{up to time $t$} that propose $\ast$.
  Repeated resamplings of the same coordinate are counted separately.
  In particular, a resampling from $\ast$ to $\ast$ adds one to both counts and leaves the number of unassigned coordinates unchanged.
  Accounting for the changes in $U_t$ gives
  \begin{align*}
    N_{\mathrm{end}}(t)-N_{\mathrm{start}}(t)=U_t-U_0,
    \qquad 0\leq t\leq T.
  \end{align*}
  Write $\mathcal G:=\{\sup_{0\leq t\leq T}U_t\leq K\}$.
  On $\mathcal G$, the identity at time $T$ gives $N_{\mathrm{end}}(T)\leq N_{\mathrm{start}}(T)+K$.
  Set $R:=\lceil16(1+a)Ks\rceil$.
  Partition the sample space into the three disjoint events $\mathcal G^c$,
  $\mathcal G\cap\{N_{\mathrm{start}}(T)>4Ks\}$, and $\mathcal G\cap\{N_{\mathrm{start}}(T)\leq4Ks\}$.
  On the third event, the balance identity gives $N_{\mathrm{end}}(T)\leq4Ks+K$.
  Therefore,
  \begin{align*}
    \Pr{Q_T\geq R}
    &=\Pr{Q_T\geq R,\ \mathcal G^c}
    +\Pr{Q_T\geq R,\ \mathcal G,\ N_{\mathrm{start}}(T)>4Ks}\\
    &\quad+\Pr{Q_T\geq R,\ \mathcal G,\ N_{\mathrm{start}}(T)\leq4Ks}\\
    &\leq\Pr{\mathcal G^c}
    +\Pr{\mathcal G,\ N_{\mathrm{start}}(T)>4Ks}\\
    &\quad+\Pr{Q_T\geq R,\ N_{\mathrm{end}}(T)\leq4Ks+K}.
  \end{align*}
  We bound the last two terms by $e^{-T}/3$ each;
  the first term will be controlled by the choice of $K$ below.

  For the last term, generate the choices of whether to propose $\ast$ as an infinite sequence of independent Bernoulli variables of mean $p:=1/(1+a)$.
  The event requires at most $4Ks+K$ successes among the first $R$ choices.
  Since $s\geq\log2$, we have $4Ks+K\leq8Ks\leq Rp/2$.
  A binomial lower-tail bound therefore gives
  \begin{align*}
    \Pr{Q_T\geq R,\ N_{\mathrm{end}}(T)\leq4Ks+K}
    \leq e^{-Rp/8}
    \leq e^{-2Ks}
    \leq e^{-T}/3.
  \end{align*}

  For the second term, let $\mathcal F_t$ denote the history of the simulation up to time $t$ and write $X_t=\sigma$ for the current configuration.
  The chain's clock rings at rate $W(\sigma)$, and each ring selects $u$ with probability $\theta_u(\sigma)/W(\sigma)$.
  The count $N_{\mathrm{start}}(t)$ increases exactly when the selected coordinate is unassigned.
  Since $\theta_u(\sigma)=1$ for these coordinates, its conditional instantaneous rate is
  \begin{align*}
    W(\sigma)\sum_{u:\sigma_u=\ast}\frac{\theta_u(\sigma)}{W(\sigma)}
    =\sum_{u:\sigma_u=\ast}1=U_t.
  \end{align*}
  More precisely, for each fixed $0\leq t<T$, as $\delta\downarrow0$,
  \begin{align*}
    \Pr{N_{\mathrm{start}}(t+\delta)-N_{\mathrm{start}}(t)=1
    \mid\mathcal F_t}
    =U_t\delta+o(\delta).
  \end{align*}
  Define
  \begin{align*}
    \tau_K&:=\inf\{t\in[0,T]:U_t>K\},\qquad \inf\varnothing:=\infty,\\
    \widehat N_{\mathrm{start}}(t)&:=N_{\mathrm{start}}(t\wedge\tau_K),
    \qquad 0\leq t\leq T.
  \end{align*}
  Let
  \begin{align*}
    Z\sim\operatorname{Poisson}(KT).
  \end{align*}
  The stopped process has instantaneous rate at most $K$.
  By coupling its increases to a subset of the rings of a Poisson clock of rate $K$,
  the stopped count $\widehat N_{\mathrm{start}}(T)$ is stochastically dominated by $Z$.
  On $\mathcal G$, it equals $N_{\mathrm{start}}(T)$.
  Using $T\leq s$ and the exponential moment of $Z$, we obtain
  \begin{align*}
    \Pr{\mathcal G,\ N_{\mathrm{start}}(T)>4Ks}
    &\leq\Pr{\widehat N_{\mathrm{start}}(T)>4Ks}
    \leq\Pr{Z>4Ks}\\
    &\leq e^{(e-1)KT-4Ks}
    \leq e^{-2Ks}
    \leq e^{-T}/3.
  \end{align*}
  Together with the balance identity, these two tail bounds prove the reduction to $\sup_{t\leq T}U_t$.

  It remains to choose $K$ so that $\Pr{\mathcal G^c}\leq e^{-T}/3$.
  Write $\xi_t(\sigma):=\Pr{X_t=\sigma}$.
  Because the $L^2(\nu_a)$ norm of the density $\xi_t/\nu_a$ cannot increase with time,
  and the initial distribution is $3$-warm, for every fixed $t\in[0,T]$ we have
  \begin{align*}
    \sum_{\sigma\in\Omega_{\mathrm{aux}}}
    \frac{\xi_t(\sigma)^2}{\nu_a(\sigma)}
    \leq\sum_{\sigma\in\Omega_{\mathrm{aux}}}
    \frac{\xi_0(\sigma)^2}{\nu_a(\sigma)}\leq3.
  \end{align*}
  Writing the tail probability as a sum of products, Cauchy--Schwarz gives
  \begin{align*}
    \Pr{U_t\geq K}
    &=\sum_{\sigma:\mathsf{un}(\sigma)\geq K}
    \frac{\xi_t(\sigma)}{\sqrt{\nu_a(\sigma)}}\sqrt{\nu_a(\sigma)}\\
    &\leq\left(\sum_{\sigma\in\Omega_{\mathrm{aux}}}
    \frac{\xi_t(\sigma)^2}{\nu_a(\sigma)}\right)^{1/2}
    \left(\sum_{\sigma:\mathsf{un}(\sigma)\geq K}
    \nu_a(\sigma)\right)^{1/2}\\
    &\leq\sqrt{3\nu_a(\mathsf{un}\geq K)}.
  \end{align*}
  To bound $\Pr{\mathcal G^c}$, we count the resamplings at which a first crossing above $K$ could occur.
  Let $S_1<\cdots<S_{Q_T}$ be the resampling times up to time $T$,
  including resamplings that leave the assignment unchanged.
  Write $U_{S_j-}$ for the number of unassigned vertices immediately before the $j$-th resampling, and define
  \begin{align*}
    C_t:=\sum_{j=1}^{Q_T}
    \one_{\{S_j\leq t,\ U_{S_j-}\geq K\}},
    \qquad 0\leq t\leq T.
  \end{align*}
  Thus $C_t$ counts all resamplings \emph{up to time $t$} that start with at least $K$ unassigned vertices.
  Each resampling changes $U_t$ by at most one, so if $U_0\leq K$,
  a first crossing above $K$ contributes to $C_T$.
  Consequently,
  \begin{align*}
    \mathcal G^c&\subseteq\{U_0>K\}\cup\{C_T\geq1\},\\
    \Pr{\mathcal G^c}
    &\leq\Pr{U_0>K}+\Pr{C_T\geq1}
    \leq\Pr{U_0>K}+\E{C_T}.
  \end{align*}

  To estimate $\E{C_T}$, recall that the total resampling rate satisfies $W(X_t)\leq W_{\max}:=n(1+J)=O_\alpha(n^3)$.
  For fixed $0\leq t<T$, conditional on $\mathcal F_t$, the first ring after time $t$ is counted exactly when $U_t\geq K$.
  Hence, as $\delta\downarrow0$,
  \begin{align*}
    \E{C_{t+\delta}-C_t\mid\mathcal F_t}
    =W(X_t)\one_{\{U_t\geq K\}}\delta+O(W_{\max}^2\delta^2).
  \end{align*}
  The error is uniform in $t$.
  The number of rings in an interval of length $\delta$ is dominated by a Poisson variable of mean $W_{\max}\delta$,
  so intervals containing two or more rings contribute at most $O(W_{\max}^2\delta^2)$ to the conditional expectation.
  Taking expectations, summing over a uniform partition of $[0,T]$, and letting the mesh tend to zero gives
  \begin{align*}
    \E{C_T}
    =\int_0^T\E{W(X_t)\one_{\{U_t\geq K\}}}\,dt
    \leq W_{\max}\int_0^T\Pr{U_t\geq K}\,dt.
  \end{align*}
  Combining these inequalities with the fixed-time bound above yields
  \begin{align*}
    \Pr{\mathcal G^c}
    &\leq\Pr{U_0>K}+W_{\max}\int_0^T\Pr{U_t\geq K}\,dt\\
    &\leq(1+W_{\max}T)\sqrt{3\nu_a(\mathsf{un}\geq K)}\\
    &\leq2(1+W_{\max}T)\exp\left(\frac{A_\alpha h-K\log2}{2}\right),
  \end{align*}
  where the last inequality is \cref{eq:burnin-unmatched-tail}.
  Since $\log(1+W_{\max}T)=O_\alpha(s)$, choosing $K:=\lceil D_\alpha(h+s)\rceil$ with $D_\alpha$ sufficiently large makes this probability at most $e^{-T}/3$.
  For this choice, $R=O_\alpha((1+a)(h+s)s)$, so the reduction proves \cref{eq:burnin-proposal-count} and hence the proposition.
\end{proof}

\subsection{Assembling the perfect-matching sampler}
\label{sub:initialization}

For permanent estimation, we need the following version of \cref{lem:fast-burnin} with an adjustable error tolerance.

\begin{lemma}
  \label{lem:fast-burnin-general}
  For all $\alpha>1$, there exists a constant $C_\alpha>0$ such that the following holds.
  For $0\leq i\leq L$, let $\lambda_i$ be a temperature in the cooling schedule \cref{schedule},
  and let $\mu_i$ be the distribution on $\mathcal{P}$ prescribed by $\lambda_i$.
  For any $r\in V_1$, given $\alpha$-rough hole weights $w_i(r,v)$ at temperature $\lambda_i$ for all $v\in V_2$,
  and $0<\eta<1/4$, there is a randomized algorithm which, in $C_\alpha n^2\log^4(n/\eta)$ time,
  outputs a perfect matching from a distribution $\widehat\mu_i$ satisfying
  \begin{align*}
    \|\widehat\mu_i-\mu_i\|_{\mathrm{TV}}\leq\eta.
  \end{align*}
\end{lemma}

The proof also gives the sharper running time for a fixed error tolerance in \cref{lem:fast-burnin}.

\begin{proof}
  For $n=1$, return the unique perfect matching.
  Assume $n\geq2$.
  Rescale the activities to obtain $B$, preprocess its rows, and construct the schedule $a_0,\ldots,a_m$ of \cref{lem:burnin-schedule}.
  Put
  \begin{align*}
    \delta:=\frac{\eta}{8(m+1)},\qquad
    T:=\log(2/\delta),\qquad
    R_k:=\left\lceil C_\alpha n(1+\sqrt{a_k})
    \bigl(T+\log(2n)\bigr)^2\right\rceil.
  \end{align*}
  At stage $k$, simulate the assignment Glauber dynamics at $a_k$ up to time $T$,
  with a limit of $R_k$ resamplings.
  If another resampling would exceed this limit, use $\sigma^\ast$ as the stage output.
  For every $3$-warm input distribution, the spectral gap bound in \cref{clm:burnin-relaxation} and the event-count bound in \cref{prop:burnin-event-count} give output error at most $(1+\sqrt2/2)e^{-T}\leq\delta$ in total variation.

  Start at $\sigma^\ast$ and apply the stage simulations with resampling limits $R_k$ successively.
  At stage zero the input is $3$-warm;
  at later stages compare with a reference run started from the exact distribution $\nu_{a_{k-1}}$,
  which is $3$-warm for $\nu_{a_k}$.
  Applying the same limited simulation to two input distributions cannot increase their total variation distance,
  so induction gives error at most $(k+1)\delta$ after stage $k$.
  In particular, the output distribution $\xi$ of one run through the activity schedule satisfies $\|\xi-\nu_J\|_{\mathrm{TV}}\leq\eta/8$.
  This argument does not require the actual intermediate distributions to be warm.

  Since $\nu_J(\mathcal P)\geq1/2$ and $\nu_J(\,\cdot\mid\mathcal P)=\mu_i$,
  we have $\xi(\mathcal P)\geq1/3$ and
  \begin{align*}
    \|\xi(\,\cdot\mid\mathcal P)-\mu_i\|_{\mathrm{TV}}
    \leq\frac{2\|\xi-\nu_J\|_{\mathrm{TV}}}{\nu_J(\mathcal P)}
    \leq\frac\eta2.
  \end{align*}
  Make $b:=\lceil\log(2/\eta)/\log(3/2)\rceil$ independent runs through all stages and return the first perfect matching,
  or an arbitrary fixed perfect matching if all trials fail.
  The failure probability is at most $(2/3)^b\leq\eta/2$;
  conditional on success, the returned distribution is $\xi(\,\cdot\mid\mathcal P)$.
  Thus the final total variation error is at most $\eta$.

  Each run costs $O_\alpha(n^2\log(2n)\log^2(n/\eta))$ by \cref{prop:burnin-simulation,eq:burnin-schedule-cost}.
  Including the $O(n^2)$ preprocessing cost, the total is
  \begin{align*}
    O_\alpha\left(n^2\log(2n)\log^2(n/\eta)\log(1/\eta)\right)
    =O_\alpha\left(n^2\log^4(n/\eta)\right).
  \end{align*}
  For fixed $\eta=1/16$, the sharper bound is $O_\alpha(n^2\log^3(2n))$, proving \cref{lem:fast-burnin}.
\end{proof}

\section{Proofs of the permanent approximation guarantees}
\label{sec:missing}

This section concerns estimation across the temperatures $\lambda_0,\ldots,\lambda_L$,
using the fixed-temperature sampler from \cref{sec:fast-sampling-proofs}.
We first bound the cooling schedule and the second moment of the products used by both the JSV and FAR algorithms.
We then turn to the cubic-time guarantee in \cref{thm:01-main} and its extension to nonnegative matrices in \cref{cor:weighted-main}.

\subsection{Cooling schedule}
\label{sub:cooling-proofs}

The schedule in \cref{lem:general-cooling-schedule} below, which also yields \cref{lem:cooling-schedule},
is the non-adaptive schedule of~\cite[Algorithm~1]{GLYYZ26} run with input bounds four times larger than necessary,
so that every relevant log partition function changes by at most $1/4$ per step.
We first fix notation and record the whole-product estimate of~\cite{GLYYZ26} in our sign convention.

Let $H:\Omega\to[0,H_{\max}]$.
For nonempty $S\subseteq\Omega$, let
\begin{align*}
  Z_S(\beta):=\sum_{M\in S}e^{-\beta H(M)},
  \qquad
  z_S(\beta):=\log Z_S(\beta),
\end{align*}
and let $\mu_{S,\beta}$ be the distribution on $S$ proportional to $e^{-\beta H(M)}$.
Since $-z_S'(\beta)$ and $z_S''(\beta)$ are the mean and the variance of $H$ under $\mu_{S,\beta}$,
the function $z_S$ is nonincreasing and convex, and $0\leq-z_S'\leq H_{\max}$.
Given $0=\beta_0<\cdots<\beta_L$, write $d_i:=\beta_{i+1}-\beta_i$ and $\delta_i:=z_S(\beta_i)-z_S(\beta_{i+1})\geq0$.
Let $M_i\sim\mu_{S,\beta_i}$ be independent, and put $Y_i:=e^{-d_iH(M_i)}$ and $Y:=\prod_{i=0}^{L-1}Y_i$.
Then $\E{Y}=Z_S(\beta_L)/Z_S(0)$, and
\begin{align*}
  \log\frac{\E{Y_i^2}}{\E{Y_i}^2}
  =z_S(\beta_i+2d_i)+z_S(\beta_i)-2z_S(\beta_{i+1})
  =\delta_i-\bigl(z_S(\beta_{i+1})-z_S(\beta_{i+1}+d_i)\bigr).
\end{align*}
If $i\leq L-2$ and $d_i\leq d_{i+1}$, the subtracted term is at least $(d_i/d_{i+1})\delta_{i+1}$,
because $-z_S'$ is nonincreasing.
Consequently, as in the proof of~\cite[Theorem~1.1]{GLYYZ26},
if $d_i\leq d_{i+1}\leq(1+\gamma)d_i$ for $0\leq i\leq L-3$, then
\begin{align}
  \log\frac{\E{Y^2}}{\E{Y}^2}
  \leq\delta_0+\delta_{L-1}+\frac{\gamma}{1+\gamma}\sum_{i=1}^{L-2}\delta_i.
  \label{eq:whole-product-moment}
\end{align}

\begin{lemma}
  \label{lem:general-cooling-schedule}
  Let $H:\Omega\to[0,H_{\max}]$ with $H_{\max}\geq1$, and let $0<\beta_{\max}\leq\Delta$ with $\Delta\geq1$.
  Suppose that $z_S(0)-z_S(\beta_{\max})\leq\Delta$ for $S=\mathcal P$ and for every $S=\mathcal N(u,v)$.
  Set $\beta_0:=0$ and, until $\beta_L=\beta_{\max}$,
  \begin{align*}
    d_i:=\min\left\{\beta_{\max}-\beta_i,\ \frac14\max\left\{\frac1{H_{\max}},\frac{\beta_i}{\Delta}\right\}\right\},
    \qquad
    \beta_{i+1}:=\beta_i+d_i.
  \end{align*}
  Then $L=O(\Delta\log(2H_{\max}))$, and $\delta_i\leq1/4$ for every such $S$ and every $0\leq i<L$.
  In particular, the ideal hole weights $Z_{\mathcal P}(\beta_i)/Z_{\mathcal N(u,v)}(\beta_i)$ change by factors in $[e^{-1/4},e^{1/4}]\subseteq[3/4,4/3]$ between consecutive temperatures.
  Moreover, for every such $S$, the product $Y$ has relative second moment at most $e^{3/4}<64/27$.
\end{lemma}

\begin{proof}
  Replacing $\beta$ by $-\beta$ turns $Z_S$ into the partition function of~\cite{GLYYZ26} on the state space $S$.
  Our schedule then becomes the output of~\cite[Algorithm~1]{GLYYZ26} with endpoints $-\beta_{\max}$ and $0$,
  and with $4\Delta$ and $4H_{\max}$ as the input bounds $q$ and $h$;
  in particular, it does not depend on $S$.
  Every step satisfies $d_i\leq\frac14\max\{1/H_{\max},\beta_{\max}/\Delta\}\leq\frac14$,
  so the stopping rule $d_i>2$ there is never triggered, and only the last step can be truncated.
  Hence~\cite[Lemma~2.1]{GLYYZ26} gives $L=O(\Delta\log(2H_{\max}))$,
  and~\cite[Lemma~2.3]{GLYYZ26} gives $d_i\leq d_{i+1}\leq(1+\frac1{4\Delta})d_i$ for $0\leq i\leq L-3$.
  Both lemmas concern only the schedule, so they do not use the assumption $H(\Omega)\subseteq\{0\}\cup[1,\infty)$ made in~\cite{GLYYZ26}.

  For the step bound, we repeat the derivative bound of~\cite[Lemma~2.2]{GLYYZ26} with the actual bounds $\Delta$ and $H_{\max}$.
  For $0<\beta\leq\beta_{\max}$, convexity and monotonicity give $\beta(-z_S'(\beta))\leq z_S(0)-z_S(\beta)\leq\Delta$.
  Together with $-z_S'\leq H_{\max}$, this gives $d_i(-z_S'(\beta_i))\leq\frac14$,
  and hence $\delta_i\leq d_i(-z_S'(\beta_i))\leq\frac14$, since $-z_S'$ is nonincreasing.
  Applying this to $\mathcal P$ and to $\mathcal N(u,v)$ bounds the change of the ideal hole weights.
  Finally, \cref{eq:whole-product-moment} with $\gamma=\frac1{4\Delta}$,
  together with $\delta_0,\delta_{L-1}\leq\frac14$ and $\sum_i\delta_i\leq\Delta$, gives
  \begin{equation*}
    \log\frac{\E{Y^2}}{\E{Y}^2}
    \leq\frac14+\frac14+\frac{\Delta}{4\Delta+1}
    <\frac34.
    \qedhere
  \end{equation*}
\end{proof}

\begin{proof}[Proof of \cref{lem:cooling-schedule}]
  Apply \cref{lem:general-cooling-schedule} with $H(M):=|M\setminus E|$, $H_{\max}:=n$,
  $\beta_{\max}:=\lceil\log_2(n!)\rceil$, and $\Delta:=2\beta_{\max}$, and set $\lambda_i:=e^{-\beta_i}$.
  Then $\lambda_i(M)=e^{-\beta_iH(M)}$, so $Z_i=Z_{\mathcal P}(\beta_i)$, $\mu_i=\mu_{\mathcal P,\beta_i}$,
  and the product in \cref{lem:cooling-schedule} is the product $Y$ for $S=\mathcal P$.
  It remains to verify the hypothesis on $z_S$.
  Every $S$ contains a matching with at most one nonedge.
  For $\mathcal P$, use $M^*$.
  For holes $(u,v)$, delete $(u,v)$ from $M^*$ if it is present;
  otherwise delete its edges $(u,v')$ and $(u',v)$ and insert $(u',v')$.
  Hence $Z_S(\beta_{\max})\geq e^{-\beta_{\max}}$, while $Z_S(0)=|S|\leq n!$, so
  \begin{align*}
    z_S(0)-z_S(\beta_{\max})\leq\log(n!)+\beta_{\max}\leq2\beta_{\max}.
  \end{align*}
  \Cref{lem:general-cooling-schedule} now gives $L=O(\beta_{\max}\log n)=O(n\log^2 n)$, \cref{eq:cooling-ratio},
  and \cref{eq:cooling-variance}.
\end{proof}

\subsection{Proof of the cubic-time guarantee}
\label{subsec:improved-main-proof}

\begin{proof}[Proof of \cref{thm:01-main}]
  The case $n=1$ is immediate.
  For $n\geq2$, first check whether the input graph has a perfect matching, returning zero if it does not.
  This takes $O(n^3)$ time.
  Otherwise, fix such a matching $M^*$ and use the cooling schedule of \cref{lem:cooling-schedule}.

  Run the FAR learning procedure of \cref{sec:algorithm-FAR} with confidence parameter $\delta=1/8$.
  Since accurate weights at one temperature remain rough at the next,
  the occupation estimates bound the probability of failure at each temperature,
  conditioned on all earlier refinements succeeding.
  Summing over the possible temperatures of the first failed refinement gives accurate anchor weights at all temperatures with probability at least $15/16$.
  The learning time is $O(n^3\log^6 n)$.

  Fix the learned weights at every temperature, assuming they are all accurate,
  and condition on these values for the rest of the argument.
  Use roughness parameter $\alpha=2$ for every sampler call.
  For Phase~2 of \cref{sec:algorithm-FAR}, take
  \begin{align*}
    N:=\lceil128\varepsilon^{-2}\rceil,
    \qquad \eta:=\frac1{32(L+1)N}.
  \end{align*}
  All sampler calls use independent randomness.
  For comparison, let $M^\circ_{i,t}\sim\mu_i$ be independent for all $0\leq i\leq L$ and $1\leq t\leq N$,
  and let $X_t^\circ$ be the corresponding exact whole-product estimates defined in \cref{subsec:phase-two}.
  Writing $\overline X^\circ:=N^{-1}\sum_t X_t^\circ$, we have
  \begin{align*}
    \E{X_t^\circ}&=\per{A},\\
    \frac{\E{(X_t^\circ)^2}}{(\per{A})^2}
    &\leq\frac{64}{27p^*}\leq\frac{128}{27}.
  \end{align*}
  Here the indicator that $M^\circ_{L,t}\in\mathcal P^*$ is independent of the product over the earlier temperatures.
  Thus $\Var{\overline X^\circ}/(\per{A})^2\leq4/N$, and Chebyshev's inequality gives
  \begin{align*}
    \Pr{\left|\overline X^\circ-\per{A}\right|
      >\varepsilon\per{A}}
    \leq\frac4{N\varepsilon^2}\leq\frac1{32}.
  \end{align*}

  By \cref{lem:fast-burnin-general}, each actual sample can be coupled to its exact counterpart with disagreement probability at most $\eta$.
  Taking these couplings independently, the probability that any of the $(L+1)N$ pairs disagree is at most $(L+1)N\eta=1/32$.
  When all pairs agree, the actual estimate equals $\overline X^\circ$.
  Consequently, conditional on successful learning, counting fails with probability at most $1/16$.
  Including learning failure gives total failure probability at most $1/8$,
  which proves the required success probability of at least $3/4$.

  On successful learning, the prescribed schedules and limits on the number of resamplings bound the cost of every sampler call.
  Since $L=O(n\log^2 n)$ and $\log(n/\eta)=O(\log(n/\varepsilon))$, the sampling time is
  \begin{align*}
    O\left((L+1)Nn^2\log^4(n/\eta)\right)
    =O\left(n^3\varepsilon^{-2}\log^2 n\,
    \log^4(n/\varepsilon)\right)
    =O\left(n^3\varepsilon^{-2}\log^6(n/\varepsilon)\right).
  \end{align*}
  Forming all products takes only $O(N(L+1)n)$ additional operations.
  The learning and preprocessing costs are absorbed in the same bound.
  To enforce this running time on every execution, return zero if the total number of operations exceeds $C n^3\varepsilon^{-2}\log^6(n/\varepsilon)$.
  For sufficiently large $C$, this limit is never reached when learning succeeds,
  so it adds no failure probability beyond the learning error already counted.
\end{proof}

\subsection{Extension to nonnegative matrices}
\label{sub:weighted-inputs}

\begin{proof}[Proof of \cref{cor:weighted-main}]
  The case $n=1$ is immediate.
  For $n\geq2$, we extend to nonnegative matrices as in~\cite{JSV04,BSVV08},
  implementing the preprocessing of~\cite[Theorem~3.1]{LSW98} in $O(n^3)$ operations.

  Return zero if the support of $A$ has no perfect matching.
  Otherwise, find a maximum-product perfect matching in $O(n^3)$ arithmetic operations,
  as in the proof of~\cite[Theorem~3.1]{LSW98}, and permute it to the diagonal, retaining the notation $A$.

  Find positive column factors $y_i$ satisfying, for all $1\leq i,j\leq n$,
  \begin{align*}
    a_{ii}y_i\geq a_{ij}y_j.
  \end{align*}
  Such factors exist by~\cite[Theorem~3.1]{LSW98}.
  In logarithmic coordinates, this gives
  \begin{align*}
    z_i-z_j\geq\log(a_{ij}/a_{ii})
    \qquad(a_{ij}>0).
  \end{align*}
  Solve this feasible difference-constraint system on $n$ vertices and at most $n^2$ arcs by Bellman--Ford,
  in $O(n^3)$ operations.
  Implement both this computation and the assignment multiplicatively by replacing additions and subtractions of logarithmic quantities with multiplications and divisions,
  leaving comparisons unchanged.
  Thus no logarithms need be evaluated in the preprocessing.

  Set
  \begin{align*}
    s_i:=\sum_j a_{ij}y_j,
    \qquad
    a'_{ij}:=\frac{a_{ij}y_j}{s_i},
    \qquad
    c_{\mathrm{sc}}:=\frac{\prod_jy_j}{\prod_is_i}.
  \end{align*}
  As in~\cite{LSW98}, the resulting matrix satisfies
  \begin{align*}
    \sum_j a'_{ij}=1,
    \qquad
    a'_{ii}=\max_j a'_{ij}\geq1/n,
    \qquad
    \per{A'}=c_{\mathrm{sc}}\per{A}.
  \end{align*}
  In particular, the diagonal perfect matching has weight at least $n^{-n}$.
  The total preprocessing cost is $O(n^3)$.

  Put $\tau:=n^{-n}/(2n!)$ and regularize $A'$ to $\widetilde A$ by setting $\widetilde a_{uv}:=\max\{a'_{uv},\tau\}$.
  For $0\leq\beta\leq1$, use activities $\lambda_\beta(u,v):=\widetilde a_{uv}^{\beta}$ and energy $H(M):=-\sum_{(u,v)\in M}\log\widetilde a_{uv}$.
  Since $\tau\leq\widetilde a_{uv}\leq1$, take $H_{\max}=n\log(1/\tau)=O(n^2\log n)$.
  For every $S=\mathcal P$ or $\mathcal N(u,v)$, diagonal edges and at most one additional edge give $Z_S(1)\geq\tau n^{-n}$.
  Together with $Z_S(0)\leq n!$, this gives $z_S(0)-z_S(1)\leq2\log(1/\tau)$.
  Thus \cref{lem:general-cooling-schedule}, with endpoint $1$ and $\Delta=2\log(1/\tau)=O(n\log n)$,
  yields $0=\beta_0<\cdots<\beta_L=1$ with $L=O(n\log^2 n)$;
  set $\lambda_i(u,v):=\widetilde a_{uv}^{\beta_i}$.

  Apply the samplers in \cref{lem:fast-burnin,lem:fast-burnin-general} and the FAR learning procedure along this schedule.
  FAR batches start from the perfect matchings returned by auxiliary samplers initialized at $\sigma^\ast$.
  In Phase~2, replace the indicator in \cref{subsec:phase-two} by $g(M):=\prod_{(u,v)\in M}\frac{a'_{uv}}{\widetilde a_{uv}}$.
  Regularization increases the permanent by at most $n!\tau\leq\per{A'}/2$,
  so $0\leq g\leq1$ and $\E[\mu_L]{g}=\per{A'}/\per{\widetilde A}\geq2/3$.
  The relative second moment of $g$ is therefore at most $3/2$.
  Using the same confidence and sampling parameters as in \cref{subsec:improved-main-proof} and dividing the final estimate by $c_{\mathrm{sc}}$ gives relative error at most $\varepsilon$ with probability at least $3/4$,
  within the running time in \cref{cor:weighted-main}.
\end{proof}

\subsection{Finite-precision implementation and bit complexity}
\label{sub:bit-complexity}

We implement the algorithm of \cref{subsec:improved-main-proof} for $0/1$ inputs,
with accuracy $\varepsilon/2$ and the schedules specified below.
Its proof bounds the failure probability by $1/8$, leaving room for numerical errors.
The case $n=1$ is immediate.
For $n\geq2$, return zero if no perfect matching exists;
otherwise, fix one.
Write $\ell:=\log(n/\varepsilon)$ and let $D=O(n^3\varepsilon^{-2}\ell^6)$ be the bound on the total number of operations,
returning zero if this bound is exceeded.

\paragraph{Schedules with bounded encoding length.}
Put $\beta_{\max}:=\lceil\log_2(n!)\rceil$ and $h:=2^{-\lceil\log_2(256n\beta_{\max})\rceil}$.
Starting at $\beta_0=0$, use
\begin{align*}
  f(\beta)&:=\frac18\max\left\{\frac1n,
  \frac{\beta}{2\beta_{\max}}\right\},\\
  d_i&:=\min\left\{\beta_{\max}-\beta_i,
  h\left\lfloor\frac{f(\beta_i)}h\right\rfloor\right\},
  \qquad \beta_{i+1}:=\beta_i+d_i,
  \qquad \lambda_i:=e^{-\beta_i}.
\end{align*}
Thus the exponents $\beta_i$ and their increments are stored exactly as integer multiples of $h$;
the activities $\lambda_i$ are evaluated in floating-point form.
For every step except a possibly shortened last step,
\begin{align*}
  d_i\geq f(\beta_i)-h\geq\frac{31}{32}f(\beta_i),
  \qquad \frac h{d_i}\leq\frac1{31\beta_{\max}}.
\end{align*}
Since $f$ has Lipschitz constant $1/(16\beta_{\max})$, for $i<L-1$,
\begin{align*}
  d_{i+1}
  \leq d_i+h+\frac{d_i}{16\beta_{\max}}
  \leq\left(1+\frac1{8\beta_{\max}}\right)d_i.
\end{align*}
The steps are within constant factors of those in \cref{sub:cooling-proofs}, so $L=O(n\log^2 n)$.
The same derivative bound gives a decrease of at most $1/8$ in each log partition function per step.
Consequently, the ideal hole weights still change by factors in $[3/4,4/3]$,
and the second-moment argument there bounds the log relative second moment by
\begin{align*}
  \frac18+\frac{2\beta_{\max}}{8\beta_{\max}+1}<\frac38<\frac12.
\end{align*}
In particular, all bounds used in \cref{lem:cooling-schedule} continue to hold.

For the assignment-activity schedule, choose the fixed constants $J_\alpha$ and $K\geq\max\{A_\alpha,J_\alpha\}$ to be integers,
and use the power of two $q(a)$ from \cref{sub:burnin-schedule-proofs}.
Put $q_0:=q(1/n)$ and $h_a:=1/(8nq_0)$, start at $a_0=1/n$, and set
\begin{align*}
  a_{k+1}:=\min\left\{J,
  a_k+h_a\left\lfloor\frac{a_k}{q(a_k)h_a}\right\rfloor\right\}.
\end{align*}
Every nonfinal increment is between $7a_k/(8q(a_k))$ and $a_k/q(a_k)$.
The upper bound preserves the $3$-warm-start proof, and the lower bound preserves the number of stages in each interval $[b,2b)$,
hence $m+1=O(n)$ and $\sum_k(1+\sqrt{a_k})=O(n)$.
Both endpoints belong to the grid $h_a\mathbb Z$.
No square root is needed to determine $q(a)$.
For $a\leq1$ test $q^2\geq16K(1+n)$;
for $a>1$ test
\begin{align*}
  q^2>16K,
  \qquad a(q^2-16K)^2\geq(16Kn)^2.
\end{align*}
These are exact rational comparisons.
All schedule numerators and denominators have $O(\log n)$ bits.
Computing $n!$ and its binary length to determine $\beta_{\max}$ costs $O(n^2\log^2 n)$ bit operations,
which is absorbed in the final bound.

\paragraph{Precision and discrete choices.}
Use floating-point arithmetic with $b=\lceil C\ell\rceil$ significant bits.
At every temperature, a perfect matching using only input edges has weight one,
and every near-perfect class contains a matching with at most one nonedge,
as shown in \cref{sub:cooling-proofs}.
Since $\lambda_i\geq e^{-\beta_{\max}}$, the ideal hole weights satisfy
\begin{align*}
  \frac1{(n-1)!}\leq w_i^*(u,v)\leq n!e^{\beta_{\max}},
  \qquad
  2^{-\beta_{\max}}\leq w_i^*(u,v)\leq2^{3\beta_{\max}}.
\end{align*}
After each completed weight update, return zero if any stored weight lies outside
\begin{align*}
  \left[2^{-\beta_{\max}-2},\,2^{3\beta_{\max}+2}\right].
\end{align*}
This test never rejects $2$-rough weights.
For stored weights in this range, the exact rescaled entries satisfy
\begin{align*}
  \frac{2^{-6\beta_{\max}-4}}n\leq B(u,v)\leq1.
\end{align*}
Thus $|\log x|=O(n\log n)$ for every stored hole weight and rescaled entry $x$,
and their floating-point exponents require $O(\log n)$ bits.

At each temperature, compare with the exact sampling process using the stored weights and the same rounded weight-update routine.
Relative rounding error at most $1/100$ preserves the $\sqrt2$-roughness of an update whose unrounded relative factor lies in $[9/11,11/9]$.
Thus the accuracy after each successful update depends on its current frequency estimates,
and errors from previous updates do not accumulate.
Use segment trees that store sums of positive weights, avoiding subtraction when excluding an entry.
For sufficiently large $C$, each discrete choice uses $O(b)$ random bits and has coupling error at most $1/(128D)$.

\paragraph{Approximating the clocks.}
Use integer stage durations, rounding the required durations upward and increasing the resampling limits accordingly.
The total resampling rate is the exact integer $W(\sigma)=n+J|C(\sigma)|\leq W_{\max}=O(n^3)$.
For a stage with limit $R$, approximate its first $R+1$ waiting times to absolute error at most a dyadic number
\begin{align*}
  \Delta\leq\frac1{256D W_{\max}(R+1)}.
\end{align*}
Store these approximations on a fixed binary grid and add them exactly.
Their accumulated error is at most $(R+1)\Delta$.
On successful coupling of the discrete choices, the two outputs agree unless a resampling event in the exact comparison process occurs within this distance of the end of the stage.
The probability of such an event is at most $2W_{\max}(R+1)\Delta$.
Generate waiting times as $-\log U/W(\sigma)$.
Except with probability $1/(128D)$ per stage, all $R+1$ uniform random numbers used for these waiting times exceed $1/[128D(R+1)]$.
On this event, $O(b)$ working bits suffice for the required absolute accuracy and the exact addition of the rounded waiting times.
Summing over at most $D$ stages and $D$ discrete choices bounds the additional failure probability by $1/32$.

\paragraph{Output accuracy and cost.}
Evaluate each product observation as
\begin{align*}
  X_t=n!\,\one_{\{M_{L,t}\in\mathcal P^*\}}
  \exp\left(-\sum_{i=0}^{L-1}d_iH(M_{i,t})\right),
  \qquad H(M):=|M\setminus E|.
\end{align*}
The sum in the exponent is an integer multiple of $h$ between zero and $n\beta_{\max}$.
Its numerator is at most $n\beta_{\max}/h=O(n^2\beta_{\max}^2)$,
so it can be accumulated exactly using $O(\log n)$ bits.
Store the observations and their sum in floating-point form, representing zeros exactly.
Increasing $C$ bounds the final relative arithmetic error by $\varepsilon/4$.
Combined with the relative error $\varepsilon/2$ of the exact algorithm,
this gives relative error at most $\varepsilon$.
The failure probabilities total at most $1/8+1/32<1/4$.

Apart from the factorial computation above, arithmetic uses $O(b)$-bit words and costs $O(b^2)$ bit operations per operation.
Standard logarithm and exponential algorithms cost $O(\mathsf M(b)\log b)$ after argument reduction;
Karatsuba multiplication gives $\mathsf M(b)=O(b^{\log_2 3})$, so this is also $O(b^2)$.
Random bits, counters, and indices fit the same bound.
The segment-tree operations are already counted in $D$.
Consequently, including preprocessing and output, the total bit complexity is
\begin{align*}
  O(Db^2)
  =O\left(n^3\varepsilon^{-2}\log^8(n/\varepsilon)\right).
\end{align*}
The running-time bound in \cref{thm:01-main} counts arithmetic operations and comparisons at unit cost;
the bound above counts bit operations.

\section*{Acknowledgements}
Heng Guo would like to thank Weiming Feng and Yucheng Fu for suggesting transferring some techniques in the two-terminal reliability setting to the permanent setting.
Eric Vigoda would like to thank Daniel {\v{S}}tefankovi\v{c} for helpful conversations related to the permanent algorithm.

The main ideas in this paper resulted from many discussions with large language models (LLMs).
The final text and proofs were written independently by the human authors,
with LLMs used for editorial suggestions.
The authors independently verified all mathematical claims and take full responsibility for the final content.

\bibliographystyle{alpha}
\setlength{\bibsep}{7pt plus 1pt minus 1pt}
\bibliography{refs}
\addcontentsline{toc}{section}{References}

\end{document}